\UseRawInputEncoding
\documentclass[11pt]{article}

\usepackage{amssymb,amsmath,amsthm,amsfonts,dsfont}
\usepackage[numbers,comma,sort&compress]{natbib}
\usepackage[colorlinks,hypertexnames=false]{hyperref}
\usepackage[all]{hypcap}
\usepackage{url}
\usepackage{graphicx}
\usepackage[font=small]{caption}
\usepackage[labelformat=simple]{subcaption}
\usepackage{float}
\usepackage{footnote}
\usepackage[margin=1in]{geometry}

\usepackage{xcolor}
\usepackage{pgfplots}

\usepackage{mathtools}
\mathtoolsset{centercolon}

\DeclarePairedDelimiterX\braket[2]{\langle}{\rangle}{#1 \delimsize\vert #2}

\let\originalleft\left
\let\originalright\right
\renewcommand{\left}{\mathopen{}\mathclose\bgroup\originalleft}
\renewcommand{\right}{\aftergroup\egroup\originalright}

\newcommand{\floor}[1]{\lfloor{#1}\rfloor}
\newcommand{\ceil}[1]{\lceil{#1}\rceil}

\newtheorem{theorem}{Theorem}

\newtheorem{lemma}[theorem]{Lemma}
\newtheorem{proposition}[theorem]{Proposition}

\newtheorem{corollary}[theorem]{Corollary}

\newcommand{\eq}[1]{(\ref{eq:#1})}
\newcommand{\thm}[1]{\hyperref[thm:#1]{Theorem~\ref*{thm:#1}}}
\newcommand{\defn}[1]{\hyperref[defn:#1]{Definition~\ref*{defn:#1}}}
\newcommand{\lem}[1]{\hyperref[lem:#1]{Lemma~\ref*{lem:#1}}}
\newcommand{\prop}[1]{\hyperref[prop:#1]{Proposition~\ref*{prop:#1}}}
\newcommand{\fig}[1]{\hyperref[fig:#1]{Figure~\ref*{fig:#1}}}
\newcommand{\tab}[1]{\hyperref[tab:#1]{Table~\ref*{tab:#1}}}
\renewcommand{\sec}[1]{\hyperref[sec:#1]{Section~\ref*{sec:#1}}}
\newcommand{\append}[1]{\hyperref[append:#1]{Appendix~\ref*{append:#1}}}
\newcommand{\cor}[1]{\hyperref[cor:#1]{Corollary~\ref*{cor:#1}}}
\newcommand{\obs}[1]{\hyperref[obs:#1]{Observation~\ref*{obs:#1}}}

\newcommand{\R}{\mathbb{R}}

\newcommand{\comment}[1]{}

\newcommand{\norm}[1]{\left\lVert#1\right\rVert}

\newcommand{\Th}[1]{\Theta\left(#1\right)}
\newcommand{\OO}[1]{O\left(#1\right)}
\newcommand{\dist}{\textrm{dist}}
\usepackage[capitalise,compress]{cleveref}
\newcommand{\commm}[1]{\left[#1\right]}

\newcommand{\scA}{\mathscr{A}}

\newcommand{\supp}{\mathcal{S}}
\newcommand{\distance}{x}
\newcommand{\Norm}[1]{\Vert #1 \Vert}
\DeclareMathOperator{\cc}{cc}

\pgfplotsset{
	log x ticks with fixed point/.style={
		xticklabel={
			\pgfkeys{/pgf/fpu=true}
			\pgfmathparse{exp(\tick)}%
			\pgfmathprintnumber[fixed relative, precision=3]{\pgfmathresult}
			\pgfkeys{/pgf/fpu=false}
		}
	},
	log y ticks with fixed point/.style={
		yticklabel={
			\pgfkeys{/pgf/fpu=true}
			\pgfmathparse{exp(\tick)}%
			\pgfmathprintnumber[fixed relative, precision=3]{\pgfmathresult}
			\pgfkeys{/pgf/fpu=false}
		}
	}
}

\newenvironment{proof-sketch}{%
	\proof}{\endproof}

\usepackage{thmtools}
\usepackage{thm-restate}

\usepackage{ mathrsfs }

\newcommand{\ad}{\mathrm{ad}}

\newcommand{\expT}{\exp_{\mathcal{T}}}
\newcommand{\acomm}{\alpha_{\mathrm{comm}}}
\newcommand{\acommtilde}{\widetilde{\alpha}_{\mathrm{comm}}}

\usepackage{multirow}
\usepackage{enumerate}
\usepackage{booktabs}
\newcommand{\abs}[1]{\left\lvert#1\right\rvert}

\definecolor{bndclr}{RGB}{0, 196, 0}
\definecolor{empclr}{RGB}{125.9700, 46.9200, 141.7800}
\definecolor{looseclr}{RGB}{236.8950, 176.9700, 31.8750}

\newtheorem{innercustomthm}{Theorem}
\newenvironment{customthm}[1]
{\renewcommand\theinnercustomthm{#1}\innercustomthm}
{\endinnercustomthm}

\newcommand\blfootnote[1]{%
	\begingroup
	\renewcommand\thefootnote{}\footnote{#1}%
	\addtocounter{footnote}{-1}%
	\endgroup
}

\newcommand{\vertiii}[1]{{\left\vert\kern-0.25ex\left\vert\kern-0.25ex\left\vert #1
		\right\vert\kern-0.25ex\right\vert\kern-0.25ex\right\vert}}
\newcommand{\Vertiii}[1]{{\vert\kern-0.25ex\vert\kern-0.25ex\vert #1
		\vert\kern-0.25ex\vert\kern-0.25ex\vert}}

\newcommand{\cO}[1]{\mathcal{O}\left(#1\right)}
\newcommand{\tildecO}[1]{\widetilde{\mathcal{O}}\left(#1\right)}
\newcommand{\Om}[1]{\Omega\left(#1\right)}

\newcommand{\Htrunc}{H_{\text{trunc}}}
\newcommand{\Hlc}{H_{\text{lc}}}
\newcommand{\Strunc}{\mathscr{S}_{\text{trunc}}}
\newcommand{\Sreduce}{\mathscr{S}_{\text{reduce}}}
\newcommand{\Slc}{\mathscr{S}_{\text{lc}}}

\let\counterwithin\relax
\usepackage{chngcntr}
\usepackage{apptools}
\AtAppendix{\counterwithin{theorem}{section}}

\begin{document}
\title{\huge A Theory of Trotter Error}
\author
{Andrew M.\ Childs,$^{1,2,3}$ Yuan Su,$^{1,2,3,4}$ Minh C.\ Tran,$^{3,5}$\\
	Nathan Wiebe,$^{6,7,8}$ and Shuchen Zhu$^{9}$\\
}
\date{\vspace{-10mm}}
\maketitle
\blfootnote{This is a slightly enhanced version of the article entitled \emph{Theory of Trotter Error with Commutator Scaling} published in Physical Review X
	\textbf{11} (2021), 011020 \cite{CSTWZ21}.}
\blfootnote{$^{1}$Department of Computer Science, University of Maryland.}
\blfootnote{$^{2}$Institute for Advanced Computer Studies, University of Maryland.}
\blfootnote{$^{3}$Joint Center for Quantum Information and Computer Science, University of Maryland.}
\blfootnote{$^{4}$Institute for Quantum Information and Matter, California Institute of Technology.}
\blfootnote{$^{5}$Joint Quantum Institute, University of Maryland.}
\blfootnote{$^{6}$Department of Physics, University of Washington.}
\blfootnote{$^{7}$Pacific Northwest National Laboratory.}
\blfootnote{$^{8}$Google Inc., Venice CA.}
\blfootnote{$^{9}$Department of Computer Science, Georgetown University.}

\begin{abstract}
The Lie-Trotter formula, together with its higher-order generalizations, provides a simple approach to decomposing the exponential of a sum of operators. Despite significant effort, the error scaling of such product formulas remains poorly understood.

We develop a theory of Trotter error that overcomes the limitations of prior approaches based on truncating the Baker-Campbell-Hausdorff expansion. Our analysis directly exploits the commutativity of operator summands, producing tighter error bounds for both real- and imaginary-time evolutions. Whereas previous work achieves similar goals for systems with geometric locality or Lie-algebraic structure, our approach holds in general.

We give a host of improved algorithms for digital quantum simulation and quantum Monte Carlo methods, nearly matching or even outperforming the best previous results. Our applications include: (i) a simulation of second-quantized plane-wave electronic structure, nearly matching the interaction-picture algorithm of Low and Wiebe; (ii) a simulation of $k$-local Hamiltonians almost with induced $1$-norm scaling, faster than
the qubitization algorithm of Low and Chuang; (iii) a simulation of rapidly decaying power-law interactions, outperforming the Lieb-Robinson-based approach of Tran et al.; (iv) a hybrid simulation of clustered Hamiltonians, dramatically improving the result of Peng, Harrow, Ozols, and Wu; and (v) quantum Monte Carlo simulations of the transverse field Ising model and quantum ferromagnets, tightening previous analyses of Bravyi and Gosset.

We obtain further speedups using the fact that product formulas can preserve the locality of the simulated system. Specifically, we show that local observables can be simulated with complexity independent of the system size for power-law interacting systems, which implies a Lieb-Robinson bound nearly matching a recent result of Tran et al.

Our analysis reproduces known tight bounds for first- and second-order formulas. We further investigate the tightness of our bounds for higher-order formulas. For quantum simulation of a one-dimensional Heisenberg model with an even-odd ordering of terms, our result overestimates the complexity by only a factor of $5$. Our bound is also close to tight for power-law interactions and other orderings of terms. This suggests that our theory can accurately characterize Trotter error in terms of both the asymptotic scaling and the constant prefactor.
\end{abstract}

\newpage
{
	\thispagestyle{empty}
	\clearpage\tableofcontents
	\thispagestyle{empty}
}
\newpage

\section{Introduction}
\label{sec:intro}

Product formulas provide a convenient approach to decomposing the evolution of a sum of operators. The Lie product formula was introduced in the study of Lie groups in the late 1800s; later developments considered more general operators and higher-order approximations. Originally studied in the context of pure mathematics, product formulas have found numerous applications in other areas, such as applied mathematics (under the name ``splitting method'' or ``symplectic integrators''), physics (under the name ``Trotterization''), and theoretical computer science.

This paper considers the application of product formulas to simulating quantum systems. It has been known for over two decades that these formulas are useful for digital quantum simulation and quantum Monte Carlo methods. However, their error scaling is poorly understood and existing bounds can be several orders of magnitude larger than what are observed in practice, even for simulating relatively small systems.

We develop a theory of Trotter error that directly exploits the commutativity of operator summands to give tighter bounds. Whereas previous work achieves similar goals for systems with geometric locality or Lie-algebraic structure, our theory has no such restrictions. We present a host of examples in which product formulas can nearly match or even outperform state-of-the-art simulation results. We accompany our analysis with numerical calculation, which suggests that the bounds also have nearly tight constant prefactors.

We hope this work will motivate further studies of the product-formula approach, which has been deemphasized in recent years in favor of more advanced simulation algorithms that are easier to analyze but harder to implement. Indeed, despite the sophistication of these ``post-Trotter methods'' and their optimality in certain general models, our work shows that they can be provably outperformed by product formulas for simulating many quantum systems.

\subsection{Simulating quantum systems by product formulas}
\label{sec:sim_pf}
Simulating the dynamics of quantum systems is one of the most promising applications of digital quantum computers.
Classical computers apparently require exponential time to simulate typical quantum dynamics.
This intractability led Feynman \cite{Fey82} and others to propose the idea of quantum computers. In 1996, Lloyd gave the first explicit quantum algorithm for simulating $k$-local Hamiltonians \cite{Llo96}. Subsequent work considered the broader class of sparse Hamiltonians \cite{AT03,BACS05,FractionalQuery14,BCK15,LC17,Low18} and developed techniques for simulating particular physical systems \cite{WBCHT14,Pou15,BWMMNC18,MEABY18,JLP12,lanyon2010towards,Cao19}, with potential applications to developing new pharmaceuticals, catalysts, and materials. The study of quantum simulation has also inspired the design of various quantum algorithms for other problems \cite{HHL09,BS17,FGG07,CCDFGS03,Berry14}.

Lloyd's approach to quantum simulation is based on product formulas. Specifically, let $H=\sum_{\gamma=1}^{\Gamma}H_\gamma$ be a $k$-local Hamiltonian (i.e., each $H_\gamma$ acts nontrivially on $k=\cO{1}$ qubits). Assuming $H$ is time independent, evolution under $H$ for time $t$ is described by the unitary operation $e^{-itH}$. When $t$ is small, this evolution can be well approximated by the Lie-Trotter formula $\mathscr{S}_1(t)=e^{-itH_\Gamma}\cdots e^{-itH_1}$, where each $e^{-itH_\gamma}$ can be efficiently implemented on a quantum computer. To simulate for a longer time, we may divide the evolution into $r$ \emph{Trotter steps} and simulate each step with \emph{Trotter error} at most $\epsilon/r$. We choose the \emph{Trotter number} $r$ to be sufficiently large so that the entire simulation achieves an error of at most $\epsilon$.
The Lie-Trotter formula only provides a first-order approximation to the evolution, but higher-order approximations are also known from the work of Suzuki and others \cite{Suz91,bk:BC16}. While many previous works focused on the performance of specific formulas, the theory we develop holds for any formula; we use the term \emph{product formula} to emphasize this generality. A quantum simulation algorithm using product formulas does not require ancilla qubits, making this approach advantageous for near-term experimental demonstration.

Recent studies have provided alternative simulation algorithms beyond the product-formula approach (sometimes called ``post-Trotter methods''). Some of these algorithms have logarithmic dependence on the allowed error \cite{FractionalQuery14,BCCKS14,BCK15,LC17,LC16,LW18}, an exponential improvement over product formulas. However, this does not generally lead to an exponential reduction in time complexity for practical applications of quantum simulation. In practice, the simulation accuracy is often chosen to be constant. Then the error dependence only enters as a constant prefactor, which may not significantly affect the overall gate complexity. The reduction in complexity is more significant when quantum simulation is used as a subroutine in another quantum algorithm (such as phase estimation), since this may require high-precision simulation to ensure reliable behavior. However, this logarithmic error dependence typically replaces a factor that scales polynomially with time or the system size by another that scales logarithmically, giving only a polynomial reduction in the complexity. Furthermore, the constant-factor overhead and extra space requirements of post-Trotter methods may make them uncompetitive with the product-formula approach in practice.

Product formulas and their generalizations \cite{HP18,COS18,LKW19,OWC19} can perform significantly better when the operator summands commute or nearly commute---a unique feature that does not seem to hold for other quantum simulation algorithms \cite{FractionalQuery14,BCCKS14,BCK15,LC17,LC16,LW18,Campbell18}.
This effect has been observed numerically in previous studies of quantum simulations of condensed matter systems \cite{CMNRS18} and quantum chemistry \cite{RWSWT17,BMWAW15,wecker2015solving}.
An intuitive explanation of this phenomenon comes from truncating the Baker-Campbell-Hausdorff (BCH) expansion. However, the intuition that the lowest-order terms of the BCH expansion are dominant is surprisingly difficult to justify (and sometimes is not even valid \cite{CS19,WBCHT14}). Thus, previous work established loose Trotter error bounds, sometimes suggesting poor performance. Our results rigorously demonstrate that for many systems, such arguments do not accurately reflect the true performance of product formulas.

Product-formula decompositions directly translate terms of the Hamiltonian into elementary simulation steps, making them well suited to preserve certain properties such as the locality
of the simulated system. We show that this property can be used to further reduce the simulation cost when the goal is to simulate local observables as opposed to the full dynamics \cite{TPP19,KGE14}.

Besides digital quantum simulation, product formulas can also be applied to quantum Monte Carlo methods, in which the goal is to classically compute certain properties of the Hamiltonian, such as the partition function, the free energy, or the ground energy. Our results can also be applied to improve the efficiency of previous applications of quantum Monte Carlo methods for systems such as the transverse field Ising model \cite{Bravyi15} and quantum ferromagnets \cite{BG17}.

\subsection{Previous analyses of Trotter error}
\label{sec:pre_analyses}
We now briefly summarize prior approaches to analyzing Trotter error for simulating quantum systems, and we discuss their limitations.

The original work of Lloyd \cite{Llo96} analyzes product formulas by truncating the Taylor expansion (or the BCH expansion). Recall that the Lie-Trotter formula $\mathscr{S}_1(t)$ provides a first-order approximation to the evolution, so $\mathscr{S}_1(t)=e^{-itH}+\OO{t^2}$. To simplify the analysis, Lloyd dropped all higher-order terms in the Taylor expansion and focused only on the terms of lowest order $t^2$. This approach is intuitive and has been employed by subsequent works to give rough estimation of Trotter error. The drawback of this analysis is that it implicitly assumes that high-order terms are dominated by the lowest-order term. However, this does not necessarily hold for many systems such as nearest-neighbor lattice Hamiltonians \cite{CS19} and chemical Hamiltonians \cite{WBCHT14} when the time step $t$ is fixed.

This issue was addressed in the seminal work of Berry, Ahokas, Cleve, and Sanders by using a tail bound of the Taylor expansion \cite{BACS05}, giving a concrete bound on the Trotter error for high-order Suzuki formulas. For a Hamiltonian $H=\sum_{\gamma=1}^{\Gamma}H_\gamma$ containing $\Gamma$ summands, their bound scales
with $\Gamma\max_\gamma\norm{H_\gamma}$, although it is not hard to improve this \cite{HP18} to
$\sum_{\gamma=1}^{\Gamma}\norm{H_\gamma}$ \cite{Suzuki85,LKW19}. Regardless of which scaling to use, this worst-case analysis does not exploit the commutativity of Hamiltonian summands and the resulting complexity is worse than many post-Trotter methods.

Error bounds that exploit the commutativity of summands are known for low-order formulas, such as the Lie-Trotter formula \cite{Huyghebaert_1990,Suzuki85} and the second-order Suzuki formula \cite{Suzuki85,DT10,WBCHT14,Kivlichan19}. These bounds are tight in the sense that they match the lowest-order term of the BCH expansion up to an application of the triangle inequality. However, it is unclear whether they can be generalized, say, to the fourth- or the sixth-order case, which are still reasonably simple and can provide a significant advantage in practice \cite{CMNRS18}.

Instead, previous works made compromises to obtain improved analyses of higher-order formulas. Somma gave an improved bound by representing the Trotter error as an infinite series of nested commutators \cite{Somma16}. This approach is advantageous when the simulated system has an underlying Lie-algebraic structure with small structure factors, such as for a quantum harmonic oscillator and certain nonquadratic potentials. However, this reduces to the worst-case analysis of Berry, Ahokas, Cleve, and Sanders for other systems. 

An alternative approach of Thalhammer represented the error of a $p$th-order product formula using commutators of order up to $q$ for $q\geq p$ \cite{Thalhammer08}, with the ($q+1$)st-order remainder further bounded by some tail bound. This analysis is bottlenecked by the use of the tail bound. The special case where $q=p+1$ was studied in~\cite{CMNRS18} and the result was applied to estimate the quantum resource for simulating a one-dimensional Heisenberg model, which only offers a modest improvement over the worst-case analysis.

In recent work~\cite{CS19}, Childs and Su gave a Trotter error bound in which only the lowest-order error appears, avoiding manipulation of infinite series or use of tail bounds. As an immediate application, they showed that product formulas can nearly optimally simulate lattice systems with geometrically local interactions, justifying an earlier claim of Jordan, Lee, and Preskill~\cite{JLP12} in the context of simulating quantum field theory. Their improvement is based on a representation of Trotter error introduced by Descombes and Thalhammer~\cite{DT10}, which streamlines the previous analysis \cite{Thalhammer08}. In this approach, the Trotter error is represented using commutators nested with conjugations of matrix exponentials. For Hamiltonians with nearest-neighbor interactions, Ref.\ \cite{CS19} gave an argument based on locality to cancel the majority of the Trotter error. However, this approach reduces to the worst-case scenario for systems lacking geometric locality. In contrast, our representation of Trotter error does not have this restriction and results in speedups for simulating various strongly long-range interacting systems (see \tab{result_summary}).

For other related studies of Trotter error in the context of numerical analysis, we refer the reader to \cite{Jahnke2000,Thalhammer08,Thalhammer12,mclachlan_quispel_2002,McLachlan95,hairer2006geometric} and the references therein.

\subsection{Trotter error with commutator scaling}
\label{sec:commutator_scaling}
We give a new bound on the Trotter error that depends on nested commutators of the operator summands. This bound is formally stated in \sec{theory_rep} and previewed here.

\begin{customthm}{}[Trotter error with commutator scaling]
	\label{thm:main_result}
	Let $H=\sum_{\gamma=1}^{\Gamma}H_\gamma$ be an operator consisting of $\Gamma$ summands and let $t\geq 0$. Let $\mathscr{S}(t)$ be a $p$th-order $\Upsilon$-stage product formula as in \sec{prelim_pf}. Define
	$\acommtilde=\sum_{\gamma_1,\gamma_2,\ldots,\gamma_{p+1}=1}^{\Gamma}\norm{\big[H_{\gamma_{p+1}},\cdots\big[H_{\gamma_2},H_{\gamma_1}\big]\big]}$, where $\norm{\cdot}$ is the spectral norm. Then the additive error $\mathscr{A}(t)$ and the multiplicative error $\mathscr{M}(t)$, defined respectively by $\mathscr{S}(t)=e^{tH}+\mathscr{A}(t)$ and $\mathscr{S}(t)=e^{tH}(I+\mathscr{M}(t))$, can be asymptotically bounded as
	\begin{equation}
	\norm{\mathscr{A}(t)}=\cO{\acommtilde t^{p+1}e^{2t \Upsilon\sum_{\gamma=1}^{\Gamma}\norm{H_{\gamma}}}},\qquad
	\norm{\mathscr{M}(t)}=\cO{\acommtilde t^{p+1}e^{2t \Upsilon\sum_{\gamma=1}^{\Gamma}\norm{H_{\gamma}}}}.
	\end{equation}
	Furthermore, if the $H_\gamma$ are anti-Hermitian, corresponding to physical Hamiltonians, we have
	\begin{equation}
	\norm{\mathscr{A}(t)}=\cO{\acommtilde t^{p+1}},\qquad
	\norm{\mathscr{M}(t)}=\cO{\acommtilde t^{p+1}}.
	\end{equation}
\end{customthm}

We emphasize that this theorem does not follow from truncating the BCH series. Although the ($p+1$)st-order term of the BCH series is also a linear combination of nested commutators similar to $\acommtilde$, such a term can be dominated by a higher-order term when $t$ is fixed, as is the case for nearest-neighbor lattice systems \cite[Supplementary Section I]{CS19} and quantum chemistry \cite[Appendix B]{WBCHT14}. Truncating the BCH series ignores significant, potentially dominant error contributions and thus does not accurately characterize the Trotter error.

The above expression for our asymptotic error bound is succinct and easy to evaluate. In \sec{app}, we compute $\acommtilde$ for various examples, including the second-quantized electronic-structure Hamiltonians, $k$-local Hamiltonians, rapidly decaying power-law interactions, and clustered Hamiltonians. We further study the tightness of the prefactor of our bound in \sec{prefactor} and give a numerical implementation for one-dimensional Heisenberg models with either nearest-neighbor interactions or power-law interactions.

Although the definition of a specific product formula depends on the ordering of operator summands, our asymptotic bound does not. As an immediate consequence, the asymptotic speedups we obtain in \sec{app_dqs} hold irrespective of how we order the operator summands in the simulation. For the special case of nearest-neighbor lattice models, this answers a previous question of \cite{CS19} regarding the ``ordering robustness'' of higher-order formulas. However, the ordering becomes important if our goal is to simulate local observables or to get error bounds with tight constant prefactors, as we further discuss in \sec{app_local_obs} and \sec{prefactor}, respectively.

As mentioned in \sec{pre_analyses}, prior Trotter error analyses typically produce loose bounds and are only effective in special cases. Our approach overcomes those limitations in the following respects:
\begin{enumerate}[(i)]
	\item our bound only contains a finite number of error terms, in contrast to the bound in~\cite{Somma16};
	\item our bound involves pure nested commutators without introducing conjugations of matrix exponentials or tail bounds, overcoming the drawbacks of~\cite{CS19,Thalhammer08,Thalhammer12};
	\item our bound reduces to the worst-case analysis of~\cite{BACS05} by further bounding terms with the triangle inequality; and
	\item for Hamiltonians with two summands, our bound encompasses the tight analyses~\cite{Huyghebaert_1990,Suzuki85,WBCHT14,DT10,Kivlichan19} of the Lie-Trotter formula and the second-order Suzuki formula as special cases.
\end{enumerate}

\subsection{Overview of results}
\label{sec:results}
The commutator scaling of Trotter error uncovers a host of examples where product formulas can nearly match or even outperform the state-of-the-art results in digital quantum simulation. These examples include: (i) a simulation of second-quantized plane-wave electronic structure with $n$ spin orbitals for time $t$ with gate complexity $n^{2+o(1)}t^{1+o(1)}$, whereas the state-of-the-art approach performs simulation in the interaction picture \cite{LW18} with cost $\tildecO{n^2t}$ and likely large overhead; (ii) a simulation of $n$-qubit $k$-local Hamiltonians $H$ with complexity $n^k\vertiii{H}_1\norm{H}_1^{o(1)}t^{1+o(1)}$ that almost scales with the induced $1$-norm\footnote{The $1$-norm $\norm{H}_1$ and the induced $1$-norm $\vertiii{H}_1$ are formally defined in \sec{prelim_notation}. For now, it suffices to know that $\vertiii{H}_1\leq\norm{H}_1$ and that the gap can be significant for many $k$-local Hamiltonians.} $\vertiii{H}_1$, implying an improved simulation of $d$-dimensional power-law interactions that decay with distance $\distance$ as $1/\distance^{\alpha}$ for $\alpha\leq 2d$, whereas the fastest previous approach uses the qubitization algorithm \cite{LC16} with cost $\tildecO{n^k\norm{H}_1t}$; (iii) a simulation of $d$-dimensional power-law interactions $1/\distance^{\alpha}$ (for fixed $\alpha>2d$) with gate complexity $(nt)^{1+d/(\alpha-d)+o(1)}$, whereas the best previous algorithm decomposes the evolution based on Lieb-Robinson bounds \cite{Tran18} with cost $\tildecO{(nt)^{1+2d/(\alpha-d)}}$; and (iv) a hybrid simulation of clustered Hamiltonians of interaction strength $h_B$ and contraction complexity $\cc(g)$ with runtime $2^{\cO{h_B^{o(1)} t^{1+o(1)} \cc(g)/\epsilon^{o(1)}}}$, improving the previous result of $2^{\cO{h_B^2t^2 \cc(g)/\epsilon}}$ \cite{PHOW19}. We discuss these examples in more detail in \sec{app_dqs}.

We show in \sec{app_local_obs} that these gate complexities can be further improved when the goal is to simulate local observables instead of the full dynamics. We illustrate this for $d$-dimensional lattice systems with $1/\distance^\alpha$ interactions ($\alpha>2d$).
Lieb-Robinson bounds for power-law interactions~\cite{Tran18} suggest that the evolution of a local observable is mostly confined inside a light cone induced by the interactions. Simulating such an evolution by simulating the dynamics of the entire system appears redundant, especially when the system size is large.
We realize this intuition and show, without using Lieb-Robinson bounds,
that the gate count for simulating the evolution of a local observable scales as $ t^{\left(1+d\frac{\alpha -d}{\alpha -2 d}\right) \left(1+\frac{d}{\alpha -d}\right)+o(1)}
$, which is independent of the system size $n$ and smaller than simulating the dynamics of the entire system when $n = \Om{t^{d(\alpha-d)/(\alpha-2d)}}$.
The scaling also reduces to $t^{d+1+o(1)}$---proportional to the space-time volume inside a linear light cone---in the limit $\alpha\rightarrow \infty$, which corresponds to nearest-neighbor interactions.

Our bound can also be applied to improve the performance of quantum Monte Carlo simulation. In this case, we are limited to the use of second-order Suzuki formula and, due to imaginary-time evolution, the Trotter number scales at least linearly with the system size. Nevertheless, we are able to improve several existing classical simulations using our bound, without modifying the original algorithms. This includes: (i) a simulation of $n$-qubit transverse field Ising model with maximum interaction strength $j$ and precision $\epsilon$ with runtime $\tildecO{n^{45}j^{14}\epsilon^{-2}+n^{38}j^{21}\epsilon^{-9}}$, tightening the previous result of $\tildecO{n^{59}j^{21}\epsilon^{-9}}$ \cite{Bravyi15}; and (ii) a simulation of ferromagnetic quantum spin systems for (imaginary) time $\beta$ and accuracy $\epsilon$ with runtime $\tildecO{n^{92}(1+\beta^{46})/\epsilon^{25}}$, improving the previous complexity of $\tildecO{n^{115}(1+\beta^{46})/\epsilon^{25}}$ \cite{BG17}. These applications are further discussed in \sec{app_qmc}. \tab{result_summary} compares our results against the best previous ones for simulating quantum dynamics, simulating local observables, and quantum Monte Carlo simulation.

\begin{table}
	\scalebox{0.86}{\parbox{1.\linewidth}{%
	\begin{center}
		\renewcommand{\arraystretch}{1.75}
		\footnotesize
		\begin{tabular}{cccc}
			\toprule
			Application & System & Best previous result & New result\\
			\midrule
			\multirow{6}{*}{\shortstack[c]{Simulating \\quantum dynamics}}
			&Electronic structure &$\tildecO{n^2t}$ (Interaction picture) &$n^{2+o(1)}t^{1+o(1)}$ \\
			&$k$-local Hamiltonians &$\tildecO{n^k\norm{H}_1t}$ (Qubitization) &$n^k\vertiii{H}_1\norm{H}_1^{o(1)}t^{1+o(1)}$ \\
			&$1/\distance^\alpha$ ($\alpha<d$) &$\widetilde{O}(n^{4-\alpha/d}t)$ (Qubitization) &$n^{3-\alpha/d+o(1)}t^{1+o(1)}$ \\
			&$1/\distance^\alpha$ ($d\leq\alpha\leq 2d$) &$\widetilde{O}(n^3t)$ (Qubitization) &$n^{2+o(1)}t^{1+o(1)}$\\
			&$1/\distance^\alpha$ ($\alpha>2d$) &$\tildecO{(nt)^{1+2d/(\alpha-d)}}$ (Lieb-Robinson bound) &$(nt)^{1+d/(\alpha-d)+o(1)}$ \\
			&Clustered Hamiltonians &$2^{\cO{h_B^2t^2 \cc(g)/\epsilon}}$ &$2^{\cO{h_B^{o(1)} t^{1+o(1)} \cc(g)/\epsilon^{o(1)}}}$ \\
			\midrule
			\multirow{1}{*}{\shortstack[c]{Simulating \\local observables}}
			&$1/\distance^{\alpha}(\alpha>2d)$ &--- &$t^{\left(1+d\frac{\alpha -d}{\alpha -2 d}\right) \left(1+\frac{d}{\alpha -d}\right)+o(1)}$ \\
			\midrule
			\multirow{2}{*}{Monte Carlo simulation}
			&Transverse field Ising model &$\tildecO{n^{59}j^{21}\epsilon^{-9}}$ &$\tildecO{n^{45}j^{14}\epsilon^{-2}+n^{38}j^{21}\epsilon^{-9}}$ \\
			&Quantum ferromagnets &$\tildecO{n^{115}(1+\beta^{46})/\epsilon^{25}}$ &$\tildecO{n^{92}(1+\beta^{46})/\epsilon^{25}}$ \\
			\bottomrule
		\end{tabular}
		\renewcommand{\arraystretch}{1}
	\end{center}
	}}
	\caption{Comparison of our results and the best previous results for simulating quantum dynamics, simulating local observables, and quantum Monte Carlo simulation.}
	\label{tab:result_summary}
\end{table}

Given the numerous applications our bound provides in the asymptotic regime, we ask whether it has a favorable constant prefactor as well. This consideration is relevant to the practical performance of product formulas, especially for near-term quantum simulation experiments. For a two-term Hamiltonian, we show that our bound reduces to the known analyses of the Lie-Trotter formula \cite{Huyghebaert_1990,Suzuki85} and the second-order Suzuki formula \cite{DT10,WBCHT14,Kivlichan19}. We then bootstrap the result to analyze Hamiltonians with an arbitrary number of summands (\sec{prefactor_pf12}). The resulting bound matches the lowest-order term of the BCH expansion up to an application of the triangle inequality, and our analysis is thus provably tight for these low-order formulas.

We further numerically implement our bound for a one-dimensional Heisenberg model with a random magnetic field. This model can be simulated to understand condensed matter phenomena, but even a simulation of modest size seems to be infeasible for current classical computers. Childs et al.\ compared different quantum simulation algorithms for this model \cite{CMNRS18} and observed that product formulas have the best empirical performance, although their provable bounds were off by orders of magnitude even for systems of modest size, making it hard to identify with confidence the most efficient approach for near-term simulation. Reference \cite{CS19} claimed an improved fourth-order bound that is off by a factor of about $17$. Here, we give a tight bound that overestimates by only a factor of about $5$. We also give a nearly tight Trotter error bound for power-law interactions. We describe the numerical implementation of our bound in detail in \sec{prefactor_pf2k}.

Underpinning these improvements is a theory we develop concerning the \textit{types}, \textit{order conditions}, and \textit{representations} of Trotter error. We illustrate these concepts in \sec{theory_example} with the simple example of the first-order Lie-Trotter formula.

Let $H=\sum_{\gamma=1}^{\Gamma}H_\gamma$ be a sum of operators and let $\mathscr{S}(t)$ be a product formula corresponding to this decomposition. We say that $\mathscr{A}(t)$, $\mathscr{M}(t)$, and $\mathscr{E}(t)$ are the additive, multiplicative, and exponentiated Trotter error if
\begin{equation}
\mathscr{S}(t)=e^{tH}+\mathscr{A}(t),\qquad
\mathscr{S}(t)=e^{tH}(I+\mathscr{M}(t)),\qquad
\mathscr{S}(t)=\expT\bigg(\int_{0}^{t}\mathrm{d}\tau\big(H+\mathscr{E}(\tau)\big)\bigg),
\end{equation}
respectively, where $\expT$ denotes the time-ordered matrix exponential. For applications in digital quantum simulation, these three types of Trotter error are equivalent to each other. However, the multiplicative type and the exponentiated type are more versatile for analyzing quantum Monte Carlo simulation. We give a constructive definition of these error types and discuss how they are related in \sec{theory_type}.

A $p$th-order product formula $\mathscr{S}(t)$ can approximate the ideal evolution to $p$th order, in the sense that $\mathscr{S}(t)=e^{tH}+\OO{t^{p+1}}$. Motivated by this, we say that an operator-valued function $\mathscr{F}(t)$ satisfies the $p$th-order condition if $\mathscr{F}(t)=\OO{t^{p+1}}$. In \sec{theory_order}, we give order conditions for Trotter error and its various derived operators. One significance of order conditions is that they can be used to cancel low-order terms. In particular, if $\mathscr{F}(t)$ satisfies the $p$th-order condition, then all terms with order at most $p$ vanish in the Taylor series. This can be verified by brute-force differentiation when $\mathscr{F}(t)$ is explicitly given, but applying the correct order condition avoids such a cumbersome calculation.

We then consider representations of Trotter error in \sec{theory_rep}. Our representation only involves finitely many error terms, each of which is given by a nested commutator of operator summands. As mentioned earlier, these features overcome the drawbacks of previous representations and motivate a host of new applications. In deriving our representation, we work in a general setting where operator summands are not necessarily anti-Hermitian, so that our analysis simultaneously handles real-time evolutions for digital quantum simulation and imaginary-time evolutions for quantum Monte Carlo simulation.

\sec{prelim} gives a summary of background material that is necessary for understanding our Trotter error theory and its applications. \sec{discussion} concludes the paper with a brief discussion of the results and some open questions.

\section{Preliminaries}
\label{sec:prelim}

In this section, we summarize the preliminaries that we use in subsequent sections of the paper.
Specifically, we introduce notation and terminology in \sec{prelim_notation}, including various notions of norms and common asymptotic notations.
In \sec{prelim_expT}, we discuss time-ordered evolutions and their properties that are relevant to our analysis.
We then define general product formulas in \sec{prelim_pf} and prove a Trotter error bound with $1$-norm scaling.
Readers who are familiar with these preliminaries may skip ahead to \sec{theory} for the main result of our paper.

\subsection{Notation and terminology}
\label{sec:prelim_notation}
Unless otherwise noted, we use lowercase Latin letters to represent scalars, such as the evolution time $t$, the system size $n$, and the order of a product formula $p$. We also use the Greek alphabet to denote scalars, especially when we want to write a summation like $\sum_{\gamma=1}^{\Gamma}$. We use uppercase Latin letters, such as $A$, to denote operators. Throughout the paper, we assume that the underlying Hilbert space is finite dimensional and operators can be represented by complex square matrices. We expect that some of our analyses can be generalized to spaces with infinite dimensions, but we restrict ourselves to the finite-dimensional setting since this is most relevant for applications to digital quantum simulation and quantum Monte Carlo simulation. We use scripted uppercase letters, such as $\mathscr{F}(t)$, to denote operator-valued functions.

We organize scalars to form vectors $h_\gamma$ and tensors $h_{\gamma_1,\ldots,\gamma_k}$. We use standard norms for tensors, including the $1$-norm $\norm{h}_1:=\sum_{\gamma_1,\ldots,\gamma_k}\abs{h_{\gamma_1,\ldots,\gamma_k}}$, the Euclidean norm (or $2$-norm) $\norm{h}_2:=\sqrt{\sum_{\gamma_1,\ldots,\gamma_k}\abs{h_{\gamma_1,\ldots,\gamma_k}}^2}$, and the $\infty$-norm $\norm{h}_\infty:=\max_{\gamma_1,\ldots,\gamma_k}\abs{h_{\gamma_1,\ldots,\gamma_k}}$. In case there is ambiguity, we use $\vec{h}$ to emphasize the fact that $h$ is a vector (or a tensor more generally).

For an operator $A$, we use $\norm{A}$ to denote its spectral norm---the largest singular value of $A$. The spectral norm is also known as the operator norm. It is a matrix norm that satisfies the scaling property $\norm{aA}=\abs{a}\norm{A}$, the submultiplicative property $\norm{AB}\leq\norm{A}\norm{B}$, and the triangle inequality $\norm{A+B}\leq\norm{A}+\norm{B}$. If $A$ is unitary, then $\norm{A}=1$. We further use $A_{\gamma_1,\ldots,\gamma_k}$ to denote a tensor where each elementary object is an operator. We define a norm of $A_{\gamma_1,\ldots,\gamma_k}$ by taking the spectral norm of each elementary operator and evaluating the corresponding norm of the resulting tensor. For example, we have $\norm{A}_1:=\sum_{\gamma_1,\ldots,\gamma_k}\norm{A_{\gamma_1,\ldots,\gamma_k}}$ and $\norm{A}_\infty:=\max_{\gamma_1,\ldots,\gamma_k}\norm{A_{\gamma_1,\ldots,\gamma_k}}$.

For a tensor $A_{\gamma_1,\ldots,\gamma_k}$, we define
\begin{equation}
	\vertiii{A}_1:=\max_j\max_{\gamma_j}\sum_{\substack{\gamma_1,\ldots,\gamma_{j-1},\\\gamma_{j+1},\ldots,\gamma_{k}}}\norm{A_{\gamma_1,\ldots,\gamma_k}}.
\end{equation}
We call $\vertiii{A}_1$ the induced $1$-norm of $A$, since it can be seen as a generalization of the induced $1$-norm $\max_{\gamma_2}\sum_{\gamma_1}\abs{a_{\gamma_1,\gamma_2}}$ of a matrix $a_{\gamma_1,\gamma_2}$ \cite{horn2012matrix}. A quantum simulation algorithm with induced $1$-norm scaling runs faster than a $1$-norm scaled algorithm because
\begin{equation}
	\vertiii{A}_1\leq\norm{A}_1.
\end{equation}
In fact, as we will see in \sec{app_dqs}, the gap between these two norms can be significant for many realistic systems.

Let $f,g:\R\rightarrow\R$ be functions of real variables. We write $f=\OO{g}$ if there exist $c,t_0>0$ such that $\abs{f(\tau)}\leq c\abs{g(\tau)}$ whenever $\abs{\tau}\leq t_0$. Note that we consider the limit when the variable $\tau$ approaches zero as opposed to infinity, which is different from the usual setting of algorithmic analysis. For that purpose, we write $f=\cO{g}$ if there exist $c,t_1>0$ such that $\abs{f(\tau)}\leq c\abs{g(\tau)}$ for all $\abs{\tau}\geq t_1$. When there is no ambiguity, we will use $f=\cO{g}$ to also represent the case where $\abs{f(\tau)}\leq c\abs{g(\tau)}$ holds for all $\tau\in\R$. We then extend the definition of $\mathcal{O}$ to functions of positive integers and multivariate functions. For example, we use $f(n,t,1/\epsilon)=\cO{(nt)^2/\epsilon}$ to mean that $\abs{f(n,t,1/\epsilon)}\leq c(n\abs{t})^2/\epsilon$ for some $c,n_0,t_0,\epsilon_0>0$ and all $\abs{t}\geq t_0$, $0<\epsilon<\epsilon_0$, and integers $n\geq n_0$. If $\mathscr{F}(\tau)$ is an operator-valued function, we first compute its spectral norm and analyze the asymptotic scaling of $\norm{\mathscr{F}(\tau)}$. We write $f=\Om{g}$ if $g=\cO{f}$, and $f=\Th{g}$ if both $f=\cO{g}$ and $f=\Om{g}$. We use $\widetilde{\mathcal{O}}$ to suppress logarithmic factors in the asymptotic expression and $o(1)$ to represent a positive number that approaches zero as some parameter grows.

Finally, we use $\overleftarrow{\prod}$, $\prod_{\gamma=1}^{\Gamma}$ to denote a product where the elements have increasing indices from right to left and $\overrightarrow{\prod}$, $\prod_{\gamma=\Gamma}^{1}$ vice versa. Under this convention,
\begin{equation}
	\prod_{\gamma=1}^{\Gamma}A_\gamma=\overleftarrow{\prod_{\gamma}}A_\gamma=A_\Gamma\cdots A_2A_1,\qquad
	\prod_{\gamma=\Gamma}^{1}A_\gamma=\overrightarrow{\prod_{\gamma}}A_\gamma=A_1A_2\cdots A_\Gamma.
\end{equation}
We let a summation be zero if its lower limit exceeds its upper limit.

\subsection{Time-ordered evolutions}
\label{sec:prelim_expT}
Let $\mathscr{H}(\tau)$ be an operator-valued function defined for $0\leq\tau\leq t$. We say that $\mathscr{U}(\tau)$ is the time-ordered evolution generated by $\mathscr{H}(\tau)$ if $\mathscr{U}(0)=I$ and $\frac{\mathrm{d}}{\mathrm{d}\tau}\mathscr{U}(\tau)=\mathscr{H}(\tau)\mathscr{U}(\tau)$ for $0\leq\tau\leq t$. In the case where $\mathscr{H}(\tau)$ is anti-Hermitian, the function $\mathscr{U}(\tau)$ represents the evolution of a quantum system under Hamiltonian $i\mathscr{H}(\tau)$.
We do not impose any restrictions on the Hermiticity of $\mathscr{H}(\tau)$ in the development of our theory, so our analysis can be applied to not only real-time but also imaginary-time evolutions. Throughout this paper, we assume that operator-valued functions are continuous, which guarantees the existence and uniqueness of their generated evolutions \cite[p.\ 12]{bk:dollard_friedman}. We then formally represent the time-ordered evolution $\mathscr{U}(t)$ by $\expT\big(\int_{0}^{t}\mathrm{d}\tau \mathscr{H}(\tau)\big)$, where $\expT$ denotes the time-ordered exponential. In the special case where $\mathscr{H}(\tau)=H$ is constant, the generated evolution is given by an ordinary matrix exponential $\expT\big(\int_{0}^{t}\mathrm{d}\tau \mathscr{H}(\tau)\big)=e^{tH}$.

In a similar way, we define the time-ordered evolution $\expT\big(\int_{t_1}^{t_2}\mathrm{d}\tau \mathscr{H}(\tau)\big)$ generated on an arbitrary interval $t_1\leq\tau\leq t_2$. Its determinant satisfies $\det\big(\expT\big(\int_{t_1}^{t_2}\mathrm{d}\tau \mathscr{H}(\tau)\big)\big)=e^{\int_{t_1}^{t_2}\mathrm{d}\tau\mathrm{Tr}(\mathscr{H}(\tau))}\neq0$ \cite[p.\ 9]{bk:dollard_friedman}, so the inverse operator $\expT^{-1}\big(\int_{t_1}^{t_2}\mathrm{d}\tau \mathscr{H}(\tau)\big)$ exists; we denote it by $\expT\big(\int_{t_2}^{t_1}\mathrm{d}\tau \mathscr{H}(\tau)\big)$. We have thus defined $\expT\big(\int_{t_1}^{t_2}\mathrm{d}\tau \mathscr{H}(\tau)\big)$ for every pair of $t_1$ and $t_2$ in the domain of $\mathscr{H}(\tau)$.\footnote{Alternatively, we may define a time-ordered exponential by its Dyson series or by a convergent sequence of products of ordinary matrix exponentials, and verify that this alternative definition satisfies the desired differential equation. We prefer the differential-equation definition since it is more versatile for the analysis in this paper.} Time-ordered exponentials satisfy the differentiation rule \cite[p.\ 12]{bk:dollard_friedman}
\begin{equation}
\begin{aligned}
\frac{\partial}{\partial t_2}\expT\bigg(\int_{t_1}^{t_2}\mathrm{d}\tau\ \mathscr{H}(\tau)\bigg)&=\mathscr{H}(t_2)\expT\bigg(\int_{t_1}^{t_2}\mathrm{d}\tau\ \mathscr{H}(\tau)\bigg),\\
\frac{\partial}{\partial t_1}\expT\bigg(\int_{t_1}^{t_2}\mathrm{d}\tau\ \mathscr{H}(\tau)\bigg)&=-\expT\bigg(\int_{t_1}^{t_2}\mathrm{d}\tau\ \mathscr{H}(\tau)\bigg)\mathscr{H}(t_1),
\end{aligned}
\end{equation}
and the multiplicative property \cite[p.\ 11]{bk:dollard_friedman}
\begin{equation}
\expT\bigg(\int_{t_1}^{t_3}\mathrm{d}\tau\ \mathscr{H}(\tau)\bigg)=\expT\bigg(\int_{t_2}^{t_3}\mathrm{d}\tau\ \mathscr{H}(\tau)\bigg)\expT\bigg(\int_{t_1}^{t_2}\mathrm{d}\tau\ \mathscr{H}(\tau)\bigg).
\end{equation}

By definition, the operator-valued function $\mathscr{U}(t) = \expT\big(\int_{0}^{t}\mathrm{d}\tau \mathscr{H}(\tau)\big)$ satisfies the differential equation $\frac{\mathrm{d}}{\mathrm{d}\tau}\mathscr{U}(\tau)=\mathscr{H}(\tau)\mathscr{U}(\tau)$ with initial condition $\mathscr{U}(0)=I$. We then apply the fundamental theorem of calculus to obtain the integral equation
\begin{equation}
\label{eq:int_eq}
\mathscr{U}(t) = I + \int_{0}^{t}\mathrm{d}\tau\ \mathscr{H}(\tau) \mathscr{U}(\tau).
\end{equation}
We also consider a general differential equation $\frac{\mathrm{d}}{\mathrm{d}t}\mathscr{U}(t)=\mathscr{H}(t)\mathscr{U}(t)+\mathscr{R}(t)$, whose solution is given by the following variation-of-parameters formula:
\begin{lemma}[Variation-of-parameters formula {\cite[Theorem 4.9]{bk:knapp}} {\cite[p.\ 17]{bk:dollard_friedman}}]
	\label{lem:td_Duhamel}
	Let $\mathscr{H}(\tau)$, $\mathscr{R}(\tau)$ be continuous operator-valued functions defined for $\tau\in\R$. Then the first-order differential equation
	\begin{equation}
	\frac{\mathrm{d}}{\mathrm{d}t}\mathscr{U}(t)=\mathscr{H}(t)\mathscr{U}(t)+\mathscr{R}(t),\quad \mathscr{U}(0)\text{ known},
	\end{equation}
	has a unique solution given by the variation-of-parameters formula
	\begin{equation}
	\mathscr{U}(t)=\exp_{\mathcal{T}}\bigg(\int_{0}^{t}\mathrm{d}\tau\ \mathscr{H}(\tau)\bigg)\mathscr{U}(0)
	+\int_{0}^{t}\mathrm{d}\tau_1\ \exp_{\mathcal{T}}\bigg(\int_{\tau_1}^{t}\mathrm{d}\tau_2\ \mathscr{H}(\tau_2)\bigg)\mathscr{R}(\tau_1).
	\end{equation}
\end{lemma}

Let $\mathscr{H}(\tau)=\mathscr{A}(\tau)+\mathscr{B}(\tau)$ be a continuous operator-valued function with two summands defined for $0\leq\tau\leq t$. Then, the evolution under $\mathscr{H}(\tau)$ can be seen as the evolution under the rotated operator $\expT^{-1}\big(\int_{0}^{\tau}\mathrm{d}\tau_2\mathscr{A}(\tau_2)\big)\mathscr{B}(\tau)\expT\big(\int_{0}^{\tau}\mathrm{d}\tau_2\mathscr{A}(\tau_2)\big)$, followed by another evolution under $\mathscr{A}(\tau)$ that rotates back to the original frame \cite{LW18}. This is known as the ``interaction-picture'' representation in quantum mechanics and is formally stated in the following lemma.
\begin{lemma}[Time-ordered evolution in the interaction picture {\cite[p.\ 21]{bk:dollard_friedman}}]
	\label{lem:interaction_picture}
	Let $\mathscr{H}(\tau)=\mathscr{A}(\tau)+\mathscr{B}(\tau)$ be an operator-valued function defined for $\tau\in \R$ with continuous summands $\mathscr{A}(\tau)$ and $\mathscr{B}(\tau)$. Then
	\begin{equation}
	\begin{aligned}
	\exp_{\mathcal{T}}\bigg(\int_{0}^{t}\mathrm{d}\tau\ \mathscr{H}(\tau)\bigg)
	&=\exp_{\mathcal{T}}\bigg(\int_{0}^{t}\mathrm{d}\tau\ \mathscr{A}(\tau)\bigg)\\
	&\quad \cdot\expT\bigg(\int_{0}^{t}\mathrm{d}\tau_1\
	\expT^{-1}\bigg(\int_{0}^{\tau_1}\mathrm{d}\tau_2\ \mathscr{A}(\tau_2)\bigg)\mathscr{B}(\tau_1)\expT\bigg(\int_{0}^{\tau_1}\mathrm{d}\tau_2\ \mathscr{A}(\tau_2)\bigg)
	\bigg).
	\end{aligned}
	\end{equation}
\end{lemma}
\begin{proof}
	A simple calculation shows that the right-hand side of the above equation satisfies the differential equation $\frac{\mathrm{d}}{\mathrm{d}t}\mathscr{U}(t)=\mathscr{H}(t)\mathscr{U}(t)$ with initial condition $\mathscr{U}(0)=I$. The lemma then follows as $\exp_{\mathcal{T}}\big(\int_{0}^{t}\mathrm{d}\tau\mathscr{H}(\tau)\big)$ is the unique solution to this differential equation.
\end{proof}

For any continuous $\mathscr{H}(\tau)$, the evolution $\exp_{\mathcal{T}}\big(\int_{0}^{t}\mathrm{d}\tau\mathscr{H}(\tau)\big)$ it generates is invertible and continuously differentiable. Conversely, the following lemma asserts that any operator-valued function that is invertible and continuously differentiable is a time-ordered evolution generated by some continuous function.
\begin{lemma}[Fundamental theorem of time-ordered evolution {\cite[p.\ 20]{bk:dollard_friedman}}]
	\label{lem:fte}
	The following statements regarding an operator-valued function $\mathscr{U}(\tau)$ ($\tau\in \R$) are equivalent:
	\begin{enumerate}
		\item $\mathscr{U}(\tau)$ is invertible and continuously differentiable;
		\item $\mathscr{U}(\tau)=\expT\big(\int_{0}^{\tau}\mathrm{d}\tau_1\mathscr{H}(\tau_1)\big)\mathscr{U}(0)$ for some continuous operator-valued function $\mathscr{H}(\tau)$.
	\end{enumerate}
	Furthermore, in the second statement,
	$\mathscr{H}(\tau)=\big(\frac{\mathrm{d}}{\mathrm{d}\tau}\mathscr{U}(\tau)\big)\mathscr{U}^{-1}(\tau)$ is uniquely determined.
\end{lemma}

Finally, we bound the spectral norm of a time-ordered evolution $\exp_{\mathcal{T}}\big(\int_{t_1}^{t_2}\mathrm{d}\tau\mathscr{H}(\tau)\big)$ and the distance between two evolutions.
\begin{lemma}[Spectral-norm bound for time-ordered evolution {\cite[p.\ 28]{bk:dollard_friedman}}]
	\label{lem:time_ordered_norm_bound}
	Let $\mathscr{H}(\tau)$ be a continuous operator-valued function defined on $\R$. Then,
	\begin{enumerate}
		\item $\norm{\expT\big(\int_{t_1}^{t_2}\mathrm{d}\tau\mathscr{H}(\tau)\big)}\leq e^{\abs{\int_{t_1}^{t_2}\mathrm{d}\tau\norm{\mathscr{H}(\tau)}}}$; and
		\item $\norm{\expT\big(\int_{t_1}^{t_2}\mathrm{d}\tau\mathscr{H}(\tau)\big)}= 1$ if $\mathscr{H}(\tau)$ is anti-Hermitian.
	\end{enumerate}
\end{lemma}
\begin{corollary}[Distance bound for time-ordered evolutions {\cite[Appendix B]{Tran18}}]
	\label{cor:time_ordered_distance_bound}
	Let $\mathscr{H}(\tau)$ and $\mathscr{G}(\tau)$ be continuous operator-valued functions defined on $\R$. Then,
	\begin{enumerate}
		\item $\norm{\expT\big(\int_{t_1}^{t_2}\mathrm{d}\tau\mathscr{H}(\tau)\big)-\expT\big(\int_{t_1}^{t_2}\mathrm{d}\tau\mathscr{G}(\tau)\big)}
		\leq\abs{\int_{t_1}^{t_2}\mathrm{d}\tau\norm{\mathscr{H}(\tau)-\mathscr{G}(\tau)}}e^{\abs{\int_{t_1}^{t_2}\mathrm{d}\tau(\norm{\mathscr{H}(\tau)}+\norm{\mathscr{G}(\tau)})}}$; and
		\item $\norm{\expT\big(\int_{t_1}^{t_2}\mathrm{d}\tau\mathscr{H}(\tau)\big)-\expT\big(\int_{t_1}^{t_2}\mathrm{d}\tau\mathscr{G}(\tau)\big)}
		\leq\abs{\int_{t_1}^{t_2}\mathrm{d}\tau\norm{\mathscr{H}(\tau)-\mathscr{G}(\tau)}}$ if $\mathscr{H}(\tau)$ and $\mathscr{G}(\tau)$ are anti-Hermitian.
	\end{enumerate}
\end{corollary}

\subsection{Product formulas}
\label{sec:prelim_pf}
Let $H=\sum_{\gamma=1}^{\Gamma}H_\gamma$ be a time-independent operator consisting of $\Gamma$ summands, so that the evolution generated by $H$ is $e^{t\sum_{\gamma=1}^{\Gamma}H_\gamma}$. Product formulas provide a convenient way of decomposing such an evolution into a product of exponentials of individual $H_\gamma$. Examples of product formulas include the first-order Lie-Trotter formula
\begin{equation}
\label{eq:pf1}
\mathscr{S}_1(t):=e^{tH_\Gamma}\cdots e^{tH_1}
\end{equation}
and higher-order Suzuki formulas~\cite{Suz91} defined recursively via
\begin{equation}
\label{eq:pf2k}
\begin{aligned}
\mathscr{S}_2(t)&:=e^{\frac{t}{2}H_1}\cdots e^{\frac{t}{2}H_\Gamma}e^{\frac{t}{2}H_\Gamma}\cdots e^{\frac{t}{2}H_1},\\
\mathscr{S}_{2k}(t)&:=\mathscr{S}_{2k-2}(u_{k}t)^2 \, \mathscr{S}_{2k-2}((1-4u_{k})t) \, \mathscr{S}_{2k-2}(u_{k}t)^2,
\end{aligned}
\end{equation}
where $u_{k}:=1/(4-4^{1/(2k-1)})$. It is a challenge in practice to find the formula with the best performance for simulating a specific physical system~\cite{CMNRS18}. However, we address a different question, developing a theory of Trotter error that holds for a general product formula. For in-depth studies of these formulas, especially in the context of numerical analysis, we refer the reader to \cite{Jahnke2000,Thalhammer08,Thalhammer12,mclachlan_quispel_2002,McLachlan95,hairer2006geometric} and the references therein.

Specifically, we consider a product formula of the form
\begin{equation}
\mathscr{S}(t):=\prod_{\upsilon=1}^{\Upsilon}\prod_{\gamma=1}^{\Gamma}e^{ta_{(\upsilon,\gamma)}H_{\pi_{\upsilon}(\gamma)}},\label{eq:pf}
\end{equation}
where the coefficients $a_{(\upsilon,\gamma)}$ are real numbers. The parameter $\Upsilon$ denotes the number of \emph{stages} of the formula; for the Suzuki formula $\mathscr{S}_{2k}(t)$, we have $\Upsilon=2\cdot 5^{k-1}$. The permutation $\pi_{\upsilon}$ controls the ordering of operator summands within stage $\upsilon$ of the formula. For Suzuki's constructions, we alternately reverse the ordering of summands between neighboring stages, but other formulas may use general permutations. Throughout this paper, we fix $\Upsilon$, $\pi_{\upsilon}$ and assume that the coefficients $a_{(\upsilon,\gamma)}$ are uniformly bounded by $1$ in absolute value. We then consider the performance of the product formula with respect to the input operator summands $H_\gamma$ (for $\gamma=1,\ldots,\Gamma$) and the evolution time $t$.

Product formulas provide a good approximation to the ideal evolution when the time $t$ is small. Specifically, a $p$th-order formula $\mathscr{S}(t)$ satisfies
\begin{equation}
\mathscr{S}(t)=e^{tH}+\OO{t^{p+1}}.
\end{equation}
This asymptotic analysis gives the correct error scaling with respect to $t$, but the dependence on the $H_\gamma$ is ignored, so it does not provide a full characterization of Trotter error. This issue was addressed in the work of Berry, Ahokas, Cleve, and Sanders~\cite{BACS05}, who gave a concrete error bound for product formulas with dependence on both $t$ and $H_\gamma$. Their original bound depends on the $\infty$-norm $\Gamma\max_\gamma\norm{H_\gamma}$, although it is not hard to improve this to the $1$-norm scaling $\sum_{\gamma=1}^{\Gamma}\norm{H_\gamma}$. We prove a new error bound in the lemma below; for real-time evolutions, this improves a multiplicative factor of $e^{t\Upsilon\sum_{\gamma=1}^{\Gamma}\norm{H_\gamma}}$ over the best previous analysis \cite[Eq.~(13)]{LKW19}.

\begin{lemma}[Trotter error with $1$-norm scaling]
	\label{lem:trotter_error_one_norm_scaling}
	Let $H=\sum_{\gamma=1}^{\Gamma}H_\gamma$ be an operator consisting of $\Gamma$ summands and $t\geq 0$. Let $\mathscr{S}(t)=\prod_{\upsilon=1}^{\Upsilon}\prod_{\gamma=1}^{\Gamma}e^{ta_{(\upsilon,\gamma)}H_{\pi_{\upsilon}(\gamma)}}$ be a $p$th-order product formula. Then,
	\begin{equation}
	\label{eq:one_norm_error_imag}
	\norm{\mathscr{S}(t)-e^{tH}}=\cO{\bigg(\sum_{\gamma=1}^{\Gamma}\norm{H_\gamma}t\bigg)^{p+1}e^{t\Upsilon\sum_{\gamma=1}^{\Gamma}\norm{H_\gamma}}}.
	\end{equation}
	Furthermore, if $H_\gamma$ are anti-Hermitian,
	\begin{equation}
	\label{eq:one_norm_error_real}
	\norm{\mathscr{S}(t)-e^{tH}}=\cO{\bigg(\sum_{\gamma=1}^{\Gamma}\norm{H_\gamma}t\bigg)^{p+1}}.
	\end{equation}
\end{lemma}
\begin{proof}
	Since $\mathscr{S}(t)$ is a $p$th-order formula, we know from \cite[Supplementary Lemma 1]{CS19} that $\mathscr{S}(0)=\mathscr{S}'(0)=\cdots=\mathscr{S}^{(p)}(0)=0$. By Taylor's theorem,
	\begin{equation}
	\mathscr{S}(t)-e^{tH}
	=(p+1)\int_{0}^{1}\mathrm{d}u\ (1-u)^{p}\frac{t^{p+1}}{(p+1)!}\big(\mathscr{S}^{(p+1)}(ut)-H^{p+1}e^{utH}\big),
	\end{equation}
	where
	\begin{equation}
	\mathscr{S}^{(p+1)}(ut)
	=\sum_{q_{(1,1)}+\cdots+q_{(\Upsilon,\Gamma)}=p+1}\binom{p+1}{q_{(1,1)}\ \cdots\ q_{(\Upsilon,\Gamma)}}\prod_{\upsilon=1}^{\Upsilon}\prod_{\gamma=1}^{\Gamma}
	\big(a_{(\upsilon,\gamma)}H_{\pi_\upsilon(\gamma)}\big)^{q_{(\upsilon,\gamma)}}e^{uta_{(\upsilon,\gamma)}H_{\pi_\upsilon(\gamma)}}.
	\end{equation}
	The spectral norms of $\mathscr{S}^{(p+1)}(ut)$ and $H^{p+1}e^{utH}$ can be bounded as
	\begin{equation}
	\begin{aligned}
	\norm{\mathscr{S}^{(p+1)}(ut)}
	&\leq\sum_{q_{(1,1)}+\cdots+q_{(\Upsilon,\Gamma)}=p+1}\binom{p+1}{q_{(1,1)}\ \cdots\ q_{(\Upsilon,\Gamma)}}\prod_{\upsilon=1}^{\Upsilon}\prod_{\gamma=1}^{\Gamma}
	\norm{H_{\pi_\upsilon(\gamma)}}^{q_{(\upsilon,\gamma)}}e^{t\norm{H_{\pi_\upsilon(\gamma)}}}\\
	&=\bigg(\Upsilon\sum_{\gamma=1}^{\Gamma}\norm{H_\gamma}\bigg)^{p+1}e^{t\Upsilon\sum_{\gamma=1}^{\Gamma}\norm{H_\gamma}},\\
	\norm{H^{p+1}e^{utH}}
	&\leq\bigg(\sum_{\gamma=1}^{\Gamma}\norm{H_\gamma}\bigg)^{p+1}e^{t\sum_{\gamma=1}^{\Gamma}\norm{H_\gamma}}.
	\end{aligned}
	\end{equation}
	Applying these bounds to the Taylor expansion, we find that
	\begin{equation}\label{eq:trotter_error_one_norm_scaling_bound}
	\begin{aligned}
	\norm{\mathscr{S}(t)-e^{tH}}
	&\leq\frac{t^{p+1}}{(p+1)!}\bigg[\bigg(\Upsilon\sum_{\gamma=1}^{\Gamma}\norm{H_\gamma}\bigg)^{p+1}e^{t\Upsilon\sum_{\gamma=1}^{\Gamma}\norm{H_\gamma}}
	+\bigg(\sum_{\gamma=1}^{\Gamma}\norm{H_\gamma}\bigg)^{p+1}e^{t\sum_{\gamma=1}^{\Gamma}\norm{H_\gamma}}\bigg]\\
	&=\cO{\bigg(\sum_{\gamma=1}^{\Gamma}\norm{H_\gamma}t\bigg)^{p+1}e^{t\Upsilon\sum_{\gamma=1}^{\Gamma}\norm{H_\gamma}}}.
	\end{aligned}
	\end{equation}
	The special case where $H_\gamma$ are anti-Hermitian can be proved in a similar way, except we directly evaluate the spectral norm of a matrix exponential to $1$.
\end{proof}

The above bound on the Trotter error works well for small $t$. To simulate anti-Hermitian $H_\gamma$ for a large time, we divide the evolution into $r$ steps and apply the product formula within each step. The overall simulation has error
\begin{equation}
\norm{\mathscr{S}^r(t/r)-e^{tH}}
\leq r\norm{\mathscr{S}(t/r)-e^{\frac{t}{r}H}}
=\cO{\frac{\big(\sum_{\gamma=1}^{\Gamma}\norm{H_\gamma}t\big)^{p+1}}{r^p}}.
\end{equation}
To simulate with accuracy $\epsilon$, it suffices to choose
\begin{equation}
r=\cO{\frac{\big(\sum_{\gamma=1}^{\Gamma}\norm{H_\gamma}t\big)^{1+1/p}}{\epsilon^{1/p}}}.
\end{equation}
We have thus proved:
\begin{corollary}[Trotter number with $1$-norm scaling]
	\label{cor:trotter_number_one_norm_scaling}
	Let $H=\sum_{\gamma=1}^{\Gamma}H_\gamma$ be an operator consisting of $\Gamma$ summands with $H_\gamma$ anti-Hermitian and $t\geq 0$. Let $\mathscr{S}(t)$ be a $p$th-order product formula. Then, we have $\norm{\mathscr{S}^r(t/r)-e^{tH}}=\cO{\epsilon}$ provided
	\begin{equation}
	r=\cO{\frac{\big(\sum_{\gamma=1}^{\Gamma}\norm{H_\gamma}t\big)^{1+1/p}}{\epsilon^{1/p}}}.
	\end{equation}
\end{corollary}

Note that the above analysis only uses information about the norms of the summands. In the extreme case where all $H_\gamma$ commute, the Trotter error becomes zero but the above bound can be arbitrarily large. This suggests that the analysis can be significantly improved by leveraging information about commutation of the $H_\gamma$. Unfortunately, despite extensive efforts, dramatic improvements to the Trotter error bound are only known for certain low-order formulas \cite{Huyghebaert_1990,Suzuki85,WBCHT14,DT10,Kivlichan19} and special systems \cite{Somma16,CS19}.

To explain the limitations of prior approaches, it is instructive to examine a general bound developed by Descombes and Thalhammer \cite{Thalhammer08,DT10}
\begin{equation*}
\norm{\mathscr{S}(t)-e^{tH}}\leq b_{p+1}t^{p+1} + \cdots +b_{q}t^{q}+b_{q+1}t^{q+1},
\end{equation*}
where $H=\sum_{\gamma=1}^{\Gamma}H_\gamma$ is a sum of anti-Hermitian operators, $\mathscr{S}(t)$ is a $p$th-order formula, $q\geq p$ is a positive integer, and $t\geq 0$, suggesting a choice of
\begin{equation*}
r=\max\Bigg\{\cO{\frac{b_{p+1}^{1/p}t^{1+1/p}}{\epsilon^{1/p}}},\ldots,\cO{\frac{b_{q}^{1/(q-1)}t^{1+1/(q-1)}}{\epsilon^{1/(q-1)}}},
\cO{\frac{b_{q+1}^{1/q}t^{1+1/q}}{\epsilon^{1/q}}}\Bigg\}
\end{equation*}
to simulate with accuracy $\epsilon$. Here, all the leading coefficients $b_{p+1},\ldots,b_q$ depend on nested commutators of $H_\gamma$, but $b_{q+1}$ is determined by commutators interlaced with matrix exponentials, which is technically challenging to evaluate except for geometrically local systems. Consequently, a bound on $b_{q+1}$ must be used, resulting in a $1$-norm scaling similar to that of \lem{trotter_error_one_norm_scaling} and a loose Trotter error estimate for simulating general quantum systems.

We develop a theory of Trotter error that directly exploits the commutativity of operator summands. The resulting bound naturally reduces to the previous bounds for low-order formulas and special systems, but our analysis uncovers a host of new speedups for product formulas that were previously unknown. The central concepts of this theory are the types, order conditions, and representations of Trotter error, which we explain in \sec{theory}.

\section{Theory}
\label{sec:theory}

We now develop a theory for analyzing Trotter error. We explain the core ideas of this theory in \sec{theory_example} using the simple example of the first-order Lie-Trotter formula.
We then discuss the analysis of a general formula.
In particular, we study various types of Trotter error in \sec{theory_type} and compute their order conditions in \sec{theory_order}.
We then derive explicit representations of Trotter error in \sec{theory_rep}, establishing the commutator scaling of Trotter error in \thm{trotter_error_comm_scaling}.
We focus on the asymptotic error scaling here, and discuss potential applications and constant-prefactor improvements of our results in \sec{app} and \sec{prefactor}, respectively.

\subsection{Example of the Lie-Trotter formula}
\label{sec:theory_example}
In this section, we use the example of the first-order Lie-Trotter formula to illustrate the general theory we develop for analyzing Trotter error. For simplicity, consider an operator $H=A+B$ with two summands. The ideal evolution generated by $H$ is given by $e^{tH}=e^{t(A+B)}$. To decompose this evolution, we may use the Lie-Trotter formula $\mathscr{S}_1(t)=e^{tB}e^{tA}$. This formula is first-order accurate, so we have $\mathscr{S}_1(t)=e^{tH}+\OO{t^2}$.

A key observation here is that the error of a product formula can have various \emph{types}. Specifically, we consider three types of Trotter error: additive error, multiplicative error, and error that appears in the exponent. Note that $\mathscr{S}_1(t)$ satisfies the differential equation $\frac{\mathrm{d}}{\mathrm{d}t}\mathscr{S}_1(t)=H\mathscr{S}_1(t)+\big[e^{tB},A\big]e^{tA}$ with initial condition $\mathscr{S}_1(0)=I$. By the variation-of-parameters formula (\lem{td_Duhamel}),
\begin{equation}
\mathscr{S}_1(t)=e^{tH}+\int_{0}^{t}\mathrm{d}\tau\ e^{(t-\tau)H}\big[e^{\tau B},A\big]e^{\tau A},
\end{equation}
so we get the additive error $\mathscr{A}_1(t)=\int_{0}^{t}\mathrm{d}\tau\ e^{(t-\tau)H}\big[e^{\tau B},A\big]e^{\tau A}$ of the Lie-Trotter formula. For error with the exponentiated type, we differentiate $\mathscr{S}_1(t)$ to get $\frac{\mathrm{d}}{\mathrm{d}t}\mathscr{S}_1(t)=\big(B+e^{tB}Ae^{-tB}\big)\mathscr{S}_1(t)$. Applying the fundamental theorem of time-ordered evolution (\lem{fte}), we have
\begin{equation}
\mathscr{S}_1(t)=\expT\bigg(\int_{0}^{t}\mathrm{d}\tau\big(B+e^{\tau B}Ae^{-\tau B}\big)\bigg),
\end{equation}
so $\mathscr{E}_1(\tau)=e^{\tau B}Ae^{-\tau B}-A$ is the error of Lie-Trotter formula that appears in the exponent. To obtain the multiplicative error, we switch to the interaction picture using \lem{interaction_picture}:
\begin{equation}
\mathscr{S}_1(t)=e^{tH}\expT\bigg(\int_{0}^{t}\mathrm{d}\tau\big(e^{-\tau H}e^{\tau B}Ae^{-\tau B}e^{\tau H}-e^{-\tau H}Ae^{\tau H}\big)\bigg),
\end{equation}
so $\mathscr{M}_1(t)=\expT\big(\int_{0}^{t}\mathrm{d}\tau\big(e^{-\tau H}e^{\tau B}Ae^{-\tau B}e^{\tau H}-e^{-\tau H}Ae^{\tau H}\big)\big)-I$ is the multiplicative Trotter error. These three types of Trotter error are equivalent for analyzing the complexity of digital quantum simulation (\sec{app_dqs}) and simulating local observables (\sec{app_local_obs}), whereas the multiplicative error and the exponentiated error are more versatile when applied to quantum Monte Carlo simulation (\sec{app_qmc}). We compute error operators for a general product formula in \sec{theory_type}.

Since product formulas provide a good approximation to the ideal evolution for small $t$, we expect all three error operators $\mathscr{A}_1(t)$, $\mathscr{E}_1(t)$, and $\mathscr{M}_1(t)$ to converge to zero in the limit $t\rightarrow 0$. The rates of convergence are what we call \emph{order conditions}. More precisely,
\begin{equation}
\begin{aligned}
\mathscr{A}_1(t)&=\int_{0}^{t}\mathrm{d}\tau\ e^{(t-\tau)H}\big[e^{\tau B},A\big]e^{\tau A}=\OO{t^2},\\
\mathscr{E}_1(t)&=e^{t B}Ae^{-t B}-A=\OO{t},\\
\mathscr{M}_1(t)&=\expT\bigg(\int_{0}^{t}\mathrm{d}\tau\big(e^{-\tau H}e^{\tau B}Ae^{-\tau B}e^{\tau H}-e^{-\tau H}Ae^{\tau H}\big)\bigg)-I=\OO{t^2}.
\end{aligned}
\end{equation}
For the Lie-Trotter formula, these conditions can be verified by direct calculation, although such an approach becomes inefficient in general. Instead, we describe an indirect approach in \sec{theory_order} to compute order conditions for a general product formula.

Finally, we consider \emph{representations} of Trotter error that leverage the commutativity of operator summands. We discuss how to represent $\mathscr{M}_1(t)$ in detail, although it is straightforward to extend the analysis to $\mathscr{A}_1(t)$ and $\mathscr{E}_1(t)$ as well. To this end, we first consider the term $e^{-\tau H}e^{\tau B}Ae^{-\tau B}e^{\tau H}$, which contains two layers of conjugations of matrix exponentials. We apply the fundamental theorem of calculus to the first layer of conjugation and obtain
\begin{equation}
e^{\tau B}Ae^{-\tau B}=A+\int_{0}^{\tau}\mathrm{d}\tau_2\ e^{\tau_2 B}\big[B,A\big]e^{-\tau_2 B}.
\end{equation}
After cancellation, this gives
\begin{equation}
\mathscr{M}_1(t)=\expT\bigg(\int_{0}^{t}\mathrm{d}\tau\int_{0}^{\tau}\mathrm{d}\tau_2\ e^{-\tau H}e^{\tau_2 B}\big[B,A\big]e^{-\tau_2 B}e^{\tau H}\bigg)-I,
\end{equation}
which implies, through \cor{time_ordered_distance_bound}, that $\norm{\mathscr{M}_1(t)}=\cO{\norm{[B,A]}t^2}$ when $A$, $B$ are anti-Hermitian and $t\geq 0$. In the above derivation, it is important that we only expand the first layer of conjugation of exponentials, that we apply the fundamental theorem of calculus only once, and that we can cancel the terms $e^{-\tau H}Ae^{\tau H}$ in pairs. The validity of such an approach in general is guaranteed by the appropriate order condition, which we explain in detail in \sec{theory_rep}.

\subsection{Error types}
\label{sec:theory_type}
In this section, we discuss error types of a general product formula. In particular, we give explicit expressions for three different types of Trotter error: the additive error, the multiplicative error, and error that appears in the exponent of a time-ordered exponential (the ``exponentiated'' error). These types are equivalent for analyzing the complexity of simulating quantum dynamics and local observables, but the latter two types are more versatile for quantum Monte Carlo simulation.

Let $H=\sum_{\gamma=1}^{\Gamma}H_\gamma$ be an operator with $\Gamma$ summands. The ideal evolution under $H$ for time $t$ is given by $e^{tH}=e^{t\sum_{\gamma=1}^{\Gamma}H_\gamma}$, which we approximate by a general product formula $\mathscr{S}(t)=\prod_{\upsilon=1}^{\Upsilon}\prod_{\gamma=1}^{\Gamma}e^{ta_{(\upsilon,\gamma)}H_{\pi_{\upsilon}(\gamma)}}$. For convenience, we use the lexicographic order on a pair of tuples $(\upsilon,\gamma)$ and $(\upsilon',\gamma')$, defined as follows: we write $(\upsilon,\gamma)\succeq(\upsilon',\gamma')$ if $\upsilon> \upsilon'$, or if $\upsilon=\upsilon'$ and $\gamma\geq \gamma'$. We have $(\upsilon,\gamma)\succ(\upsilon',\gamma')$ if both $(\upsilon,\gamma)\succeq(\upsilon',\gamma')$ and $(\upsilon,\gamma)\neq(\upsilon',\gamma')$ hold. Notations $(\upsilon,\gamma)\preceq(\upsilon',\gamma')$ and $(\upsilon,\gamma)\prec(\upsilon',\gamma')$ are defined in a similar way, except that we reverse the directions of all the inequalities. We use $(\upsilon,\gamma)-1$ to represent the immediate predecessor of $(\upsilon,\gamma)$ with respect to the lexicographic order and $(\upsilon,\gamma)+1$ to denote the immediate successor.

For the additive Trotter error, we seek an operator-valued function $\mathscr{A}(t)$ such that $\mathscr{S}(t)=e^{tH}+\mathscr{A}(t)$. This can be achieved by constructing the differential equation $\frac{\mathrm{d}}{\mathrm{d}t}\mathscr{S}(t)=H\mathscr{S}(t)+\mathscr{R}(t)$ with initial condition $\mathscr{S}(0)=I$, followed by the use of the variation-of-parameters formula (\lem{td_Duhamel}). For the exponentiated type of Trotter error, we aim to construct an operator-valued function $\mathscr{E}(t)$ such that $\mathscr{S}(t)=\expT\big(\int_{0}^{t}\mathrm{d}\tau\big(H+\mathscr{E}(\tau)\big)\big)$. We find $\mathscr{E}(t)$ by differentiating the product formula $\mathscr{S}(t)$ and applying the fundamental theorem of time-ordered evolution (\lem{fte}). Finally, we obtain the multiplicative error by switching to the interaction picture using \lem{interaction_picture}. The derivation follows from a similar analysis as in \sec{theory_example} and is detailed in \append{type}.

\begin{restatable}[Types of Trotter error]{theorem}{thmtype}
	\label{thm:error_type}
	Let $H=\sum_{\gamma=1}^{\Gamma}H_\gamma$ be an operator with $\Gamma$ summands. The evolution under $H$ for time $t\in\R$ is given by $e^{tH}=e^{t\sum_{\gamma=1}^{\Gamma}H_\gamma}$, which we decompose using the product formula $\mathscr{S}(t)=\prod_{\upsilon=1}^{\Upsilon}\prod_{\gamma=1}^{\Gamma}e^{ta_{(\upsilon,\gamma)}H_{\pi_{\upsilon}(\gamma)}}$. Then,
	\begin{enumerate}
		\item Trotter error can be expressed in the additive form $\mathscr{S}(t)=e^{tH}+\int_{0}^{t}\mathrm{d}\tau\ e^{(t-\tau)H}\mathscr{S}(\tau)\mathscr{T}(\tau)$, where
		\begin{equation}
		\begin{aligned}
		\mathscr{T}(\tau)
		=&\sum_{(\upsilon,\gamma)}\prod_{(\upsilon',\gamma')\prec(\upsilon,\gamma)}^{\longrightarrow}e^{-\tau a_{(\upsilon',\gamma')}H_{\pi_{\upsilon'}(\gamma')}}\big(a_{(\upsilon,\gamma)}H_{\pi_{\upsilon}(\gamma)}\big)\prod_{(\upsilon',\gamma')\prec(\upsilon,\gamma)}^{\longleftarrow}e^{\tau a_{(\upsilon',\gamma')}H_{\pi_{\upsilon'}(\gamma')}}\\
		&-\prod_{(\upsilon',\gamma')}^{\longrightarrow}e^{-\tau a_{(\upsilon',\gamma')}H_{\pi_{\upsilon'}(\gamma')}}H\prod_{(\upsilon',\gamma')}^{\longleftarrow}e^{\tau a_{(\upsilon',\gamma')}H_{\pi_{\upsilon'}(\gamma')}};
		\end{aligned}
		\end{equation}
		\item Trotter error can be expressed in the exponentiated form $\mathscr{S}(t)=\expT\big(\int_{0}^{t}\mathrm{d}\tau\big(H+\mathscr{E}(\tau)\big)\big)$, where
		\begin{equation}
		\mathscr{E}(\tau)=\sum_{(\upsilon,\gamma)}\prod_{(\upsilon',\gamma')\succ(\upsilon,\gamma)}^{\longleftarrow}e^{\tau a_{(\upsilon',\gamma')}H_{\pi_{\upsilon'}(\gamma')}}\big(a_{(\upsilon,\gamma)}H_{\pi_{\upsilon}(\gamma)}\big)\prod_{(\upsilon',\gamma')\succ(\upsilon,\gamma)}^{\longrightarrow}e^{-\tau a_{(\upsilon',\gamma')}H_{\pi_{\upsilon'}(\gamma')}}-H;
		\label{eq:Etau}
		\end{equation}
		\item Trotter error can be expressed in the multiplicative form $\mathscr{S}(t)=e^{tH}(I+\mathscr{M}(t))$, where
		\begin{equation}
		\begin{aligned}
		\mathscr{M}(t)&=\expT\bigg(\int_{0}^{t}\mathrm{d}\tau\ e^{-\tau H}\mathscr{E}(\tau)e^{\tau H}\bigg)-I\\
		\end{aligned}
		\end{equation}
		with $\mathscr{E}(\tau)$ as above.
	\end{enumerate}
\end{restatable}

Note that the error operators $\mathscr{T}(\tau)$ and $\mathscr{E}(\tau)$ both consist of conjugations of matrix exponentials of the form $e^{\tau A_s}\cdots e^{\tau A_{2}}e^{\tau A_{1}}Be^{-\tau A_{1}}e^{-\tau A_{2}}\cdots e^{-\tau A_s}$. To bound the Trotter error, it thus suffices to analyze such conjugations of matrix exponentials. The previous work of Somma \cite{Somma16} expanded them into infinite series of nested commutators, which is favorable for systems with appropriate Lie-algebraic structures. An alternative approach of Childs and Su \cite{CS19} represented them as commutators nested with conjugations of matrix exponentials, which provides a tight analysis for geometrically local systems. Unfortunately, both approaches can be loose in general. Instead, we apply order conditions (\sec{theory_order}) and derive a new representation of Trotter error (\sec{theory_rep}) that provides a tight analysis for general systems.

\subsection{Order conditions}
\label{sec:theory_order}
In this section, we study the order conditions of Trotter error.
By order condition, we mean the rate at which a continuous operator-valued function $\mathscr{F}(\tau)$, defined for $\tau\in\R$, approaches zero in the limit $\tau\rightarrow 0$. Formally, we write $\mathscr{F}(\tau)=\OO{\tau^{p}}$ with nonnegative integer $p$ if there exist constants $c,t_0>0$, independent of $\tau$, such that $\norm{\mathscr{F}(\tau)}\leq c\abs{\tau}^p$ whenever $\abs{\tau}\leq t_0$.

Order conditions arise naturally in the analysis of Trotter error \cite{Suz91,WBHS10,Auzinger2014,AKT14}. Indeed, a $p$th-order product formula $\mathscr{S}(t)$ has a Taylor expansion that agrees with the ideal evolution $e^{tH}$ up to order $t^{p}$, which implies the order condition $\mathscr{S}(t)=e^{tH}+\OO{t^{p+1}}$ by definition. Our approach is to use this relation in the reverse direction: given a smooth operator-valued function $\mathscr{F}(\tau)$ satisfying the order condition $\mathscr{F}(\tau)=\OO{\tau^p}$, we conclude that $\mathscr{F}(\tau)$ has a Taylor expansion where terms with order $\tau^{p-1}$ or lower vanish. We make this argument more precise in \append{order}.

We can determine the order condition of an operator-valued function through either direct calculation or indirect derivation. To illustrate this, we consider decomposing $e^{tH}=e^{t(A+B)}$ using the first-order Lie-Trotter formula $\mathscr{S}_1(t)=e^{tB}e^{tA}$. We see from \sec{theory_example} that this decomposition has the additive Trotter error
\begin{equation}
	\mathscr{A}_1(t)=\int_{0}^{t}\mathrm{d}\tau\ e^{(t-\tau)H}\big(\mathscr{S}_1'(\tau)-H\mathscr{S}_1(\tau)\big)
	=\int_{0}^{t}\mathrm{d}\tau\ e^{(t-\tau)H}\big[e^{\tau B},A\big]e^{\tau A}.
\end{equation}
We know that $\mathscr{A}_1(t)$ has order condition $\mathscr{A}_1(t)=\OO{t^2}$, which follows directly from the fact that $\mathscr{A}_1(0)=\mathscr{A}_1'(0)=0$. On the other hand, an indirect argument would proceed as follows. We use the known order condition $\mathscr{S}_1(t)=e^{tH}+\OO{t^2}$ to conclude that $\mathscr{S}_1'(\tau)-H\mathscr{S}_1(\tau)=\OO{\tau}$. Multiplying the matrix exponential $e^{(t-\tau)H}=\OO{1}$ does not change the order condition, so we still have $e^{(t-\tau)H}\big(\mathscr{S}_1'(\tau)-H\mathscr{S}_1(\tau)\big)=\OO{\tau}$. A final integration of $\int_{0}^{t}\mathrm{d}\tau$ then gives the desired condition $\mathscr{A}_1(t)=\OO{t^2}$.

Although we obtain the same order condition through two different analyses, the direct approach becomes inefficient for analyzing Trotter error of a general high-order product formula. Instead, we use the indirect analysis to prove the following theorem on the order conditions of Trotter error (see \append{order} for proof details). In \sec{theory_rep}, we apply these conditions to cancel low-order Trotter error terms and represent higher-order ones as nested commutators of operator summands.

\begin{restatable}[Order conditions of Trotter error]{theorem}{thmorder}
	\label{thm:error_order_cond}
	Let $H$ be an operator, and let $\mathscr{S}(\tau)$, $\mathscr{T}(\tau)$, $\mathscr{E}(\tau)$, and $\mathscr{M}(\tau)$ be infinitely differentiable operator-valued functions defined for $\tau\in\R$, such that
	\begin{equation}
	\begin{aligned}
	\mathscr{S}(t)&=e^{tH}+\int_{0}^{t}\mathrm{d}\tau\ e^{(t-\tau)H}\mathscr{S}(\tau)\mathscr{T}(\tau),\\
	&=\expT\bigg(\int_{0}^{t}\mathrm{d}\tau\big(H+\mathscr{E}(\tau)\big)\bigg),\\
	&=e^{tH}(I+\mathscr{M}(t)).
	\end{aligned}
	\end{equation}
	For any nonnegative integer $p$, the following conditions are equivalent:
	\begin{enumerate}
		\item $\mathscr{S}(t)=e^{tH}+\OO{t^{p+1}}$;
		\item $\mathscr{T}(\tau)=\OO{\tau^p}$;
		\item $\mathscr{E}(\tau)=\OO{\tau^p}$; and
		\item $\mathscr{M}(t)=\OO{t^{p+1}}$.
	\end{enumerate}
\end{restatable}

\subsection{Error representations}
\label{sec:theory_rep}
For a product formula with a certain error type and order condition, we now represent its error in terms of nested commutators of the operator summands.
In particular, we give upper bounds on the additive and the multiplicative errors of $p$th-order product formulas in \thm{trotter_error_comm_scaling}.

Consider an operator $H=\sum_{\gamma=1}^{\Gamma}H_\gamma$ with $\Gamma$ summands. The ideal evolution generated by $H$ is $e^{tH}$, which we decompose using a $p$th-order product formula $\mathscr{S}(t)=\prod_{\upsilon=1}^{\Upsilon}\prod_{\gamma=1}^{\Gamma}e^{ta_{(\upsilon,\gamma)}H_{\pi_{\upsilon}(\gamma)}}$. We know from \thm{error_type} that the Trotter error can be expressed in the additive form $\mathscr{S}(t)=e^{tH}+\int_{0}^{t}\mathrm{d}\tau\ e^{(t-\tau)H}\mathscr{S}(\tau)\mathscr{T}(\tau)$, the multiplicative form $\mathscr{S}(t)=e^{tH}(I+\mathscr{M}(t))$, where $\mathscr{M}(t)=\expT\big(\int_{0}^{t}\mathrm{d}\tau\ e^{-\tau H}\mathscr{E}(\tau)e^{\tau H}\big)-I$, and the exponentiated form $\mathscr{S}(t)=\expT\big(\int_{0}^{t}\mathrm{d}\tau\big(H+\mathscr{E}(\tau)\big)\big)$. Furthermore, both $\mathscr{T}(\tau)$ and $\mathscr{E}(\tau)$ consist of conjugations of matrix exponentials and have order condition $\mathscr{T}(\tau),\ \mathscr{E}(\tau)\in \OO{\tau^p}$ (\thm{error_order_cond}).

We first consider the representation of a single conjugation of matrix exponentials
\begin{equation}
\label{eq:unitary_conjugation}
e^{\tau A_s}\cdots e^{\tau A_{2}}e^{\tau A_{1}}Be^{-\tau A_{1}}e^{-\tau A_{2}}\cdots e^{-\tau A_s},
\end{equation}
where $A_1,A_2,\ldots,A_s,B$ are operators and $\tau\in\R$. Our goal is to expand this conjugation into a finite series in the time variable $\tau$. We only keep track of those terms with order $\OO{\tau^p}$, because terms corresponding to $1,\tau,\ldots,\tau^{p-1}$ will vanish in the final representation of Trotter error due to the order condition. As mentioned before, such a conjugation was previously analyzed based on a naive application of the Taylor's theorem \cite{CS19} and an infinite-series expansion \cite{Somma16}. However, those results do not represent Trotter error as a finite number of commutators of operator summands and they only apply to special systems such as those with geometrical locality or suitable Lie-algebraic structure. Our new representation overcomes these limitations.

We begin with the innermost layer $e^{\tau A_{1}}Be^{-\tau A_{1}}$. Applying Taylor's theorem to order $p-1$ with integral form of the remainder, we have
\begin{equation}
\begin{aligned}
e^{\tau A_{1}}Be^{-\tau A_{1}}&=B+\big[A_1,B\big]\tau+\cdots+\underbrace{\big[A_1,\cdots,\big[A_1}_{p-1},B\big]\cdots\big]\frac{\tau^{p-1}}{(p-1)!}\\
&\quad+\int_{0}^{\tau}\mathrm{d}\tau_2\ e^{(\tau-\tau_2)A_1}\underbrace{\big[A_1,\cdots,\big[A_1}_{p},B\big]\cdots\big]e^{-(\tau-\tau_2)A_1}\frac{\tau_2^{p-1}}{(p-1)!}.
\end{aligned}
\end{equation}
Using the abbreviation $\ad_{A_1}(B)=\big[A_1,B\big]$, we rewrite
\begin{equation}
\begin{aligned}
e^{\tau A_{1}}Be^{-\tau A_{1}}&=B+\ad_{A_1}(B)\tau+\cdots+\ad_{A_1}^{p-1}(B)\frac{\tau^{p-1}}{(p-1)!}\\
&\quad+\int_{0}^{\tau}\mathrm{d}\tau_2\ e^{(\tau-\tau_2)A_1}\ad_{A_1}^{p}(B)e^{-(\tau-\tau_2)A_1}\frac{\tau_2^{p-1}}{(p-1)!}.
\end{aligned}
\end{equation}
By the multiplication rule and the integration rule of \prop{order_cond_rule}, the last term has order
\begin{equation}
\int_{0}^{\tau}\mathrm{d}\tau_2\ e^{(\tau-\tau_2)A_1}\ad_{A_1}^{p}(B)e^{-(\tau-\tau_2)A_1}\frac{\tau_2^{p-1}}{(p-1)!}
=\OO{\tau^p}.
\end{equation}
This term cannot be canceled by the order condition, and we keep it in our expansion. The remaining terms corresponding to $1,\tau,\ldots,\tau^{p-1}$ are substituted back to the original conjugation of matrix exponentials.

We now consider the next layer of conjugation. We apply Taylor's theorem to the operators $e^{\tau A_2}Be^{-\tau A_2}$, $e^{\tau A_2}\ad_{A_1}(B)e^{-\tau A_2}$, \ldots, $e^{\tau A_2}\ad_{A_1}^{p-1}(B)e^{-\tau A_2}$ to order $p-1$, $p-2$, \ldots, $0$, respectively, obtaining
\begin{equation}
\begin{aligned}
e^{\tau A_{2}}Be^{-\tau A_{2}}&=B+\cdots+\ad_{A_2}^{p-1}(B)\frac{\tau^{p-1}}{(p-1)!}
+\int_{0}^{\tau}\mathrm{d}\tau_2\ e^{(\tau-\tau_2)A_2}\ad_{A_2}^{p}(B)e^{-(\tau-\tau_2)A_2}\frac{\tau_2^{p-1}}{(p-1)!},\\
e^{\tau A_{2}}\ad_{A_1}(B)e^{-\tau A_{2}}&=\ad_{A_1}(B)+\cdots+\ad_{A_2}^{p-2}\ad_{A_1}(B)\frac{\tau^{p-2}}{(p-2)!}\\
&\quad +\int_{0}^{\tau}\mathrm{d}\tau_2\ e^{(\tau-\tau_2)A_2}\ad_{A_2}^{p-1}\ad_{A_1}(B)e^{-(\tau-\tau_2)A_2}\frac{\tau_2^{p-2}}{(p-2)!},\\
&\vdotswithin{=}\\
e^{\tau A_2}\ad_{A_1}^{p-1}(B)e^{-\tau A_2}&=\ad_{A_1}^{p-1}(B)
+\int_{0}^{\tau}\mathrm{d}\tau_2\ e^{(\tau-\tau_2)A_2}\ad_{A_2}\ad_{A_1}^{p-1}(B)e^{-(\tau-\tau_2)A_2}.
\end{aligned}
\end{equation}
Combining with the result from the first layer, the Taylor remainders in the above equation have order
\begin{equation}
\begin{aligned}
\int_{0}^{\tau}\mathrm{d}\tau_2\ e^{(\tau-\tau_2)A_2}\ad_{A_2}^{p}(B)e^{-(\tau-\tau_2)A_2}\frac{\tau_2^{p-1}}{(p-1)!}
&=\OO{\tau^p},\\
\int_{0}^{\tau}\mathrm{d}\tau_2\ e^{(\tau-\tau_2)A_2}\ad_{A_2}^{p-1}\ad_{A_1}(B)e^{-(\tau-\tau_2)A_2}\frac{\tau_2^{p-2}}{(p-2)!}\tau
&=\OO{\tau^p},\\
&\vdotswithin{=}\\
\int_{0}^{\tau}\mathrm{d}\tau_2\ e^{(\tau-\tau_2)A_2}\ad_{A_2}\ad_{A_1}^{p-1}(B)e^{-(\tau-\tau_2)A_2}\frac{\tau^{p-1}}{(p-1)!}
&=\OO{\tau^p}.
\end{aligned}
\end{equation}
We keep these terms in our expansion and substitute the remaining ones back to the original conjugation of matrix exponentials.

We repeat this analysis for all the remaining layers of the conjugation of matrix exponentials. In doing so, we keep track of those terms with order $\OO{\tau^p}$, obtaining
\begin{equation}
\begin{aligned}
&	e^{\tau A_s}\cdots e^{\tau A_{2}}e^{\tau A_{1}}Be^{-\tau A_{1}}e^{-\tau A_{2}}\cdots e^{-\tau A_s} \\
	&=C_0+C_1\tau+\cdots+C_{p-1}\tau^{p-1}\\
&\quad +\sum_{k=1}^{s}\sum_{\substack{q_1+\cdots+q_k=p\\q_k\neq 0}}
e^{\tau A_s}\cdots e^{\tau A_{k+1}}\\
&\qquad\qquad\qquad\qquad \cdot\int_{0}^{\tau}\mathrm{d}\tau_2\ e^{\tau_2 A_k}\ad_{A_k}^{q_k}\cdots\ad_{A_1}^{q_1}(B)e^{-\tau_2 A_k}\cdot\frac{(\tau-\tau_2)^{q_k-1}\tau^{q_1+\cdots+q_{k-1}}}{(q_k-1)!q_{k-1}!\cdots q_1!}\\
&\qquad\qquad\qquad\qquad \cdot e^{-\tau A_{k+1}}\cdots e^{-\tau A_s}
\end{aligned}
\end{equation}
for some operators $C_0,C_1,\ldots,C_{p-1}$.
Due to the order condition, terms of order $1,\tau,\ldots,\tau^{p-1}$ will vanish in our final representation of the Trotter error.

We now bound the spectral norm of those terms with order $\OO{\tau^p}$. By the triangle inequality, we have an upper bound of
\begin{equation}
\begin{aligned}
&\ \sum_{k=1}^{s}\sum_{\substack{q_1+\cdots+q_k=p\\q_k\neq 0}}
\int_{0}^{\abs{\tau}}\mathrm{d}\tau_2\ \frac{(\abs{\tau}-\tau_2)^{q_k-1}\abs{\tau}^{q_1+\cdots+q_{k-1}}}{(q_k-1)!q_{k-1}!\cdots q_1!}
\norm{\ad_{A_k}^{q_k}\cdots\ad_{A_1}^{q_1}(B)}
e^{2\abs{\tau} \sum_{l=1}^{s}\norm{A_l}}\\
=&\ \sum_{k=1}^{s}\sum_{\substack{q_1+\cdots+q_k=p\\q_k\neq 0}}
\binom{p}{q_1\ \cdots\ q_k}\frac{\abs{\tau}^p}{p!}
\norm{\ad_{A_k}^{q_k}\cdots\ad_{A_1}^{q_1}(B)}
e^{2\abs{\tau} \sum_{l=1}^{s}\norm{A_l}}\\
=&\ \sum_{q_1+\cdots+q_s=p}
\binom{p}{q_1\ \cdots\ q_s}\frac{\abs{\tau}^p}{p!}
\norm{\ad_{A_s}^{q_s}\cdots\ad_{A_1}^{q_1}(B)}
e^{2\abs{\tau} \sum_{l=1}^{s}\norm{A_l}}\\
=&\ \acomm\big(A_s,\ldots,A_1,B\big)\frac{\abs{\tau}^p}{p!}e^{2\abs{\tau} \sum_{l=1}^{s}\norm{A_l}},
\end{aligned}
\end{equation}
where
\begin{equation}
\acomm\big(A_s,\ldots,A_1,B\big):=\sum_{q_1+\cdots+q_s=p}\binom{p}{q_1\ \cdots\ q_s}\norm{\ad_{A_s}^{q_s}\cdots\ad_{A_1}^{q_1}(B)}.
\end{equation}
This bound holds for arbitrary operators $A_1,A_2,\ldots,A_s$. When these operators are anti-Hermitian, we can tighten the above analysis by evaluating the spectral norm of a matrix exponential as $1$. We have therefore established:
\begin{theorem}[Commutator expansion of a conjugation of matrix exponentials]
	\label{thm:comm_exp_conj}
	Let $A_1,A_2,\ldots,A_s$ and $B$ be operators. Then the conjugation $e^{\tau A_s}\cdots e^{\tau A_{2}}e^{\tau A_{1}}Be^{-\tau A_{1}}e^{-\tau A_{2}}\cdots e^{-\tau A_s}$ $(\tau\in\R)$ has the expansion
	\begin{equation}
	e^{\tau A_s}\cdots e^{\tau A_{2}}e^{\tau A_{1}}Be^{-\tau A_{1}}e^{-\tau A_{2}}\cdots e^{-\tau A_s}
	=C_0+C_1\tau+\cdots+C_{p-1}\tau^{p-1}+\mathscr{C}(\tau).
	\end{equation}
	Here, $C_0,\ldots,C_{p-1}$ are operators independent of $\tau$. The operator-valued function $\mathscr{C}(\tau)$ is given by
	\begin{equation}
	\begin{aligned}
	\mathscr{C}(\tau):=\sum_{k=1}^{s}\sum_{\substack{q_1+\cdots+q_k=p\\q_k\neq 0}}
	&e^{\tau A_s}\cdots e^{\tau A_{k+1}}\\
	&\cdot\int_{0}^{\tau}\mathrm{d}\tau_2\ e^{\tau_2 A_k}\ad_{A_k}^{q_k}\cdots\ad_{A_1}^{q_1}(B)e^{-\tau_2 A_k}\cdot\frac{(\tau-\tau_2)^{q_k-1}\tau^{q_1+\cdots+q_{k-1}}}{(q_k-1)!q_{k-1}!\cdots q_1!}\\
	&\cdot e^{-\tau A_{k+1}}\cdots e^{-\tau A_s}.
	\end{aligned}
	\end{equation}
	Furthermore, we have the spectral-norm bound
	\begin{equation}
	\norm{\mathscr{C}(\tau)}\leq\acomm\big(A_s,\ldots,A_1,B\big)\frac{\abs{\tau}^p}{p!}e^{2\abs{\tau} \sum_{k=1}^{s}\norm{A_k}}
	\end{equation}
	for general operators and
	\begin{equation}
	\norm{\mathscr{C}(\tau)}\leq\acomm\big(A_s,\ldots,A_1,B\big)\frac{\abs{\tau}^p}{p!}
	\end{equation}
	when $A_k$ ($k=1,\ldots,s$) are anti-Hermitian, where
	\begin{equation}
	\acomm\big(A_s,\ldots,A_1,B\big)=\sum_{q_1+\cdots+q_s=p}\binom{p}{q_1\ \cdots\ q_s}\norm{\ad_{A_s}^{q_s}\cdots\ad_{A_1}^{q_1}(B)}.
	\end{equation}
\end{theorem}

We apply \thm{comm_exp_conj} to expand every conjugation of matrix exponentials of the error operators $\mathscr{T}(\tau)$ and $\mathscr{E}(\tau)$ into a finite series in $\tau$. After taking the linear combination, we obtain
\begin{equation}
\begin{aligned}
\mathscr{T}(\tau)&=T_0+T_1\tau+\cdots+T_{p-1}\tau^{p-1}+\mathscr{T}_p(\tau),\\
\mathscr{E}(\tau)&=E_0+E_1\tau+\cdots+E_{p-1}\tau^{p-1}+\mathscr{E}_p(\tau).
\end{aligned}
\end{equation}
The operator-valued functions $\mathscr{T}_p(\tau)$ and $\mathscr{E}_p(\tau)$ have order condition $\OO{\tau^p}$, whereas $T_0,\ldots,T_{p-1}$ and $E_0,\ldots,E_{p-1}$ are independent of $\tau$. By \lem{order_cond_deriv} and the order condition $\mathscr{T}(\tau),\ \mathscr{E}(\tau)\in \OO{\tau^p}$, we have
\begin{equation}
T_0=\cdots=T_{p-1}=E_0=\cdots=E_{p-1}=0,
\end{equation}
or equivalently,
\begin{equation}
\mathscr{T}(\tau)=\mathscr{T}_p(\tau),\qquad
\mathscr{E}(\tau)=\mathscr{E}_p(\tau).
\end{equation}
We then bound the spectral norm of $\mathscr{T}_p(\tau)$ and $\mathscr{E}_p(\tau)$ using \thm{comm_exp_conj}. This establishes the commutator scaling of Trotter error. We state the result below and leave the calculation details to \append{rep}.

\begin{restatable}[Trotter error with commutator scaling]{theorem}{thmcomm}
	\label{thm:trotter_error_comm_scaling}
	Let $H=\sum_{\gamma=1}^{\Gamma}H_\gamma$ be an operator consisting of $\Gamma$ summands and $t\geq 0$. Let $\mathscr{S}(t)=\prod_{\upsilon=1}^{\Upsilon}\prod_{\gamma=1}^{\Gamma}e^{ta_{(\upsilon,\gamma)}H_{\pi_{\upsilon}(\gamma)}}$ be a $p$th-order product formula. Define
	$\acommtilde=\sum_{\gamma_1,\gamma_2,\ldots,\gamma_{p+1}=1}^\Gamma\norm{\big[H_{\gamma_{p+1}},\cdots\big[H_{\gamma_2},H_{\gamma_1}\big]\big]}$. Then, the additive Trotter error and the multiplicative Trotter error, defined respectively by $\mathscr{S}(t)=e^{tH}+\mathscr{A}(t)$ and $\mathscr{S}(t)=e^{tH}(I+\mathscr{M}(t))$, can be asymptotically bounded as
	\begin{equation}
	\norm{\mathscr{A}(t)}=\cO{\acommtilde t^{p+1}e^{2t \Upsilon\sum_{\gamma=1}^{\Gamma}\norm{H_{\gamma}}}},\quad
	\norm{\mathscr{M}(t)}=\cO{\acommtilde t^{p+1}e^{2t \Upsilon\sum_{\gamma=1}^{\Gamma}\norm{H_{\gamma}}}}.
	\end{equation}
	Furthermore, if the $H_\gamma$ are anti-Hermitian, corresponding to physical Hamiltonians, we have
	\begin{equation}
	\norm{\mathscr{A}(t)}=\cO{\acommtilde t^{p+1}},\quad
	\norm{\mathscr{M}(t)}=\cO{\acommtilde t^{p+1}}.
	\end{equation}
\end{restatable}
\begin{corollary}[Trotter number with commutator scaling]
	\label{cor:trotter_number_comm_scaling}
	Let $H=\sum_{\gamma=1}^{\Gamma}H_\gamma$ be an operator consisting of $\Gamma$ summands with $H_\gamma$ anti-Hermitian and $t\geq 0$. Let $\mathscr{S}(t)=\prod_{\upsilon=1}^{\Upsilon}\prod_{\gamma=1}^{\Gamma}e^{ta_{(\upsilon,\gamma)}H_{\pi_{\upsilon}(\gamma)}}$ be a $p$th-order product formula. Define
	$\acommtilde=\sum_{\gamma_1,\gamma_2,\ldots,\gamma_{p+1}=1}^\Gamma\norm{\big[H_{\gamma_{p+1}},\cdots\big[H_{\gamma_2},H_{\gamma_1}\big]\big]}$. Then, we have $\norm{\mathscr{S}^r(t/r)-e^{tH}}=\cO{\epsilon}$, provided that
	\begin{equation}
	r=\cO{\frac{\acommtilde^{1/p}t^{1+1/p}}{\epsilon^{1/p}}}.
	\end{equation}
\end{corollary}

For any $\delta>0$, we can choose $p$ sufficiently large so that $1/p<\delta$. For this choice of $p$, we have $r=\cO{\acommtilde^\delta t^{1+\delta}/\epsilon^\delta}$. Therefore, the Trotter number scales as $r=\acommtilde^{o(1)}t^{1+o(1)}$ if we simulate with constant accuracy. To obtain the asymptotic complexity of the product-formula algorithm, it thus suffices to compute the quantity $\acommtilde=\sum_{\gamma_1,\gamma_2,\ldots,\gamma_{p+1}}\norm{\big[H_{\gamma_{p+1}},\cdots\big[H_{\gamma_2},H_{\gamma_1}\big]\big]}$, which can often be done by induction. We illustrate this by presenting a host of applications of our bound to simulating quantum dynamics (\sec{app_dqs}), local observables (\sec{app_local_obs}), and quantum Monte Carlo methods (\sec{app_qmc}).

Note that we did not evaluate the constant prefactor of our bound in \thm{trotter_error_comm_scaling}. Indeed, our proof involves inequality zooming that suffices to establish the correct asymptotic scaling but is likely loose in practice. For practical implementation, it is better to use \thm{comm_exp_conj}, which gives a concrete expression for the error operator. A general methodology to obtain error bounds with small constant factors is described in \append{pf2k}. In \sec{prefactor_pf12}, we show that our bound reduces to previous bounds for the Lie-Trotter formula \cite{Huyghebaert_1990,Suzuki85} and the second-order Suzuki formula \cite{DT10,WBCHT14,Kivlichan19,Jahnke2000}, which are known to be tight up to an application of the triangle inequality. We further provide numerical evidence in \sec{prefactor_pf2k} suggesting that our bound has a small prefactor for higher-order formulas as well.

\section{Applications}
\label{sec:app}

Our main result on the commutator scaling of Trotter error (\thm{trotter_error_comm_scaling}) uncovers a host of speedups of the product-formula approach. In this section, we give improved product-formula algorithms for digital quantum simulation (\sec{app_dqs}), simulating local observables (\sec{app_local_obs}), and quantum Monte Carlo methods (\sec{app_qmc}). We show that these results can nearly match or even outperform the best previous results for simulating quantum systems.

\subsection{Applications to digital quantum simulation}
\label{sec:app_dqs}
We now present applications of our bound to digital quantum simulation, including simulations of second-quantized electronic structure, $k$-local Hamiltonians, rapidly decaying long-range and quasilocal interactions, and clustered Hamiltonians. Throughout this section, we let $H$ be Hermitian, $t\geq 0$ be nonnegative, and we consider the real-time evolution $e^{-itH}$.

\medskip
\noindent\textbf{Second-quantized electronic structure.}
Simulating electronic-structure Hamiltonians is one of the most widely studied applications of digital quantum simulation. An efficient solution of this problem could help design and engineer new pharmaceuticals, catalysts, and materials \cite{BWMMNC18}. Recent studies have focused on solving this problem using more advanced simulation algorithms. Here, we demonstrate the power of product formulas for simulating electronic-structure Hamiltonians.

We consider the second-quantized representation of the electronic-structure problem. In the plane-wave dual basis, the electronic-structure Hamiltonian has the form \cite[Eq.\ (8)]{BWMMNC18}
\begin{equation}
\label{eq:plane_wave_dual}
\begin{aligned}
H&=\underbrace{\frac{1}{2n}\sum_{j,k,\nu}\kappa_{\nu}^2\cos[\kappa_{\nu}\cdot r_{k-j}]A_{j}^\dagger A_{k}}_{T}\\
&\quad\underbrace{-\frac{4\pi}{\omega}\sum_{j,\iota,\nu\neq 0}\frac{\zeta_\iota\cos[\kappa_{\nu}\cdot(\widetilde{r}_\iota-r_j)]}{\kappa_{\nu}^2}N_{j}}_{U}
\underbrace{+\frac{2\pi}{\omega}\sum_{\substack{j\neq k\\\nu\neq 0}}\frac{\cos[\kappa_{\nu}\cdot r_{j-k}]}{\kappa_{\nu}^2}N_{j}N_{k}}_{V},
\end{aligned}
\end{equation}
where $j,k$ range over all $n$ orbitals and $\omega$ is the volume of the computational cell. Following the assumptions of \cite{BWMMNC18,LW18}, we consider the constant density case where $n/\omega=\cO{1}$. Here, $\kappa_{\nu}=2\pi\nu/\omega^{1/3}$ are $n$ vectors of plane-wave frequencies, where $\nu$ are three-dimensional vectors of integers with elements in $[-n^{1/3},n^{1/3}]$; $r_j$ are the positions of electrons; $\zeta_\iota$ are nuclear charges such that $\sum_{\iota}|\zeta_\iota|=\cO{n}$; and $\widetilde{r}_\iota$ are the nuclear coordinates. The operators $A_j^\dagger$ and $A_k$ are electronic creation and annihilation operators, and $N_{j}=A_{j}^\dagger A_{j}$ are the number operators. The potential terms $U$ and $V$ are already diagonalized in the plane-wave dual basis. To further diagonalize the kinetic term $T$, we may switch to the plane-wave basis, which is accomplished by the fermionic fast Fourier transform $\mathrm{FFFT}$ \cite[Eq.\ (10)]{BWMMNC18}. We have
\begin{equation}
\begin{aligned}
H&=\mathrm{FFFT}^\dagger\underbrace{\bigg(\frac{1}{2}\sum_{\nu}\kappa_{\nu}^2 N_{\nu}\bigg)}_{\widetilde{T}}\mathrm{FFFT}
+U+V.
\end{aligned}
\end{equation}

To simulate the dynamics of such a Hamiltonian for time $t$, the current fastest algorithms are qubitization \cite{LC16,BGBWMPFN18} with $\tildecO{n^3t}$ gate complexity and a small prefactor, and the interaction-picture algorithm \cite{LW18} with complexity $\tildecO{n^2t}$ and a large prefactor. We show that higher-order product formulas can perform the same simulation with gate complexity $n^{2+o(1)}t^{1+o(1)}$. For the special case of the second-order Suzuki formula, this confirms a recent observation of Kivlichan et al.\ from numerical calculation \cite{Kivlichan19}.

Using the plane-wave basis for the kinetic operator and the plane-wave dual basis for the potential operators, we have that all terms in $\widetilde{T}$ and $U+V$ commute with each other, respectively. Then, we can decompose $e^{-it\tilde{T}}$ and $e^{-it(U+V)}$ into products of elementary matrix exponentials without introducing additional error, giving the product formula
\begin{equation}
\label{eq:plane_wave_dual_pf}
\begin{aligned}
&\ e^{-ita_{(\Upsilon,2)}T}e^{-ita_{(\Upsilon,1)}(U+V)}
\cdots
e^{-ita_{(1,2)}T}e^{-ita_{(1,1)}(U+V)}\\
=&\ \mathrm{FFFT}^\dagger e^{-ita_{(\Upsilon,2)}\widetilde{T}}\mathrm{FFFT}e^{-ita_{(\Upsilon,1)}(U+V)}
\cdots
\mathrm{FFFT}^\dagger e^{-ita_{(1,2)}\widetilde{T}}\mathrm{FFFT}e^{-ita_{(1,1)}(U+V)}.
\end{aligned}
\end{equation}
For practical implementation, we need to further exponentiate spin operators using a fermionic encoding, such as the Jordan-Wigner encoding. However, these implementation details do not affect the analysis of Trotter error and will thus be ignored in our discussion. The fermionic fast Fourier transform and the exponentiation of $\widetilde{T}$, $U$, and $V$ can all be implemented using the Jordan-Wigner encoding with complexity $\widetilde{O}(n)$ \cite{Ferris14,LW18}.

We compute the norm of $[H_{\gamma_{p+1}},\cdots[H_{\gamma_2},H_{\gamma_1}]]$, $H_\gamma\in\{T,U,V\}$ by induction. We show in \append{electron} that
\begin{equation}
\acommtilde=\sum_{\gamma_1,\gamma_2,\ldots,\gamma_{p+1}}\norm{\big[H_{\gamma_{p+1}},\cdots\big[H_{\gamma_2},H_{\gamma_1}\big]\big]}
=\cO{n^{p+1}}.
\end{equation}
\thm{trotter_error_comm_scaling} and \cor{trotter_number_comm_scaling} then imply that a Trotter number of $r=\cO{(nt)^{1+1/p}/\epsilon^{1/p}}$ suffices to simulate with accuracy $\epsilon$. Choosing $p$ sufficiently large, letting $\epsilon$ be constant, and implementing each Trotter step as in \cite{Ferris14,LW18}, we have the gate complexity
\begin{equation}
n^{2+o(1)}t^{1+o(1)}
\end{equation}
for simulating plane-wave electronic structure in second quantization.

\medskip
\noindent\textbf{$k$-local Hamiltonians.}
A Hamiltonian is $k$-local if it can be expressed as a linear combination of terms, each of which acts nontrivially on at most $k=\cO{1}$ qubits.
Such Hamiltonians, especially $2$-local ones, are ubiquitous in physics.
The first explicit quantum simulation algorithm by Lloyd was specifically developed for simulating $k$-local Hamiltonians \cite{Llo96} and later work provided more advanced approaches based on the linear-combination-of-unitary technique \cite{FractionalQuery14,BCCKS14,BCK15,LC17,LC16,LW18}. Here, we give an improved product-formula algorithm that can be advantageous over previous simulation methods.

We consider a $k$-local Hamiltonian acting on $n$ qubits
\begin{equation}
\label{eq:local_ham}
	H=\sum_{{j}_1,\ldots,{j}_k}H_{{j}_1,\ldots,{j}_k},
\end{equation}
where each $H_{{j}_1,\ldots,{j}_k}$ acts nontrivially only on qubits ${j}_1,\ldots,{j}_k$. We say $H_{{j}_1,\ldots,{j}_k}$ has support $\{{j}_1,\ldots,{j}_k\}$, denoting
\begin{equation}
	\supp\big(H_{{j}_1,\ldots,{j}_k}\big):=\{{j}_1,\ldots,{j}_k\}.
\end{equation}
We may assume that the summands are unitaries up to scaling and can be implemented with constant cost, for otherwise we expand them further with respect to the Pauli operators. The fastest previous approach to simulating a general $k$-local Hamiltonian is the qubitization algorithm by Low and Chuang \cite{LC16}, which has gate complexity $\tildecO{n^k\norm{H}_1t}$ where $\norm{H}_1=\sum_{{j}_1,\ldots,{j}_k}\norm{H_{{j}_1,\ldots,{j}_k}}$.

To compare with the product-formula algorithm, we need to analyze the nested commutators $[H_{\gamma_{p+1}},\cdots[H_{\gamma_2},H_{\gamma_1}]]$, where each $H_\gamma$ is some local operator $H_{{j}_1,\ldots,{j}_k}$. In order for this commutator to be nonzero, every operator must have support that overlaps with the support of operators from the inner layers. Using this idea, we estimate that
\begin{equation}
	\acommtilde=\sum_{\gamma_1,\gamma_2,\ldots,\gamma_{p+1}}\norm{\big[H_{\gamma_{p+1}},\cdots\big[H_{\gamma_2},H_{\gamma_1}\big]\big]}
	=\cO{\vertiii{H}_1^p\norm{H}_1},\label{eq:alpha_comm_one_norm}
\end{equation}
where $\vertiii{H}_1=\max_l\max_{{j}_l}\sum_{\substack{{j}_1,\ldots,{j}_{l-1},{j}_{l+1},\ldots,{j}_{k}}}\norm{H_{{j}_1,\ldots,{j}_k}}$ is the induced $1$-norm. \thm{trotter_error_comm_scaling} and \cor{trotter_number_comm_scaling} then imply that a Trotter number of $r=\cO{\vertiii{H}_1\norm{H}_1^{1/p}t^{1+1/p}/\epsilon^{1/p}}$ suffices to simulate with accuracy $\epsilon$. Choosing $p$ sufficiently large, letting $\epsilon$ be constant, and implementing each Trotter step with $\Th{n^k}$ gates, we have the total gate complexity
\begin{equation}
n^k\vertiii{H}_1\norm{H}_1^{o(1)}t^{1+o(1)}
\end{equation}
for simulating a $k$-local Hamiltonian $H$. See \append{local} for more details.

We know from \sec{prelim_notation} that the norm inequality $\vertiii{H}_1\leq \norm{H}_1$ always holds. In fact, the gap between these two norms can be significant for many $k$-local Hamiltonians. As an example, we consider $n$-qubit power-law interactions $H = \sum_{\vec{i},\vec{j}\in \Lambda} H_{\vec{i},\vec j}$ with exponent $\alpha$ \cite{Tran18}, where $\Lambda\subseteq \mathbb R^d$ is a $d$-dimensional square lattice, $H_{\vec i,\vec j}$ is an operator supported on two sites $\vec i,\vec j\in \Lambda$, and
\begin{align}
\label{eq:power-law-def}
\norm{H_{\vec i,\vec j}}\leq
\begin{cases}
1, 					&\text{ if } \vec i = \vec j,\\
\frac{1}{\norm{\vec i-\vec j}_2^\alpha},\quad &\text{ if } \vec i\neq \vec j.
\end{cases}
\end{align}
Examples of such systems include those that interact via the Coulomb interactions ($\alpha = 1$), the dipole-dipole interactions ($\alpha = 3$), and the van der Waals interactions ($\alpha = 6$).
It is straightforward to upper bound the induced $1$-norm
\begin{equation}
	\vertiii{H}_1=
	\begin{cases}
	\cO{n^{1-\alpha/d}},\quad & \text{for }0\leq\alpha<d,\\
	\cO{\log n}, & \text{for }\alpha=d,\\
	\cO 1, & \text{for }\alpha>d,
	\end{cases}
\end{equation}
whereas the $1$-norm scales like
\begin{equation}
\norm{H}_1=
\begin{cases}
\cO{n^{2-\alpha/d}},\quad & \text{for }0\leq\alpha<d,\\
\cO{n\log n}, & \text{for }\alpha=d,\\
\cO n, & \text{for }\alpha>d.
\end{cases}
\end{equation}
Thus the product-formula algorithm has gate complexity
\begin{align}
\label{eq:power-law-pf}
g_\alpha
=\begin{cases}
n^{3-\frac{\alpha}{d}+o(1)} t^{1+o(1)} & \text{ for }0\leq \alpha<d,\\
n^{2+o(1)} t^{1+o(1)} & \text{ for }\alpha\geq d,
\end{cases}
\end{align}
which has better $n$-dependence than the qubitization approach \cite{LC16}. We give further calculation details in \append{power-law}.

\medskip
\noindent\textbf{Rapidly decaying power-law and quasilocal interactions.}
We now consider $d$-dimensional power-law interactions $1/\distance^\alpha$ with exponent $\alpha>2d$ and interactions that decay exponentially with distance.
Although these Hamiltonians can be simulated using algorithms for $k$-local Hamiltonians, more efficient methods exist that exploit the locality of the systems \cite{Tran18}. We show that product formulas can also leverage locality to provide an even faster simulation.

We first consider an $n$-qubit $d$-dimensional power-law Hamiltonian $H = \sum_{i,j\in \Lambda} H_{\vec{i},\vec j}$ with exponent $\alpha>2d$. Such a Hamiltonian represents a rapidly decaying long-range system that becomes nearest-neighbor interacting in the limit $\alpha\rightarrow\infty$. For $\alpha>2d$, the state-of-the-art simulation algorithm decomposes the evolution based on the Lieb-Robinson bound with gate complexity $\tildecO{(nt)^{1+2d/(\alpha-d)}}$ \cite{Tran18}.
We give an improved approach using product formulas which has gate complexity $(nt)^{1+d/(\alpha-d)+o(1)}$.

The idea of our approach is to simulate a truncated Hamiltonian $\widetilde H = \sum_{\norm{\vec i-\vec j}_2\leq \ell} H_{\vec i,\vec j}$ by taking only the terms $H_{\vec i,\vec j}$ where $\Norm{\vec i-\vec j}_2$ is not more than $\ell$, a parameter that we determine later. The resulting $\widetilde{H}$ is a $2$-local Hamiltonian with $1$-norm $\Norm{\widetilde{H}}_1=\cO{n}$ and induced $1$-norm $\Vertiii{\widetilde{H}}_1=\cO{1}$. \thm{trotter_error_comm_scaling} and \cor{trotter_number_comm_scaling} then imply that a Trotter number of $r=\cO{n^{1/p}t^{1+1/p}/\epsilon^{1/p}}$
suffices to simulate with accuracy $\epsilon$. Choosing $p$ sufficiently large, letting $\epsilon$ be constant, and implementing each Trotter step with $\cO{n\ell^d}$ gates, we have the total gate complexity
$\ell^d(nt)^{1+o(1)}$ for simulating $\widetilde{H}$.

We know from \cor{time_ordered_distance_bound} that the approximation of $\exp(-iHt)$ by $\exp(-i\widetilde H t)$ has error
\begin{align}
\norm{e^{-itH} - e^{-it\widetilde{H}}}  = \cO{\norm{H-\widetilde{H}}t},
\end{align}
where $\Norm{H-\widetilde{H}} = \cO{n/\ell^{\alpha-d}}$ for all $\alpha>2d$. To make this at most $\cO{\epsilon}$, we choose the cutoff $\ell = \Th{\left({nt}/{\epsilon}\right)^{1/(\alpha-d)}}$.
Note that we require $nt\geq \epsilon$ and $t\leq \epsilon n^{\alpha/d-2}$ so that $n^{1/d}\geq \ell \geq 1$. This implies the gate complexity
\begin{equation}
(nt)^{1+d/(\alpha-d)+o(1)}, \label{eq:power-law-gate-count-weak}
\end{equation}
which is better than the state-of-the-art algorithm based on Lieb-Robinson bounds \cite{Tran18}. We leave the calculation details to \append{power-law}.

We also consider interactions that decay exponentially with the distance $\distance$ as $e^{-\beta \distance}$:
\begin{equation}
\label{eq:quasi-local-def}
\norm{H_{\vec i,\vec j}}\leq e^{-\beta\norm{\vec{i}-\vec{j}}_2},
\end{equation}
where $\beta>0$ is a constant. Although such interactions are technically long range, their fast decay makes them quasilocal for most applications in physics. Our approach to simulating such a quasilocal system is similar to that for the rapidly decaying power-law Hamiltonian, except we choose the cutoff $\ell = \Theta(\log(nt/\epsilon))$, giving a product-formula algorithm with gate complexity
\begin{equation}
(nt)^{1+o(1)}.
\end{equation}
See \append{power-law} for further details.

Our result for quasilocal systems is asymptotically the same as a recent result for nearest-neighbor Hamiltonians \cite{CS19}. For rapidly decaying power-law systems, we reproduce the nearest-neighbor case \cite{CS19} in the limit $\alpha\rightarrow\infty$.

\medskip
\noindent\textbf{Clustered Hamiltonians.}
We now consider the application of our theory to simulating clustered Hamiltonians \cite{PHOW19}. Such systems appear naturally in the study of classical fragmentation methods and quantum mechanics/molecular mechanics methods for simulating large molecules. Peng, Harrow, Ozols, and Wu recently proposed a hybrid simulator for clustered Hamiltonians \cite{PHOW19}. Here, we show that the performance of their simulator can be significantly improved using our Trotter error bound.

Let $H$ be a Hamiltonian acting on $n$ qubits. Following the same setting as in \cite{PHOW19}, we assume that each term in $H$ acts on at most two qubits with spectral norm at most one, and each qubit interacts with at most a constant number $d'$ of other qubits. We further assume that the qubits are grouped into multiple parties and write
\begin{equation}
H=A+B=\sum_l H_l^{(1)}+\sum_l H_l^{(2)},\quad \forall l:\norm{H_l^{(1)}},\norm{H_l^{(2)}}\leq 1,
\end{equation}
where terms in $A$ act on qubits within a single party and terms in $B$ act between two different parties.

The key step in the approach of Peng et al.\ is to group the terms within each party in $A$ and simulate the resulting Hamiltonian. This is accomplished by applying product formulas to the decomposition
\begin{equation}
\label{eq:clustered_ham}
H=A+\sum_l H_l^{(2)}.
\end{equation}
Using the first-order Lie-Trotter formula, Ref.\ \cite{PHOW19} chooses the Trotter number
\begin{equation}
r=\cO{\frac{h_B^2 t^2}{\epsilon}}
\end{equation}
to ensure that the error of the decomposition is at most $\epsilon$, where $h_B=\sum_l\Norm{H_l^{(2)}}$ is the interaction strength. Here, we use \thm{trotter_error_comm_scaling} and \cor{trotter_number_comm_scaling} to show that it suffices to take
\begin{equation}
r=\cO{\frac{d'^{\frac{1+p}{2}} h_B^{\frac{1}{p}} t^{1+\frac{1}{p}}}{\epsilon^{\frac{1}{p}}}}=\cO{\frac{h_B^{1/p} t^{1+1/p}}{\epsilon^{1/p}}}
\end{equation}
using a $p$th-order product formula
\begin{equation}
\label{eq:clustered_pf}
\mathscr{S}(t)=e^{-ita_\Upsilon A}\prod_{l}e^{-ita_{(\Upsilon,l)}H_l^{(2)}}\cdots e^{-ita_1A}\prod_{l}e^{-ita_{(1,l)}H_l^{(2)}}.
\end{equation}
This improves the analysis of \cite{PHOW19} for the first-order formula and extends the result to higher-order cases. Details can be found in \append{cluster}.

The hybrid simulator of \cite{PHOW19} has runtime $2^{\cO{r\cdot\mathrm{cc}(g)}}$, where $r$ is the Trotter number and $\mathrm{cc}(g)$ is the contraction complexity of the interaction graph $g$ between the parties. Our improved choice of $r$ thus provides a dramatic improvement.

\subsection{Applications to simulating local observables}
\label{sec:app_local_obs}
In this section, we consider quantum simulation of local observables. Our goal is to simulate the time evolution $\scA(t): = e^{itH} A e^{-itH}$ of an observable $A$, where the support $\supp(A)$ can be enclosed in a $d$-dimensional ball of constant radius $\distance_0$ on a $d$-dimensional lattice $\Lambda\subseteq \mathbb R^d$. Throughout this section, we consider power-law interactions with exponent $\alpha>2d$ and we assume $t\geq 0$.

Although a local observable can be simulated by simulating the full dynamics as in \sec{app_dqs}, this is not the most efficient approach. Instead, we use product formulas to give an algorithm whose gate complexity is independent of the system size for a short-time evolution; this complexity is much smaller than the cost of full simulation. As a byproduct, we prove a Lieb-Robinson-type bound for power-law Hamiltonians that nearly matches a recent bound of Tran et al.\ \cite{Tran18}.

\medskip
\noindent\textbf{Locality of time-evolved observables.}
Our approach is to approximate the evolution $\scA(t) = e^{itH} A e^{-itH}$ of the local observable $A$ by $e^{it\Hlc}Ae^{-it\Hlc}$, where $\Hlc$ is a Hamiltonian supported within a light cone originating from $A$ at time $0$. Although this can be achieved using Lieb-Robinson bounds \cite{Tran18}, we give a direct construction using product formulas.

Without loss of generality, we assume that the Hamiltonian $H$ is supported on an infinite lattice.\footnote{For Hamiltonians that are supported on finite lattices, we simply add trivial terms supported outside the lattices.}
The idea behind our approach is as follows. We first truncate the original Hamiltonian to obtain $\Htrunc$. We group the terms of $\Htrunc$ into $d$-dimensional shells based on their distance to the observable and use a product formula $\Strunc(t)$ to approximate the evolution. Unlike in \sec{app_dqs}, we choose a specific ordering of the summands so that the majority of the terms in $\Strunc(t)$ can be commuted through the observable to cancel their counterparts in $\Strunc^\dag(t)$. We define the reduced product formula $\Sreduce(t)$ as in \fig{local-sim-demo} by collecting all the remaining terms in $\Strunc(t)$. This gives an accurate approximation to a short-time evolution. For larger times, we divide the evolution into $r$ Trotter steps and apply the above approximation within each step. We reverse this procedure within the light cone to obtain $\Slc(t)$, which simulates the desired Hamiltonian $\Hlc$. See \fig{local-sim-step} for a step-by-step illustration of this approach.

We consider a general observable $B$ and we assume that $\mathcal S(B)$---the support of $B$---is a $d$-dimensional ball of radius $y_0$ centered on the origin.
We analyze $B$ as opposed to the original observable $A$ so that our argument not only applies to the first Trotter step, but also to later steps where $A$ is evolved and its support is expanded.
We denote by $\dist(\vec i,\mathcal S(B)):=\inf_{\vec{j}\in\mathcal S(B)}\norm{\vec{i}-\vec{j}}_2$ the distance between $\vec{i}$ and $\mathcal{S}(B)$,
by $\mathcal B_{y} := \{\vec i\in \Lambda:\dist(\vec i,\mathcal S(B))\leq y\}$ a ball of radius $y+y_0$ centered on $\mathcal S(B)$,
and by $\Delta\mathcal B_{\gamma\ell} = \mathcal B_{\gamma\ell}\setminus\mathcal B_{(\gamma-1)\ell}$ the shell containing sites between distance $(\gamma-1)\ell$ and $\gamma\ell$ from $\mathcal S(B)$, where $\ell\geq 1$ is a parameter to be chosen later and $\gamma\in\mathbb N$ is a nonnegative integer---with the convention that $\mathcal B_{-\ell}=\varnothing$ so that $\Delta \mathcal B_0 = \mathcal B_0 = \mathcal S(B)$.
We illustrate the sets $\mathcal B_{\gamma\ell}$ and $\Delta\mathcal B_{\gamma\ell}$ for several values of $\gamma$ in \fig{local-sim-demo}.

\begin{figure}[t]
	\centering
	\includegraphics[width=0.75\textwidth]{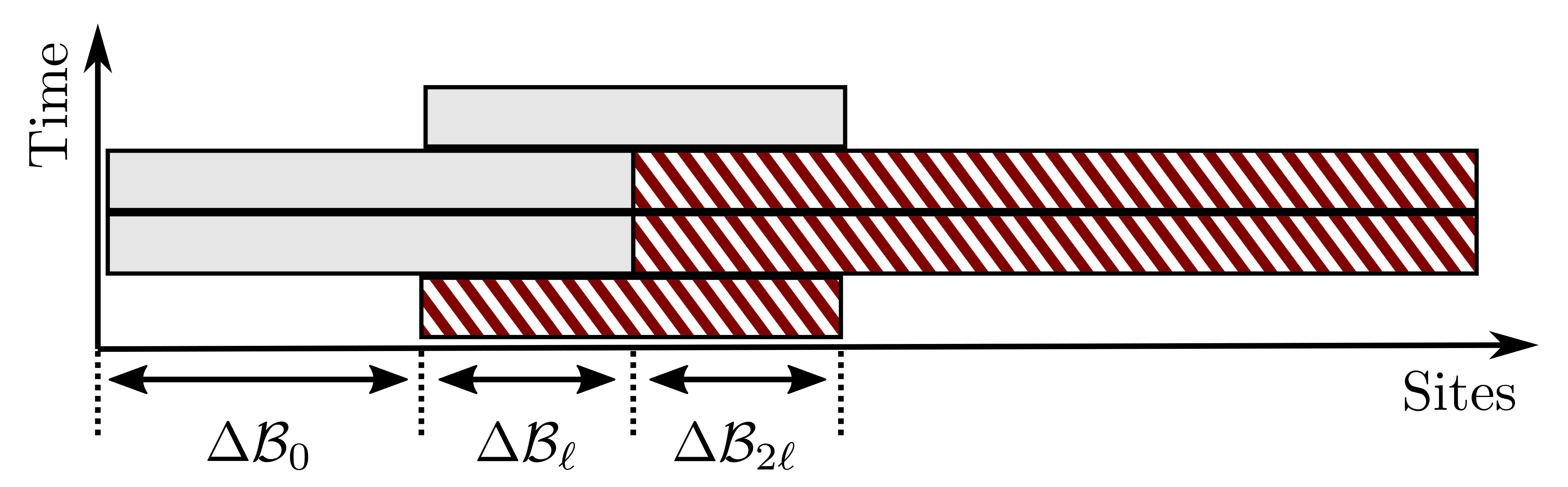}
	\caption{Demonstration of the second-order product formula for simulating the evolution of an observable $B$ supported on $\mathcal S(B) = \Delta\mathcal B_0$.
		Each rectangle represents a unitary supported on the sites covered by the width of the rectangle.
		The evolution unitary $e^{-itH}$ is decomposed using the second-order product formula into $\Upsilon = 2$ stages.
		Each stage is a sequence of $\Gamma = 3$ matrix exponentials generated by Hamiltonian terms supported on parts of the system.
		Some of these unitaries (red shaded rectangles) can subsequently be commuted through $B$ in the expression $\Strunc^\dag(t) B \Strunc(t)$ to cancel out with their Hermitian conjugates.
		As a result, the time-evolved version of $B$ can be effectively described by the remaining unitaries (light-gray rectangles).
	}
	\label{fig:local-sim-demo}
\end{figure}

Starting from the power-law Hamiltonian $H = \sum_{\vec i,\vec j\in\Lambda} H_{\vec i,\vec j}$, we group terms based on their distance to the observable $B$ and define
\begin{align}
&H_1 = \sum_{\vec i,\vec j\in \mathcal B_{\ell}} H_{\vec i,\vec j},\\
&H_\gamma = \sum_{\vec i,\vec j\in\Delta \mathcal B_{\gamma\ell}} H_{\vec i,\vec j}+ \sum_{\substack{
		\vec i\in\Delta \mathcal B_{(\gamma-1)\ell}\\
		\vec j\in\Delta \mathcal B_{\gamma\ell}}
} H_{\vec i,\vec j} \quad\text{ for }\gamma=2,\dots,\Gamma-1,\\
&H_{\Gamma} = \sum_{\vec i,\vec j\notin\mathcal B_{(\Gamma-2)\ell}} H_{\vec i,\vec j}
\end{align}
with constant $\Gamma$ to be chosen later. In this construction, all $H_\gamma$ with even $\gamma$ commute and all $H_\gamma$ with odd $\gamma$ commute.
We consider the truncated Hamiltonian
\begin{equation}
	\Htrunc = \sum_{\gamma=1}^{\Gamma} H_\gamma
\end{equation}
instead of $H$, which incurs a truncation error of
\begin{equation}
	\epsilon_1=\cO{\norm{e^{-itH}-e^{-it\Htrunc}}}=\cO{\norm{H-\Htrunc} t}= \cO{\frac{(y_0+\Gamma\ell)^{d-1}t}{\ell^{\alpha-d-1}}}.
\end{equation}
See \append{local-obs} for proof details.

Next, we simulate the evolution $e^{-it\Htrunc}$ using the $p$th-order product formula [See \cref{eq:pf} and \cref{fig:local-sim-demo}]:
\begin{align}
\Strunc(t) = \prod_{\upsilon=1}^{\Upsilon} \prod_{\gamma=1}^{\Gamma} e^{-i t a_{(\upsilon,\gamma) }H_{\pi_{\upsilon}(\gamma)}},\label{eq:loc_pf}
\end{align}
where we put additional constraints on the permutation $\pi_\nu$:
\begin{align}
\pi_{\upsilon}(1,2,3,4,5,6,\ldots) = \begin{cases}
(2,4,6,\dots,1,3,5,\dots),\quad& \text{ if } \upsilon \text{ is odd},\\
(1,3,5,\dots,2,4,6,\dots),& \text{ if } \upsilon \text{ is even}.
\end{cases}
\end{align}
Such a permutation can be realized using Suzuki's original construction~\cite{Suz91} and taking into account that $\commm{H_{2k},H_{2k'}}=0$ and $\commm{H_{2k+1},H_{2k'+1}}=0$ for all $k,k'$.
Using \thm{trotter_error_comm_scaling}, we show in \append{local-obs} that the error of approximating $e^{-it\Htrunc}$ by $\Strunc(t)$ is
\begin{align}
\epsilon_2 &= \norm{e^{-it\Htrunc}-\Strunc(t)} \nonumber\\
&= \cO{\sum_{\gamma_1,\dots,\gamma_{p+1}=1}^{\Gamma} \norm{\commm{H_{\gamma_{p+1}},\dots,\commm{H_{\gamma_2},H_{\gamma_1}}}}t^{p+1}}
=\cO{(y_0+\Gamma\ell)^{d-1}\ell t^{p+1}}.\label{eq:epsilon2}
\end{align}

Note that the Hamiltonian terms $H_{\gamma}$ (for $\gamma\geq 2$) commute with $B$.
Therefore, the exponentials in $\Strunc(t)$ corresponding to these terms can be commuted through $B$ to cancel with their counterparts in $\Strunc(t)^\dag$. By choosing the constant $\Gamma = \Upsilon+1$, we have
\begin{align}
\Strunc^\dag(t) B \Strunc(t) = \Sreduce^\dag (t) B \Sreduce(t),
\end{align}
where
\begin{align}
\Sreduce(t) = \prod_{\upsilon=1}^{\Upsilon} \prod_{\gamma=1}^{\upsilon} e^{-i t a_{(\upsilon,\pi_{\upsilon}^{-1}(\gamma)) }H_{\gamma}}.\label{eq:double_tilde_S}
\end{align}
We call  $\Sreduce(t)$  the reduced product formula. This approximates the evolution $e^{-itH}$ of local observable $B$ with error
\begin{equation}
\begin{aligned}
\norm{e^{itH} Be^{-itH}-\Sreduce^\dag(t) B \Sreduce(t)}
&\leq\norm{e^{itH} Be^{-itH}-e^{it\Htrunc} Be^{-it\Htrunc}}\\
&\quad +\norm{e^{it\Htrunc} Be^{-it\Htrunc}-\Strunc^\dagger(t)B\Strunc(t)}\\
&\quad +\norm{\Strunc^\dagger(t)B\Strunc(t)-\Sreduce^\dagger(t)B\Sreduce(t)}\\
&\leq 2\norm B (\epsilon_1 +\epsilon_2)+0\\
&= \cO{\norm{B}t(y_0+\Gamma\ell)^{d-1}\left(\frac{1}{\ell^{\alpha-d-1}} + \ell t^{p}\right)}.
\end{aligned}
\end{equation}

The above decomposition is accurate for a short-time evolution. For larger times, we divide the simulation into $r$ Trotter steps and apply this decomposition within each step. We analyze the error in a similar way as above, except that $B$ is defined by applying the reduced product formula to the observable $A$. Since the spectral norm is invariant under unitary transformations, we have $\norm{B}=\norm{A}=\cO{1}$. Another difference is that the support of the observable is expanded by $\Gamma\ell$ after each Trotter step; i.e., we set $y_0$ to be $\distance_0$, $\distance_0+\Gamma\ell$,\ldots, and $\distance_0+r\Gamma\ell$. Using the triangle inequality, we bound the error of the reduced product formula by
\begin{equation}
	\cO{t(\distance_0+r\Gamma\ell)^{d-1}\left(\frac{1}{\ell^{\alpha-d-1}} + \ell \frac{t^p}{r^{p}}\right)}.
\end{equation}

We now apply the above procedure in the reverse direction, but only to Hamiltonian terms within the light cone, incurring a truncation error at most $\epsilon_1$ and a Trotter error at most $\epsilon_2$. This replaces $\Sreduce(t)$ by $\Slc(t)$, the product formula that simulates the Hamiltonian $\Hlc$ whose terms have distance at most $r\Gamma\ell$ to the local observable $A$. See \fig{local-sim-step} for a step-by-step illustration of this approach. We analyze the error in a similar way as above, establishing the following result on evolving local observables.

\begin{figure}[t]
	\centering
		\includegraphics[width=0.95\textwidth]{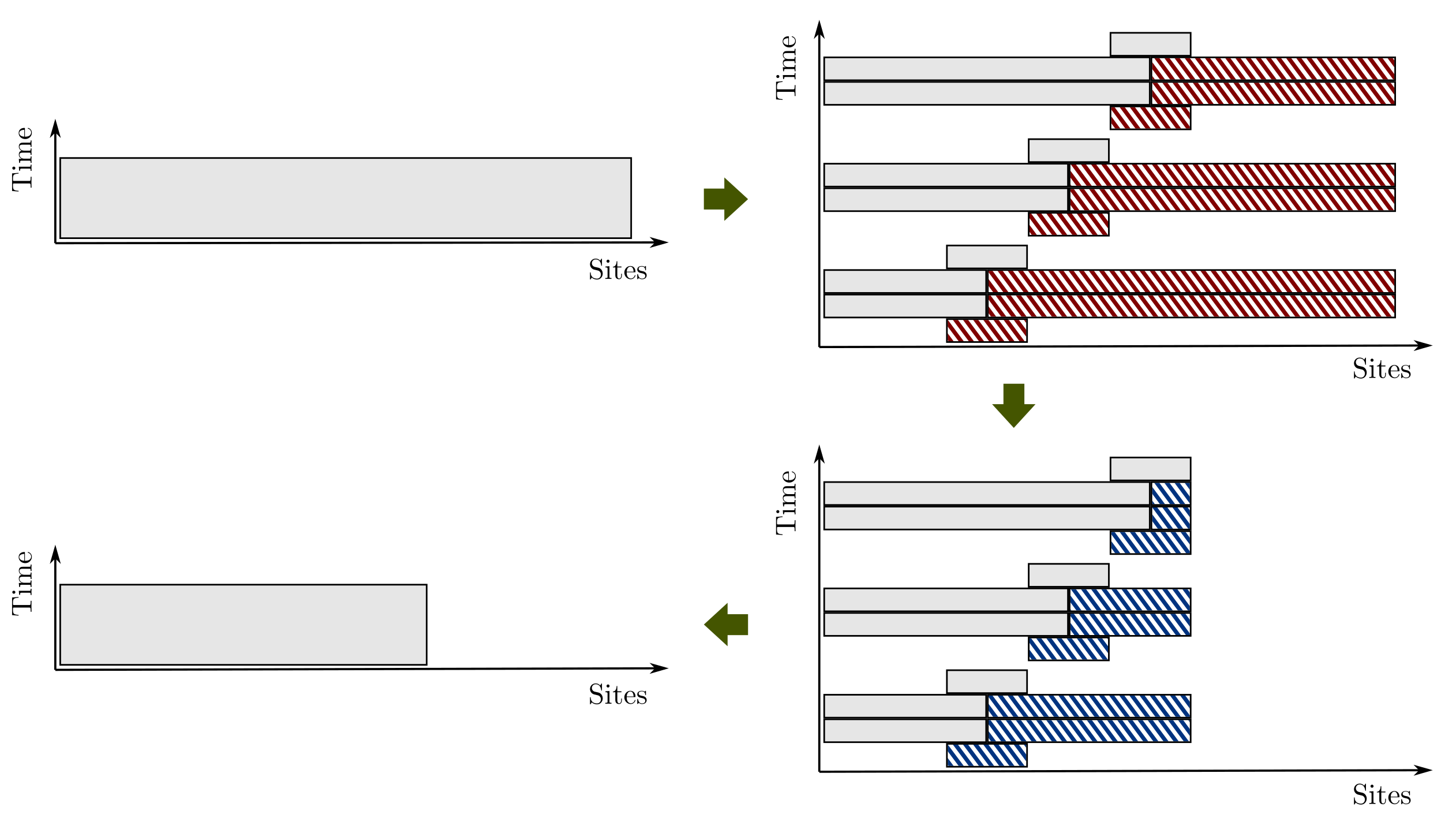}
	\caption{Construction of the Hamiltonian $\Hlc$ within the light cone such that $e^{itH} Ae^{-itH}\approx e^{it\Hlc} A e^{-it\Hlc}$. Each rectangle represents a unitary supported on the sites covered by the width of the rectangle. Specifically, the rectangle in the top-left panel represents the evolution $e^{-itH}$ that we want to decompose. We divide the evolution into $r$ steps. Within each step, we truncate the Hamiltonian to $\Htrunc$ and decompose its evolution using a product formula $\Strunc$. We commute certain matrix exponentials in $\Strunc$ (represented by red shaded rectangles) through the observable to cancel their counterpart, obtaining $\Sreduce$ in the top-right panel. We reverse this procedure within the light cone to construct $\Slc$ in the bottom-right panel, which approximates $e^{-it\Hlc}$ as illustrated in the bottom-left panel.}
	\label{fig:local-sim-step}
\end{figure}

\begin{restatable}[Product-formula decomposition of evolutions of local observables]{reprop}{proplocalobs}
	\label{prop:pf-local-obs}
	Let $\Lambda\subseteq \mathbb R^d$ be a $d$-dimensional square lattice. Let $H$ be a power-law Hamiltonian \eq{power-law-def} with exponent $\alpha>2d$ and $A$ be an observable with support enclosed in a $d$-dimensional ball of constant radius $x_0$. Construct the Hamiltonian $\Hlc$ as above using $p$th-order $\Upsilon$-stage product formulas $\Strunc(t)$, $\Sreduce(t)$, and $\Slc(t)$. Then, the support of $\Hlc$ has radius $\distance_0+r\Gamma\ell$ and
	\begin{equation}
	\label{eq:light_cone_ham}
	\begin{aligned}
	\norm{e^{itH} Ae^{-itH}-e^{it\Hlc} A e^{-it\Hlc}}
	&= \cO{t(\distance_0+r\Gamma\ell)^{d-1}\left(\frac{1}{\ell^{\alpha-d-1}} + \ell \frac{t^p}{r^{p}}\right)},
	\end{aligned}
	\end{equation}
	where the positive integer $\ell$ is a parameter and $\Gamma=\Upsilon+1$ is constant.
\end{restatable}

\medskip
\noindent\textbf{Gate complexity of simulating local observables.}
We now analyze the gate complexity of simulating local observables using the decomposition in \prop{pf-local-obs}. Assuming the support $\supp(A)$ has constant radius $\distance_0 = \cO{1}$ and $\Gamma = \cO{1}$, we simplify the error bound in \eq{light_cone_ham} to
\begin{align}
\norm{e^{itH} Ae^{-itH}-e^{it\Hlc} A e^{-it\Hlc}}
=\cO{t(r\ell)^{d-1}\left(\frac{1}{\ell^{\alpha-d-1}}+\frac{\ell t^p}{r^{p}}\right)}.
\end{align}
To minimize the error, we choose the cutoff $\ell = \Th{\left(\frac{r}{t}\right)^{\frac{p}{\alpha-d}}}$, which is larger than $1$ provided $r\geq t$ (and recall that we assume $\alpha>2d$, so in particular, $\alpha > d$). With this choice of $\ell$, the error becomes
\begin{align}
\cO{tr^{d-1}\left(\frac{ t}{r}\right)^{\frac{p(\alpha-2d)}{\alpha-d}}}
= \cO{\frac{t^{\frac{p(\alpha-2d)+\alpha-d}{\alpha-d}}}{r^{\frac{p(\alpha-2d)-(\alpha-d)(d-1)}{\alpha-d}}}}.
\end{align}
We then choose an appropriate Trotter number $r$ as detailed in \append{local-obs} and find that
\begin{align}
g_\alpha
={
	t^{\left(1+d\frac{\alpha -d}{\alpha -2 d}\right) \left(1+\frac{d}{\alpha -d}\right)+o(1)}
}
\end{align}
gates suffice to simulate a local observable with constant accuracy.
The gate count is independent of the system size and thus less than the cost of simulating the full dynamics \eq{power-law-pf} when the system size is $n=\Om{t^{d(\alpha-d)/(\alpha-2d)}}$.
However, in contrast to the simulation of $e^{-itH}$ where the asymptotic error scaling is robust against the reordering of Hamiltonian terms, we obtain a smaller error for simulating $\scA(t)$ by defining product formulas with a special ordering that preserves the locality of the simulated system.

Additionally, in the limit $\alpha\rightarrow\infty$ which corresponds to nearest-neighbor interactions, we have the gate count
\begin{align}
g_\infty
=
{t^{d+1+o(1)}}.
\end{align}
This has a clear physical intuition:
it is (nearly) proportional to the space-time volume $t^{d+1}$ inside a linear light cone generated by the evolution.

\medskip
\noindent\textbf{Lieb-Robinson-type bound for power-law Hamiltonians.}
The Lieb-Robinson bounds---first derived for nearest-neighbor interactions~\cite{LR72} and subsequently generalized to power-law systems~\cite{HK06,NOS06,NS06,GFMG14,FGCG15,SWK15,Tran18,CL19}---have found numerous applications in physics, including designing new algorithms for quantum simulations~\cite{HHKL18,Tran18}.
They bound the speed at which a local disturbance spreads in quantum systems.
Here, we show that the decomposition of \prop{pf-local-obs} constructed using product formulas also implies a Lieb-Robinson-type bound for power-law Hamiltonians.

The subject of the Lieb-Robinson bounds is usually the commutator norm
\begin{align}
\mathcal C(t,\rho) = \norm{\commm{e^{itH} Ae^{-itH},B}},
\end{align}
where $A,B$ are two operators whose supports have distance
\begin{equation}
	\dist(\mathcal S(A),\mathcal S(B)) = \inf\limits_{\vec{i}\in \mathcal S(A),\vec{j}\in\mathcal S(B)}\norm{\vec{i}-\vec{j}}_2 = \rho
\end{equation}
and $e^{-itH}$ is the time evolution unitary generated by a power-law Hamiltonian $H$.
Our above discussion shows that $e^{itH} Ae^{-itH}$ is approximately $e^{it\Hlc} Ae^{-it\Hlc}$, which is supported on a ball of radius $\distance = \cO{r\ell} = \cO{r\left(r/t\right)^{\frac{p}{\alpha-d}}}$ centered on $\mathcal S(A)$.
By choosing $r = \Th{{\rho^{\frac{\alpha-d}{\alpha-d+p}}}{t^{\frac{p}{\alpha-d+p}}}}$ so that $\distance<\rho$, we make $e^{it\Hlc} Ae^{-it\Hlc}$ commute with $B$ and therefore $\mathcal C(t,\rho)$ is small. More precisely,
\begin{align}
\mathcal C(t,r)
=\cO{\frac{t^{\frac{p(\alpha-2d)+\alpha-d}{\alpha-d}}}{r^{\frac{p(\alpha-2d)-(\alpha-d)(d-1)}{\alpha-d}}}}
=\cO{\frac
	{t^{\frac{(p+1) (\alpha-d)}{\alpha -d+p}}}
	{\rho^{\frac{p(\alpha-2d)-(\alpha-d)(d-1)}{\alpha-d+p}}}}
={\frac
	{t^{\alpha-d+o(1)}}
	{\rho^{\alpha-2d+o(1)}}}.
\end{align}
Note that we have implicitly assumed that $\rho^{\frac{\alpha-d}{\alpha-d+p}}t^{\frac{p}{\alpha-d+p}}\geq 1$ so that we can choose $r\geq 1$.
The bound implies a light cone $t \gtrsim \rho^{\frac{\alpha-2d}{\alpha-d}+o(1)}$, which can be made arbitrarily close to the light cone $t \gtrsim \rho^{\frac{\alpha-2d}{\alpha-d}}$ of the recent bound in Ref.~\cite{Tran18} for all values of $d$.\footnote{More recent bounds~\cite{CL19,Kuwahara19} provide tighter light cones than in Tran et al.~\cite{Tran18} for $\alpha>2d+1$.
}

\subsection{Applications to quantum Monte Carlo simulation}
\label{sec:app_qmc}
We now apply our result to improving the performance of quantum Monte Carlo simulation. Here, the goal is to approximate certain properties of the Hamiltonian, such as the partition function, rather than simulating the full dynamics. We consider two specific systems: the transverse field Ising model of \cite{Bravyi15} and the ferromagnetic quantum spin systems of \cite{BG17}. For both simulations, the ideal evolution is decomposed using the second-order Suzuki formula and we show that such a decomposition can be made more efficient using our tightened analysis.

\medskip
\noindent\textbf{Transverse field Ising model.}
Consider the following $n$-qubit transverse field Ising model:
\begin{equation}
H=-A-B,\quad
A=\sum_{1\leq u<v\leq n}j_{u,v}Z_uZ_v,\quad
B=\sum_{1\leq u\leq n}h_uX_u.
\end{equation}
Here, $X_u$ and $Z_u$ are Pauli operators acting on the $u$th qubit, and $j_{u,v}\geq 0$ and $h_u\geq 0$ are nonnegative coefficients. Define $j:=\max\{j_{u,v},h_u\}$ to be the maximum norm of the interactions. Our goal is to approximate the partition function
\begin{equation}
\mathcal{Z}=\mathrm{Tr}\big(e^{-H}\big)
\end{equation}
up to a multiplicative error $0<\epsilon<1$.

Reference \cite{Bravyi15} solves this problem with an efficient classical algorithm. A key step in their algorithm is a decomposition of the evolution operator using the second-order Suzuki formula, so that
\begin{equation}
\mathcal{Z}'=\mathrm{Tr}\big(e^{\frac{1}{2r}A}e^{\frac{1}{r}B}e^{\frac{1}{2r}A}\big)^r
\lessapprox(1+\epsilon)\mathrm{Tr}\big(e^{-H}\big)=(1+\epsilon)\mathcal{Z}.
\end{equation}
However, their original analysis does not exploit the commutativity relation between $A$ and $B$, and can be improved by the techniques developed here.

Note that this is different from the usual setting of digital quantum simulation. Indeed, as the matrix exponentials in the product formula are no longer unitary, we will introduce an additional multiplicative factor when we apply \thm{trotter_error_comm_scaling}. Also, we need to estimate the multiplicative error as opposed to the additive error of the Trotter decomposition, which is addressed by the following lemma.

\begin{lemma}[Relative perturbation of eigenvalues {\cite[Theorem 2.1]{EI95} \cite[Theorem 5.4]{Ipsen98}}]
	\label{lem:relative}
	Let matrix $C$ be positive semidefinite and $D$ be nonsingular. Assume that the eigenvalues $\lambda_i(C)$ and $\lambda_i(D^\dagger CD)$ are ordered nonincreasingly. Then,
	\begin{equation}
	\lambda_i(D^\dagger CD)\leq\lambda_i(C)\norm{D^\dagger D}.
	\end{equation}
\end{lemma}

Let $A$ and $B$ be Hermitian matrices and consider the evolution $e^{t(A+B)}$ with $t\geq 0$. Our goal is to choose $r$ sufficiently large so that the eigenvalues are approximated as
\begin{equation}
\lambda_i\Big(\big(e^{\frac{t}{2r}A}e^{\frac{t}{r}B}e^{\frac{t}{2r}A}\big)^r\Big)\approx \lambda_i\Big(e^{t(A+B)}\Big)
\end{equation}
up to a small multiplicative error. We define
\begin{equation}
\begin{aligned}
U&:=e^{\frac{t}{r}(A+B)},\\
V&:=e^{\frac{t}{2r}A}e^{\frac{t}{r}B}e^{\frac{t}{2r}A},\\
W&:=\exp_{\mathcal{T}}\bigg(\int_{0}^{\frac{t}{r}}\mathrm{d}\tau\ e^{-\tau (A+B)}\bigg[e^{\frac{\tau}{2}A}Be^{-\frac{\tau}{2}A}-B
+e^{\frac{\tau}{2}A}e^{\tau B}\frac{A}{2}e^{-\tau B}e^{-\frac{\tau}{2}A}-\frac{A}{2}\bigg]e^{\tau (A+B)}\bigg).
\end{aligned}
\end{equation}
Then, both $U$ and $V$ are positive-semidefinite operators and we know from \thm{error_type} that $V=UW$. In \append{qmc}, we show that
\begin{equation}
\norm{W}\leq
\exp\bigg(\bigg(\frac{t^3}{8r^3}\norm{\big[A,\big[A,B\big]\big]}+\frac{t^3}{12r^3}\norm{\big[B,\big[B,A\big]\big]}\bigg)e^{4\frac{t}{r}(\norm{A}+\norm{B})}\bigg).
\end{equation}

Our goal is to bound the eigenvalues $\lambda_i\big(V^r\big)$ in terms of $\lambda_i\big(U^r\big)$. This can be done recursively as follows. We first replace the rightmost $V$ by $UW$ and the leftmost $V$ by $W^\dagger U$. Invoking \lem{relative}, we have
\begin{equation}
\lambda_i\big(V^r\big)
=\lambda_i\big(W^\dagger UV^{r-2}UW\big)
\leq\lambda_i\big(UV^{r-2}U\big)\norm{W}^2.
\end{equation}
By \cite[Theorem 1.3.22]{horn2012matrix},
\begin{equation}
\lambda_i\big(UV^{r-2}U\big)
=\lambda_i\big(V^{\frac{r}{2}-1}UUV^{\frac{r}{2}-1}\big).
\end{equation}
We now apply a similar procedure to obtain
\begin{equation}
\begin{aligned}
\lambda_i\big(V^{\frac{r}{2}-1}UUV^{\frac{r}{2}-1}\big)
&=\lambda_i\big(W^\dagger UV^{\frac{r}{2}-2}UUV^{\frac{r}{2}-2}UW\big)\\
&\leq\lambda_i\big(UV^{\frac{r}{2}-2}UUV^{\frac{r}{2}-2}U\big)\norm{W}^2\\
&=\lambda_i\big(V^{\frac{r}{4}-1}UUV^{\frac{r}{2}-2}UUV^{\frac{r}{4}-1}\big)\norm{W}^2\\
&\leq\lambda_i\big(UV^{\frac{r}{4}-2}UUV^{\frac{r}{2}-2}UUV^{\frac{r}{4}-2}U\big)\norm{W}^4\\
&=\lambda_i\big(V^{\frac{r}{4}-1}UUV^{\frac{r}{4}-2}UUV^{\frac{r}{4}-2}UUV^{\frac{r}{4}-1}\big)\norm{W}^4\\
&\leq\lambda_i\big(UV^{\frac{r}{4}-2}UUV^{\frac{r}{4}-2}UUV^{\frac{r}{4}-2}UUV^{\frac{r}{4}-2}U\big)\norm{W}^6.
\end{aligned}
\end{equation}
To ensure that this recursion is valid, we choose $r$ to be a power of $2$. Since any positive integer is between $2^m$ and $2^{m+1}$ for some $m\geq 0$, this choice only enlarges $r$ by a factor of at most $2$. Overall,
\begin{equation}
\lambda_i\big(V^r\big)
\leq\lambda_i\big(U^r\big)\norm{W}^r.
\end{equation}

We know that
\begin{equation}
\norm{W}^r\leq
\exp\bigg(\bigg(\frac{t^3}{8r^2}\norm{\big[A,\big[A,B\big]\big]}+\frac{t^3}{12r^2}\norm{\big[B,\big[B,A\big]\big]}\bigg)e^{4\frac{t}{r}(\norm{A}+\norm{B})}\bigg).
\end{equation}
We first choose
\begin{equation}
r\geq 4t\big(\norm{A}+\norm{B}\big)
\end{equation}
so that $e^{4\frac{t}{r}(\norm{A}+\norm{B})}\leq e<4$. We then set
\begin{equation}
r\geq\max\bigg\{\sqrt{\frac{t^3}{\epsilon}\norm{\big[A,\big[A,B\big]\big]}},\sqrt{\frac{2t^3}{3\epsilon}\norm{\big[B,\big[B,A\big]\big]}}\bigg\}
\end{equation}
so that both $\frac{t^3}{8r^2}\norm{\big[A,\big[A,B\big]\big]}$ and $\frac{t^3}{12r^2}\norm{\big[B,\big[B,A\big]\big]}$ are bounded by $\epsilon/8$. Therefore, we have $\norm{W}^r\leq e^\epsilon$ as long as $r$ is a power of $2$ satisfying
\begin{equation}
r\geq\max\bigg\{4t\big(\norm{A}+\norm{B}\big),\sqrt{\frac{t^3}{\epsilon}\norm{\big[A,\big[A,B\big]\big]}},\sqrt{\frac{2t^3}{3\epsilon}\norm{\big[B,\big[B,A\big]\big]}}\bigg\},
\end{equation}
which implies
\begin{equation}
\mathcal{Z}'
=\sum_{i}\lambda_i\big(V^r\big)
\leq\sum_{i}\lambda_i\big(U^r\big)e^\epsilon
\approx(1+\epsilon)\sum_{i}\lambda_i\big(U^r\big)
=(1+\epsilon)\mathcal{Z}
\end{equation}
assuming $\epsilon\ll 1$.

Following similar arguments, we can show that this choice of $r$ also gives a lower bound of $\mathcal{Z}'$ with $\mathcal{Z}'\gtrapprox(1-\epsilon)\mathcal{Z}$. Indeed, we can bound the eigenvalues $\lambda_i\big(U^r\big)$ in terms of $\lambda_i\big(V^r\big)$ using \lem{relative} and the relation $U=VW^{-1}$. Using \lem{time_ordered_norm_bound} and the fact that $W^{-1}$ is the reversal of the time-ordered exponential $W$, we have $\norm{W^{-1}}^r\leq e^\epsilon$ as well, giving $\sum_{i}\lambda_i\big(U^r\big)
\leq\sum_{i}\lambda_i\big(V^r\big)e^\epsilon$. The lower bound now follows since $1/e^{\epsilon}\approx 1/(1+\epsilon)\approx 1-\epsilon$ for $\epsilon\ll 1$. We have therefore approximated the partition function up to a multiplicative error $\epsilon$.

We now specialize our result to the transverse field Ising Hamiltonian with $t=1$. We find that
\begin{equation}
\norm{A}=\cO{n^2 j},\quad
\norm{B}=\cO{nj},\quad
\norm{\big[A,\big[A,B\big]\big]}=\cO{n^3j^3},\quad
\norm{\big[B,\big[B,A\big]\big]}=\cO{n^2j^3},
\end{equation}
which implies
\begin{equation}
r=\cO{n^2j+n^{3/2}j^{3/2}\epsilon^{-1/2}}.
\end{equation}
By \cite[p.\ 17]{Bravyi15}, this gives a fully polynomial randomized approximation scheme (FPRAS) with running time
\begin{equation}
\tildecO{n^{17}r^{14}\epsilon^{-2}}
=\tildecO{n^{45}j^{14}\epsilon^{-2}+n^{38}j^{21}\epsilon^{-9}},
\end{equation}
improving over the previous complexity of
\begin{equation}
\tildecO{n^{59}j^{21}\epsilon^{-9}}.
\end{equation}

\medskip
\noindent\textbf{Quantum ferromagnets.}
We now apply our technique to improve the Monte Carlo simulation of ferromagnetic quantum spin systems \cite{BG17}. Such systems are described by the $n$-qubit Hamiltonian
\begin{equation}
H=\sum_{1\leq u<v\leq n}\big({-}b_{uv}X_uX_v+c_{uv}Y_uY_v\big)+\sum_{u=1}^{n}d_u\big(I+Z_u\big),
\end{equation}
where $0\leq b_{uv}\leq 1$, $-b_{uv}\leq c_{uv}\leq b_{uv}$, and $-1\leq d_{uv}\leq 1$. It will be convenient to rewrite these Hamiltonians using the coefficients $p_{uv}=(b_{uv}-c_{uv})/2$ and $q_{uv}=(b_{uv}+c_{uv})/2$ as
\begin{equation}
H=\sum_{1\leq u<v\leq n}p_{uv}\big({-}X_uX_v-Y_uY_v\big)+\sum_{1\leq u<v\leq n}q_{uv}\big({-}X_uX_v+Y_uY_v\big)+\sum_{u=1}^{n}d_u\big(I+Z_u\big).
\end{equation}
Since $|c_{uv}|\leq b_{uv}\leq 1$, we have $p_{uv},q_{uv}\in[0,1]$.

Our goal is to approximate the partition function
\begin{equation}
\mathcal{Z}(\beta,H)=\mathrm{Tr}\big[e^{-\beta H}\big]
\end{equation}
for $\beta>0$. Following the setting of \cite{BG17}, we restrict ourselves to the $n$-qubit matchgate set
\begin{equation}
\label{eq:gate_set}
\bigg\{f_{u}\big(e^{\pm t}\big),g_{uv}(t),h_{uv}(t)\ \bigg|\ u,v=1,\ldots,n,\ u\neq v,\ 0<t<\frac{1}{2}\bigg\},
\end{equation}
where
\begin{equation}
f\big(e^{\pm t}\big)=
\begin{bmatrix}
e^{\pm t} & 0\\
0 & 1
\end{bmatrix}
,\qquad g(t)=
\begin{bmatrix}
1+t^2 & 0 & 0 & t\\
0 & 1 & 0 & 0\\
0 & 0 & 1 & 0\\
t & 0 & 0 & 1
\end{bmatrix}
,\qquad h(t)=
\begin{bmatrix}
1 & 0 & 0 & 0\\
0 & 1+t^2 & t & 0\\
0 & t & 1 & 0\\
0 & 0 & 0 & 1
\end{bmatrix}
\end{equation}
and the subscripts $u,v$ indicate the qubits on which the gates act nontrivially. The motivations for using these gates can be found in \cite{BG17} which we do not repeat here. These gates approximately implement the exponential of the Hamiltonian terms in the sense that
\begin{equation}
f_u\big(e^{\pm t}\big)=e^{\pm\frac{t}{2}(I+Z_u)},\qquad
g_{uv}(t)=e^{-\frac{t}{2}(-X_uX_v+Y_uY_v)+\OO{t^2}},\qquad
h_{uv}(t)=e^{-\frac{t}{2}(-X_uX_v-Y_uY_v)+\OO{t^2}}.
\end{equation}

We divide the evolution into $r$ steps and apply the second-order Suzuki formula within each step, obtaining
\begin{equation}
\begin{aligned}
e^{-\frac{\beta}{r}H}
&\approx\prod_{1\leq u\leq n}e^{-\frac{\beta}{2r}d_u(I+Z_u)}\prod_{1\leq u<v\leq n}e^{-\frac{\beta}{2r}q_{uv}({-}X_uX_v+Y_uY_v)}\prod_{1\leq u<v\leq n}e^{-\frac{\beta}{2r}p_{uv}({-}X_uX_v-Y_uY_v)}\\
&\quad\cdot\prod_{1\leq u<v\leq n}e^{-\frac{\beta}{2r}p_{uv}({-}X_uX_v-Y_uY_v)}\prod_{1\leq u<v\leq n}e^{-\frac{\beta}{2r}q_{uv}({-}X_uX_v+Y_uY_v)}\prod_{1\leq u\leq n}e^{-\frac{\beta}{2r}d_u(I+Z_u)}\\
&\approx\prod_{1\leq u\leq n}f_u\big(e^{-\frac{\beta}{r}d_u}\big)\prod_{1\leq u<v\leq n}g_{uv}\bigg(\frac{\beta}{r}q_{uv}\bigg)\prod_{1\leq u<v\leq n}h_{uv}\bigg(\frac{\beta}{r}p_{uv}\bigg)\\
&\quad\cdot\prod_{1\leq u<v\leq n}h_{uv}\bigg(\frac{\beta}{r}p_{uv}\bigg)\prod_{1\leq u<v\leq n}g_{uv}\bigg(\frac{\beta}{r}q_{uv}\bigg)\prod_{1\leq u\leq n}f_u\big(e^{-\frac{\beta}{r}d_u}\big).
\end{aligned}
\end{equation}
Here, we have two sources of error: the Trotter error and the error from using the gate set \eq{gate_set}. We choose
\begin{equation}
\label{eq:r_cond1}
r> 2\beta,
\end{equation}
so that we can implement the product formula using gates from \eq{gate_set} with parameters
\begin{equation}
-\frac{1}{2}<-\frac{\beta}{r}d_u<\frac{1}{2},\quad
0<\frac{\beta}{r}q_{uv}<\frac{1}{2},\quad
0<\frac{\beta}{r}p_{uv}<\frac{1}{2}.
\end{equation}

In \append{qmc}, we use the interaction picture (\lem{interaction_picture}) to show that
\begin{equation}
\begin{aligned}
&\prod_{1\leq u\leq n}f_u\big(e^{-\frac{\beta}{r}d_u}\big)\prod_{1\leq u<v\leq n}g_{uv}\bigg(\frac{\beta}{r}q_{uv}\bigg)\prod_{1\leq u<v\leq n}h_{uv}\bigg(\frac{\beta}{r}p_{uv}\bigg)\\
&\quad\cdot\prod_{1\leq u<v\leq n}h_{uv}\bigg(\frac{\beta}{r}p_{uv}\bigg)\prod_{1\leq u<v\leq n}g_{uv}\bigg(\frac{\beta}{r}q_{uv}\bigg)\prod_{1\leq u\leq n}f_u\big(e^{-\frac{\beta}{r}d_u}\big)\\
=&\ e^{-\frac{\beta}{r}H} U,
\end{aligned}
\end{equation}
where the operator $U$ has spectral norm bounded by
\begin{equation}
\norm{U}\leq
\exp\bigg(\frac{2n^2\beta^2}{r^2}e^{\frac{4n^2\beta}{r}}
+\frac{cn^4\beta^3}{r^3}e^{\frac{12n^2\beta}{r}}\bigg)
\end{equation}
for some constant $c>0$. The remaining analysis proceeds in a similar way as that of the transverse field Ising model. We find that each eigenvalue of
\begin{equation}
\begin{aligned}
\bigg[&\prod_{1\leq u\leq n}f_u\big(e^{-\frac{\beta}{r}d_u}\big)\prod_{1\leq u<v\leq n}g_{uv}\bigg(\frac{\beta}{r}q_{uv}\bigg)\prod_{1\leq u<v\leq n}h_{uv}\bigg(\frac{\beta}{r}p_{uv}\bigg)\\
&\quad\cdot\prod_{1\leq u<v\leq n}h_{uv}\bigg(\frac{\beta}{r}p_{uv}\bigg)\prod_{1\leq u<v\leq n}g_{uv}\bigg(\frac{\beta}{r}q_{uv}\bigg)\prod_{1\leq u\leq n}f_u\big(e^{-\frac{\beta}{r}d_u}\big)\bigg]^r
\end{aligned}
\end{equation}
approximates the corresponding eigenvalue of the ideal evolution $e^{-\beta H}$ with a multiplicative factor
\begin{equation}
\norm{U}^r\leq
\exp\bigg(\frac{2n^2\beta^2}{r}e^{\frac{4n^2\beta}{r}}
+\frac{cn^4\beta^3}{r^2}e^{\frac{12n^2\beta}{r}}\bigg).
\end{equation}
We first set
\begin{equation}
\label{eq:r_cond2}
r\geq 24n^2\beta,
\end{equation}
so that
\begin{equation}
\norm{U}^r\leq
\exp\bigg(\frac{4n^2\beta^2}{r}
+\frac{2cn^4\beta^3}{r^2}\bigg).
\end{equation}
We then choose
\begin{equation}
\label{eq:r_cond34}
r\geq\max\bigg\{\frac{8n^2\beta^2}{\epsilon},\frac{2\sqrt{c}n^2\beta^{3/2}}{\epsilon^{1/2}}\bigg\}
\end{equation}
to ensure that the multiplicative error is at most $\epsilon$. By \eq{r_cond1}, \eq{r_cond2}, and \eq{r_cond34},
\begin{equation}
r=\cO{\frac{n^2\ceil{\beta}^2}{\epsilon}},
\end{equation}
which gives the total gate complexity \cite[Supplementary p.\ 7]{BG17}
\begin{equation}
j:=4n^2r=\cO{\frac{(1+\beta^2)n^4}{\epsilon}}.
\end{equation}

The result of \cite[Theorem 2]{BG17} gave a Monte Carlo simulation algorithm for the ferromagnetic quantum spin systems. To improve that result, we also need to estimate the error of partial sequence of the product formula as in \cite[Eq.\ (13)]{BG17}. This can be done in a similar way as our above analysis. The resulting randomized approximation scheme has runtime
\begin{equation}
\tildecO{\frac{j^{23}}{\epsilon^2}}=\tildecO{\frac{n^{92}(1+\beta^{46})}{\epsilon^{25}}},
\end{equation}
which improves the runtime of the original Bravyi-Gosset algorithm
\begin{equation}
\tildecO{\frac{n^{115}(1+\beta^{46})}{\epsilon^{25}}}.
\end{equation}

\section{Error bounds with small prefactors}
\label{sec:prefactor}

We now derive Trotter error bounds with small prefactors. These bounds complement the above asymptotic analysis and can be used to optimize near-term implementations of quantum simulation. In \sec{prefactor_pf12}, we show that our analysis reproduces previous tight error bounds for the first- and second-order formulas~\cite{Huyghebaert_1990,DT10,Jahnke2000}. We then give numerical evidence in \sec{prefactor_pf2k} showing that our higher-order bounds are close to tight for certain nearest-neighbor interactions and power-law Hamiltonians. Throughout this section, we let $H$ be Hermitian, $t\in\R$, and we decompose the real-time evolution $e^{-itH}$.

\subsection{First- and second-order error bounds}
\label{sec:prefactor_pf12}
We derive error bounds for the first-order Lie-Trotter formula and second-order Suzuki formula following the idea of \cite{Suzuki85,Kivlichan19}. In this approach, we first analyze the Trotter error of decomposing the evolution of a two-term Hamiltonian. We then bootstrap the result to analyze general Hamiltonians with an arbitrary number of operator summands. The resulting bounds are nearly tight because they match the lowest-order term of the BCH expansion up to an application of the triangle inequality \cite{Huyghebaert_1990,Suzuki85,DT10,WBCHT14,Kivlichan19}.

Let $H=A+B$ be a two-term Hamiltonian. The evolution under $H$ for time $t\geq 0$ is given by $e^{-itH}=e^{-it(A+B)}$, which we decompose using the first-order Lie-Trotter formula $\mathscr{S}_1(t)=e^{-itB}e^{-itA}$. We first construct the differential equation
\begin{equation}
	\frac{\mathrm{d}}{\mathrm{d}t}\mathscr{S}_1(t)=-iH\mathscr{S}_1(t)+e^{-itB}\big(e^{itB}iAe^{-itB}-iA\big)e^{-itA}
\end{equation}
with initial condition $\mathscr{S}_1(0)=I$. Using the variation-of-parameters formula (\lem{td_Duhamel}),
\begin{equation}
	\mathscr{S}_1(t)=e^{-itH}+\int_{0}^{t}\mathrm{d}\tau_1\ e^{-i(t-\tau_1)H}e^{-i\tau_1B}\big(e^{i\tau_1B}iAe^{-i\tau_1B}-iA\big)e^{-i\tau_1A}.
\end{equation}
Using \thm{error_order_cond} or by direct calculation, we find the order condition $e^{i\tau_1B}iAe^{-i\tau_1B}-iA=\OO{\tau_1}$, which implies
\begin{equation}
	e^{i\tau_1B}iAe^{-i\tau_1B}-iA=\int_{0}^{\tau_1}\mathrm{d}\tau_2\ e^{i\tau_2B}\big[iB,iA\big]e^{-i\tau_2B}.
\end{equation}
Altogether, we have the representation
\begin{equation}
	\mathscr{S}_1(t)=e^{-itH}+\int_{0}^{t}\mathrm{d}\tau_1\int_{0}^{\tau_1}\mathrm{d}\tau_2\
	e^{-i(t-\tau_1)H}e^{-i\tau_1B}e^{i\tau_2B}\big[iB,iA\big]e^{-i\tau_2B}e^{-i\tau_1A}
\end{equation}
and the error bound for $t\geq 0$
\begin{equation}
	\norm{\mathscr{S}_1(t)-e^{-itH}}\leq\frac{t^2}{2}\norm{[B,A]}.
\end{equation}
We bootstrap this bound to analyze a general Hamiltonian $H=\sum_{\gamma}^{\Gamma}H_\gamma$. By the triangle inequality,
\begin{equation}
\begin{aligned}
	\norm{\prod_{\gamma=1}^{\Gamma}e^{-itH_\gamma}-e^{-it\sum_{\gamma=1}^{\Gamma}H_\gamma}}
	&\leq\sum_{\gamma_1=1}^{\Gamma}
	\norm{e^{-it\sum_{\gamma_2=\gamma_1+1}^{\Gamma}H_{\gamma_2}}\prod_{\gamma_2=1}^{\gamma_1}e^{-itH_{\gamma_2}}
		-e^{-it\sum_{\gamma_2=\gamma_1}^{\Gamma}H_{\gamma_2}}\prod_{\gamma_2=1}^{\gamma_1-1}e^{-itH_{\gamma_2}}}\\
	&=\sum_{\gamma_1=1}^{\Gamma}
	\norm{e^{-it\sum_{\gamma_2=\gamma_1+1}^{\Gamma}H_{\gamma_2}}e^{-itH_{\gamma_1}}
		-e^{-it\sum_{\gamma_2=\gamma_1}^{\Gamma}H_{\gamma_2}}}\\
	&\leq\frac{t^2}{2}\sum_{\gamma_1=1}^{\Gamma}\norm{\bigg[\sum_{\gamma_2=\gamma_1+1}^{\Gamma}H_{\gamma_2},H_{\gamma_1}\bigg]}.
\end{aligned}
\end{equation}
We have thus obtained:
\begin{proposition}[Tight error bound for the first-order Lie-Trotter formula]
	Let $H=\sum_{\gamma=1}^{\Gamma}H_\gamma$ be a Hamiltonian consisting of $\Gamma$ summands and $t\geq 0$. Let $\mathscr{S}_1(t)=\prod_{\gamma=1}^{\Gamma}e^{-itH_\gamma}$ be the first-order Lie-Trotter formula. Then, the additive Trotter error can be bounded as
	\begin{equation}
		\norm{\mathscr{S}_1(t)-e^{-itH}}\leq\frac{t^2}{2}\sum_{\gamma_1=1}^{\Gamma}\norm{\bigg[\sum_{\gamma_2=\gamma_1+1}^{\Gamma}H_{\gamma_2},H_{\gamma_1}\bigg]}.
	\end{equation}
\end{proposition}

A generalization of this analysis gives an error bound for the second-order Suzuki formula with a small prefactor. For the two-term case, our goal is to decompose the evolution $e^{-itH}=e^{-it(A+B)}$ using the product formula $\mathscr{S}_2(t)=e^{-i\frac{t}{2}A}e^{-itB}e^{-i\frac{t}{2}A}$. Using the variation-of-parameters formula (\lem{td_Duhamel}), we have
\begin{equation}
\mathscr{S}_2(t)=e^{-itH}+
\int_{0}^{t}\mathrm{d}\tau_1\ e^{-i(t-\tau_1)H}
e^{-i\frac{\tau_1}{2}A}\mathscr{T}_2(\tau_1)e^{-i\tau_1B}e^{-i\frac{\tau_1}{2}A},
\end{equation}
where
\begin{equation}
	\mathscr{T}_2(\tau_1)=e^{-i\tau_1B}\bigg(-i\frac{A}{2}\bigg)e^{i\tau_1B}+i\frac{A}{2}
	+e^{i\frac{\tau_1}{2}A}\big(iB\big)e^{-i\frac{\tau_1}{2}A}-iB.
\end{equation}
By \thm{error_order_cond} or a direct calculation, we find the order condition $\mathscr{T}_2(\tau_1)=\OO{\tau_1^2}$, which implies
\begin{equation}
	\mathscr{T}_2(\tau_1)=\int_{0}^{\tau_1}\mathrm{d}\tau_2\int_{0}^{\tau_2}\mathrm{d}\tau_3
	\bigg(e^{-i\tau_3B}\bigg[-iB,\bigg[-iB,-i\frac{A}{2}\bigg]\bigg]e^{i\tau_3B}
	+e^{i\frac{\tau_3}{2}A}\bigg[i\frac{A}{2},\bigg[i\frac{A}{2},iB\bigg]\bigg]e^{-i\frac{\tau_3}{2}A}\bigg).
\end{equation}
Altogether, we have the representation
\begin{equation}
\begin{aligned}
	\mathscr{S}_2(t)&=e^{-itH}+
	\int_{0}^{t}\mathrm{d}\tau_1\int_{0}^{\tau_1}\mathrm{d}\tau_2\int_{0}^{\tau_2}\mathrm{d}\tau_3\
	e^{-i(t-\tau_1)H}e^{-i\frac{\tau_1}{2}A}\\
	&\qquad\qquad\cdot\bigg(e^{-i\tau_3B}\bigg[-iB,\bigg[-iB,-i\frac{A}{2}\bigg]\bigg]e^{i\tau_3B}
	+e^{i\frac{\tau_3}{2}A}\bigg[i\frac{A}{2},\bigg[i\frac{A}{2},iB\bigg]\bigg]e^{-i\frac{\tau_3}{2}A}\bigg)
	e^{-i\tau_1B}e^{-i\frac{\tau_1}{2}A},
\end{aligned}
\end{equation}
and the error bound for $t\geq 0$
\begin{equation}
	\norm{\mathscr{S}_2(t)-e^{-itH}}\leq
	\frac{t^3}{12}\norm{[B,[B,A]]}
	+\frac{t^3}{24}\norm{[A,[A,B]]}.
\end{equation}
For a general Hamiltonian $H=\sum_{\gamma=1}^{\Gamma}H_\gamma$, we apply the triangle inequality to get
\begin{equation}
\begin{aligned}
	\norm{\prod_{\gamma=\Gamma}^{1}e^{-i\frac{t}{2}H_{\gamma}}\prod_{\gamma=1}^{\Gamma}e^{-i\frac{t}{2}H_{\gamma}}-e^{-it\sum_{\gamma=1}^{\Gamma}H_\gamma}}
	\leq&\sum_{\gamma_1=1}^{\Gamma} \bigg\Vert\prod_{\gamma_2=\gamma_1}^{1}e^{-i\frac{t}{2}H_{\gamma_2}}
		\cdot e^{-it\sum_{\gamma_2=\gamma_1+1}^{\Gamma}H_{\gamma_2}}
		\cdot\prod_{\gamma_2=1}^{\gamma_1}e^{-i\frac{t}{2}H_{\gamma_2}}\\
		&\qquad-\prod_{\gamma_2=\gamma_1-1}^{1}e^{-i\frac{t}{2}H_{\gamma_2}}
		\cdot e^{-it\sum_{\gamma_2=\gamma_1}^{\Gamma}H_{\gamma_2}}
		\cdot\prod_{\gamma_2=1}^{\gamma_1-1}e^{-i\frac{t}{2}H_{\gamma_2}}\bigg\Vert\\
	=&\sum_{\gamma_1=1}^{\Gamma}\norm{e^{-i\frac{t}{2}H_{\gamma_1}}
		e^{-it\sum_{\gamma_2=\gamma_1+1}^{\Gamma}H_{\gamma_2}}
		e^{-i\frac{t}{2}H_{\gamma_1}}
		-e^{-it\sum_{\gamma_2=\gamma_1}^{\Gamma}H_{\gamma_2}}}\\
	\leq&\ \frac{t^3}{12}\sum_{\gamma_1=1}^{\Gamma}\norm{\bigg[\sum_{\gamma_3=\gamma_1+1}^{\Gamma}H_{\gamma_3},\bigg[\sum_{\gamma_2=\gamma_1+1}^{\Gamma}H_{\gamma_2},H_{\gamma_1}\bigg]\bigg]}\\
	&+\frac{t^3}{24}\sum_{\gamma_1=1}^{\Gamma}\norm{\bigg[H_{\gamma_1},\bigg[H_{\gamma_1},\sum_{\gamma_2=\gamma_1+1}^{\Gamma}H_{\gamma_2}\bigg]\bigg]}.
\end{aligned}
\end{equation}

\begin{proposition}[Tight error bound for the second-order Suzuki formula]
	Let $H=\sum_{\gamma=1}^{\Gamma}H_\gamma$ be a Hamiltonian consisting of $\Gamma$ summands and $t\geq 0$. Let $\mathscr{S}_2(t)=\prod_{\gamma=\Gamma}^{1}e^{-i\frac{t}{2}H_{\gamma}}\prod_{\gamma=1}^{\Gamma}e^{-i\frac{t}{2}H_{\gamma}}$ be the second-order Suzuki formula. Then, the additive Trotter error can be bounded as
	\begin{equation}
	\begin{aligned}
	\norm{\mathscr{S}_2(t)-e^{-itH}}\leq&\
	\frac{t^3}{12}\sum_{\gamma_1=1}^{\Gamma}\norm{\bigg[\sum_{\gamma_3=\gamma_1+1}^{\Gamma}H_{\gamma_3},\bigg[\sum_{\gamma_2=\gamma_1+1}^{\Gamma}H_{\gamma_2},H_{\gamma_1}\bigg]\bigg]}\\
	&+\frac{t^3}{24}\sum_{\gamma_1=1}^{\Gamma}\norm{\bigg[H_{\gamma_1},\bigg[H_{\gamma_1},\sum_{\gamma_2=\gamma_1+1}^{\Gamma}H_{\gamma_2}\bigg]\bigg]}.
	\end{aligned}
	\end{equation}
\end{proposition}

\subsection{Higher-order error bounds}
\label{sec:prefactor_pf2k}
We now consider error bounds for higher-order product formulas. Compared to the first- and second-order cases, these formulas are harder to analyze due to their more complicated definitions. Nevertheless, higher-order formulas have better asymptotic scaling and can be surprisingly efficient even for simulating small systems, as observed in \cite{CMNRS18}.

We have analyzed the error of higher-order product formulas in \sec{theory}. That analysis is sufficient to establish the commutator scaling in \thm{trotter_error_comm_scaling}, but the resulting bounds have large prefactors. Here, we propose heuristic strategies to tighten the analysis and numerically benchmark our bounds for nearest-neighbor and power-law Hamiltonians. Our heuristics are specified in \append{pf2k}.

Although we do not have a rigorous proof of the tightness of our higher-order bounds, numerical evidence suggests that they are close to tight for various systems. We first consider simulating a one-dimensional Heisenberg model with a random magnetic field
\begin{align}
H=\sum_{j=1}^{n-1} \big(X_jX_{j+1}+Y_jY_{j+1}+Z_jZ_{j+1} + h_j Z_j\big)
\label{eq:heisenberg}
\end{align}
with coefficients $h_j \in [-1,1]$ chosen uniformly at random, where $X_j$, $Y_j$, and $Z_j$ are Pauli operators acting on the $j$th qubit. This system can be simulated to understand the transition between the many-body localized phase and the thermalized phase in condensed matter physics, although a classical simulation is only feasible when the system size is small \cite{LLA15}.

We classify the summands of the Hamiltonian into two groups and set
\begin{equation}
\label{eq:even-odd}
\begin{aligned}
A&=\sum_{j=1}^{\floor{\frac{n}{2}}}\big(X_{2j-1}X_{2j}+Y_{2j-1}Y_{2j}+Z_{2j-1}Z_{2j} + h_{2j-1} Z_{2j-1}\big),\\
B&=\sum_{j=1}^{\ceil{\frac{n}{2}}-1}\big(X_{2j}X_{2j+1}+Y_{2j}Y_{2j+1}+Z_{2j}Z_{2j+1} + h_{2j} Z_{2j}\big).
\end{aligned}
\end{equation}
Here, all the summands in $A$ (and $B$) commute with each other, so we can further decompose exponentials like $e^{-ita_kA}$ (and $e^{-itb_kB}$) without introducing error, giving a product formula with summands ordered in an \textit{even-odd} pattern \cite{CS19}. We also consider grouping Hamiltonian summands as
\begin{equation}
\label{eq:x-y-z}
H_1=\sum_{j=1}^{n-1}X_{j}X_{j+1},\qquad
H_2=\sum_{j=1}^{n-1}Y_{j}Y_{j+1},\qquad
H_3=\sum_{j=1}^{n-1}\big(Z_{j}Z_{j+1} + h_{j} Z_{j}\big),
\end{equation}
which we call the \textit{X-Y-Z} ordering \cite{CMNRS18}. Similar to the even-odd ordering, the summands in $H_1$, $H_2$, and $H_3$ commute with each other respectively, so the corresponding exponentials can also be decomposed without error. Note that our asymptotic bounds in \thm{trotter_error_comm_scaling} and \cor{trotter_number_comm_scaling} hold irrespective of the ordering of Hamiltonian summands, but the prefactors will depend on the choice of ordering. Our choice here maximizes the commutativity of the Hamiltonian.

Up to a difference on the boundary condition, Ref.\ \cite{CMNRS18} estimates the resource requirements of simulating the Heisenberg model using various quantum algorithms. They find that product formulas, especially the fourth-order and the sixth-order formulas, can outperform more recent quantum algorithms for simulating small instances of \eq{heisenberg}, although their best Trotter error bound is loose by several orders of magnitude. This is alleviated in \cite{CS19}, which gives a fourth-order bound that overestimates the gate complexity by about a factor of $17$. For a fair comparison, we numerically implement our approach to analyze the fourth-order formula $\mathscr{S}_4(t)$ as well; see \append{pf2k} for detailed derivations.

For the even-odd ordering, we need to compute all the nested commutators of $A$ and $B$. We do this by fixing one term $X_{2j-1}X_{2j}+Y_{2j-1}Y_{2j}+Z_{2j-1}Z_{2j} + h_{2j-1} Z_{2j-1}$ of $A$ in the innermost layer and simplifying all the outer terms using geometrical locality. We then apply the triangle inequality to analyze the summation of terms over $j=1,\ldots,\floor{\frac{n}{2}}$. We use a similar approach to analyze the \textit{X-Y-Z} ordering. This computes our error bounds for small $t$. To simulate for a longer time, we divide the evolution into $r$ Trotter steps and apply our bounds within each step. We seek the smallest Trotter number $r$ for which the estimated error is at most some desired $\epsilon$. This can be efficiently computed using a binary search as described in~\cite{CMNRS18}.

We compare our improved analysis with the best previous bounds \cite{CMNRS18,CS19} for simulating the Heisenberg model \eq{heisenberg}. Specifically, we consider the so-called analytic bound \cite[Proposition F.4]{CMNRS18}, which applies to both the even-odd and the \textit{X-Y-Z} ordering. The commutator bound of \cite[Theorem F.11]{CMNRS18} offers a slight improvement over the analytic bound, but its numerical implementation requires extensive classical computations and so we only compare the existing result for the \textit{X-Y-Z} ordering. Likewise, we compare the locality-based bound of \cite[Supplementary Material IV B]{CS19} only for the even-odd case, although it can exploit the geometrical locality of the \textit{X-Y-Z} ordering as well.

To understand how tight our bounds are, we also include the empirical Trotter number by directly computing the error $\norm{\big(\mathscr{S}_4(t/r)\big)^r-e^{-itH}}$ for $n=4,\ldots,12$ and extrapolating the results to larger systems. We choose the evolution time $t=n$ and set the simulation accuracy $\epsilon=10^{-3}$ as in~\cite{CMNRS18} and \cite{CS19}. For each system size, we generate five instances of Hamiltonians with random coefficients. Our results are plotted in \fig{pfbound}.

\begin{figure}[t]
	\centering
	\begin{subfigure}{.5\linewidth}
		\resizebox{.95\textwidth}{!}{
			\begin{tikzpicture}
			\begin{axis}[
			log x ticks with fixed point,
			xtick={10,100},
			xmode=log,
			ymode=log,
			xmin = 3,
			xmax = 300,
			ymin=10^1,
			ymax=10^9,
			width=12cm,
			ymajorgrids=true,
			yminorgrids=true,
			legend style={at={(0.02,0.98)},anchor=north west},
			xlabel={$n$},
			ylabel={$r$}, ylabel near ticks,
			legend cell align={left},
			clip=false
			]

			\addlegendimage{empty legend}
			\addlegendentry[yshift=0pt]{\hspace{-.4cm}
			\textbf{Even-odd}}

			\addplot[only marks,mark=square*,color=red] coordinates {
				(10,136921)
				(16,442697)
				(23,1095749)
				(32,2500049)
				(45,5859047)
				(64,14125505)
				(91,34037275)
				(128,79838219)
				(181,189775803)
				(256,451362181)
			};
			\addlegendentry{Analytic \cite{CMNRS18}}

			\addplot[only marks,mark=square*,color=blue,error bars/.cd,y dir=both,y explicit] coordinates {
				(10,2197.20000000000) +- (0,14.6696966567138)
				(16,5031.60000000000) +- (0,42.7176778395081)
				(23,9050.60000000000) +- (0,60.7107898153203)
				(32,15224.6000000000) +- (0,55.3741817095296)
				(45,25827.8000000000) +- (0,135.564375851475)
				(64,44428.8000000000) +- (0,196.851974844044)
				(91,75915.6000000000) +- (0,274.783187258609)
				(128,127334.200000000) +- (0,203.733158813189)
				(181,215319.400000000) +- (0,249.596274010651)
				(256,362547) +- (0,874.553886275740)
			};
			\addlegendentry{Locality \cite{CS19}}

			\addplot[only marks,mark=square*,color=bndclr,error bars/.cd,y dir=both,y explicit] coordinates {
				(10,645.400000000000) +- (0,7.92464510246358)
				(16,1347.80000000000) +- (0,13.2928552237659)
				(23,2357) +- (0,23.6008474424119)
				(32,3897.40000000000) +- (0,23.9645571626100)
				(45,6553.80000000000) +- (0,42.3284774117851)
				(64,11198.2000000000) +- (0,45.7733546946256)
				(91,19072) +- (0,92.0679097188592)
				(128,31873.8000000000) +- (0,104.979521812590)
				(181,53936.2000000000) +- (0,111.333732534214)
				(256,90564.4000000000) +- (0,269.900537235479)
			};
			\addlegendentry{Our bound}

			\addplot[only marks,mark=square*,color=empclr,error bars/.cd,y dir=both,y explicit] coordinates {
				(4,24.8000000000000) +- (0,1.64316767251550)
				(5,40.4000000000000) +- (0,1.51657508881031)
				(6,54.2000000000000) +- (0,1.09544511501033)
				(7,72.6000000000000) +- (0,1.34164078649987)
				(8,86.4000000000000) +- (0,2.60768096208106)
				(9,109.600000000000) +- (0,2.60768096208106)
				(10,126.400000000000) +- (0,1.94935886896179)
				(11,147.800000000000) +- (0,1.30384048104053)
				(12,168.800000000000) +- (0,4.76445169982864)
			};
			\addlegendentry{Empirical}

			\addplot[
			color = black,
			mark = none
			]	coordinates {
				( 3, 6.758225815471977e+03 )
				( 300, 6.704381736480477e+08 )
			}
			node[right,pos=1.01] {$r=\cO{n^{2.50}}$};

			\addplot[
			color = black,
			mark = none
			]	coordinates {
				( 3, 3.650696255838341e+02 )
				( 300, 4.830357154306989e+05 )
			}
			node[right,pos=1.01] {$r=\cO{n^{1.56}}$};

			\addplot[
			color = black,
			mark = none
			]	coordinates {
				( 3, 1.052947651247897e+02 )
				( 300, 1.167366537578390e+05 )
			}
			node[right,pos=1.01] {$r=\cO{n^{1.52}}$};

			\addplot[
			color = black,
			mark = none
			]	coordinates {
				( 3, 17.614419211410862 )
				( 300, 3.343574191738355e+04 )
			}
			node[right,pos=1.01] {$r=\cO{n^{1.64}}$};

			\end{axis}
			\end{tikzpicture}
		}
	\end{subfigure}%
	~
	\begin{subfigure}{.5\linewidth}
		\resizebox{.95\textwidth}{!}{
			\begin{tikzpicture}
			\begin{axis}[
			log x ticks with fixed point,
			xtick={10,100},
			xmode=log,
			ymode=log,
			xmin = 3,
			xmax = 300,
			ymin=10^1,
			ymax=10^9,
			width=12cm,
			ymajorgrids=true,
			yminorgrids=true,
			legend style={at={(0.02,0.98)},anchor=north west},
			xlabel={$n$},
			ylabel={$r$}, ylabel near ticks,
			legend cell align={left},
			clip=false
			]

			\addlegendimage{empty legend}
			\addlegendentry[yshift=0pt]{\hspace{-.4cm}
			\textbf{\emph{X-Y-Z}}}

			\addplot[only marks,mark=square*,color=red] coordinates {
				(10,136921)
				(16,442697)
				(23,1095749)
				(32,2500049)
				(45,5859047)
				(64,14125505)
				(91,34037275)
				(128,79838219)
				(181,189775803)
				(256,451362181)
			};
			\addlegendentry{Analytic \cite{CMNRS18}}

			\addplot[only marks,mark=square*,color=blue] coordinates {
				(10,12485)
				(16,37919)
				(23,88370)
				(32,189647)
				(45,415390)
				(64,931115)
				(91,2082798)
				(128,4543406)
				(181,10036070)
				(256,22204888)
			};
			\addlegendentry{Commutator \cite{CMNRS18}}

			\addplot[only marks,mark=square*,color=bndclr,error bars/.cd,y dir=both,y explicit] coordinates {
				(10,968.066) +- (0,13.9814)
				(16,2019.88) +- (0,24.6252)
				(23,3505.5) +- (0,51.4801)
				(32,5834.46) +- (0,49.5033)
				(45,9783.66) +- (0,89.9138)
				(64,16681.5) +- (0,125.156)
				(91,28435.3) +- (0,235.595)
				(128,47460.9) +- (0,142.931)
				(181,79989.5) +- (0,222.254)
				(256,134466.) +- (0,979.267)
			};
			\addlegendentry{Our bound}

			\addplot[only marks,mark=square*,color=empclr,error bars/.cd,y dir=both,y explicit] coordinates {
				(4.,31.) +- (0,2.23607)
				(5.,42.4) +- (0,1.34164)
				(6.,60.4) +- (0,1.81659)
				(7.,76.8) +- (0,3.70135)
				(8.,99.2) +- (0,4.96991)
				(9.,115.) +- (0,2.82843)
				(10.,134.4) +- (0,3.91152)
				(11.,158.2) +- (0,5.01996)
				(12.,177.8) +- (0,2.77489)
			};
			\addlegendentry{Empirical}

			\addplot[
			color = black,
			mark = none
			]	coordinates {
				( 3, 6.758225815471977e+03 )
				( 300, 6.704381736480477e+08 )
			}
			node[right,pos=1.01] {$r=\cO{n^{2.50}}$};

			\addplot[
			color = black,
			mark = none
			]	coordinates {
				( 3, 8.022182443606885e+02 )
				( 300, 3.238711096677811e+07 )
			}
			node[right,pos=1.01] {$r=\cO{n^{2.30}}$};

			\addplot[
			color = black,
			mark = none
			]	coordinates {
				( 3,158.384 )
				( 300,173191.0)
			}
			node[right,pos=1.01] {$r=\cO{n^{1.52}}$};

			\addplot[
			color = black,
			mark = none
			]	coordinates {
				( 3,19.392 )
				( 300,33109.1)
			}
			node[right,pos=1.01] {$r=\cO{n^{1.62}}$};

			\end{axis}
			\end{tikzpicture}
		}
	\end{subfigure}%
	\caption{ Comparison of $r$ for the Heisenberg model using the analytic bound \cite[Proposition F.4]{CMNRS18}, commutator bound \cite[Theorem F.11]{CMNRS18}, locality-based bound \cite[Supplementary Material IV B]{CS19}, and our new bounds \prop{pf4_bound_2term} and \prop{pf4_bound_3term}. Error bars are omitted as they are negligibly small on the plot. Straight lines show power-law fits to the data. Note that the exponent for the empirical data is based on brute-force simulations of small systems, and thus may not precisely capture the true asymptotic scaling due to finite-size effects. \label{fig:pfbound}}
\end{figure}
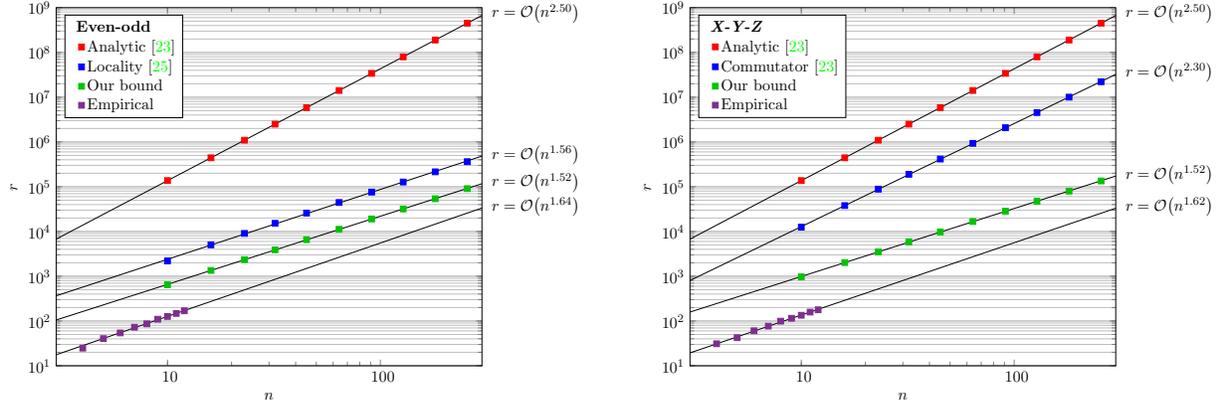

We find that the asymptotic scaling of our new bounds matches that of the empirical result up to finite-size effects and the prefactors are significantly tightened. At $n=10$, the Trotter number predicted by our bounds is loose only by a factor of $5.1$ for the even-odd ordering of terms and $7.2$ for the \textit{X-Y-Z} ordering. In comparison, the commutator bound of \cite{CMNRS18} only exploits the commutativity of the lowest-order term of the BCH series and is bottlenecked by the use of tail bounds. The previous bound \cite{CS19} based on geometrical locality is also uncompetitive since it cannot directly leverage the nested commutators of the Hamiltonian terms.

In addition to nearest-neighbor interactions, we also consider the simulation of a one-dimensional Heisenberg with power-law interactions:
\begin{align}
\label{eq:heisenberg-power-law}
	H = \sum_{j=1}^{n-1}\sum_{k = j+1}^{n} \frac{1}{\abs{j-k}^\alpha}(X_jX_k+Y_jY_k+Z_jZ_k) + \sum_{j=1}^{n-1} h_j Z_j,
\end{align}
where $h_j$ are again chosen randomly in $[-1,1]$ and $\alpha\geq 0$ is a constant.
Similarly to the case of nearest-neighbor interactions, we use the fourth-order product formula with \textit{X-Y-Z} ordering
\begin{equation}
\label{eq:x-y-z-power-law}
\begin{aligned}
H_1&=\sum_{j=1}^{n-1}\sum_{k = j+1}^{n} \frac{1}{\abs{j-k}^\alpha}X_jX_k,\\
H_2&=\sum_{j=1}^{n-1}\sum_{k = j+1}^{n} \frac{1}{\abs{j-k}^\alpha}Y_jY_k,\\
H_3&=\sum_{j=1}^{n-1}\sum_{k = j+1}^{n} \bigg(\frac{1}{\abs{j-k}^\alpha}Z_jZ_k+h_{j} Z_{j}\bigg)
\end{aligned}
\end{equation}
and compare the empirical Trotter number against that predicted by the best previous bound (\lem{trotter_error_one_norm_scaling}) and our new bound for simulating \eq{heisenberg-power-law} for time $t = n$ with accuracy $\epsilon = 10^{-3}$. We consider different values of $\alpha$, ranging from $\alpha = 0$ (strong power-law interactions) to $\alpha = 4$ (rapidly decaying power-law interactions).
We note that for $\alpha>2$, we simulate the evolution using the truncation algorithm [with the asymptotic gate complexity given by \cref{eq:power-law-gate-count-weak}], whereas for $\alpha\leq 2$, we simply simulate the entire Hamiltonian [with the asymptotic gate complexity in \cref{eq:power-law-pf}].

At $n=10$, the empirical Trotter numbers are $552\pm45$ ($\alpha = 0$) and $129\pm 6$ ($\alpha=4$), where the standard deviation is obtained by averaging over five instances of the random field $h_j$.
Meanwhile, our bound gives $5609\pm 3$ ($\alpha = 0$) and $885\pm 32$ ($\alpha=4$)---about 10.2 and 6.9 times looser compared to the empirical values respectively.
Our bound for power-law interactions with small $\alpha$ performs slightly worse than for the nearest-neighbor interactions, partly due to the fact that the triangle inequality is invoked more often for the power-law interactions.

We note that the number of interaction terms in a long-range interacting Hamiltonian scales as $n^2$, making it difficult to compute our bound exactly at large $n$.
Instead, we further upper bound the norm of the nested commutator using triangle inequalities and estimate this upper bound using a counting argument similarly to \cref{eq:power-law-pf}.
In \fig{pfbound-power-law}, we plot the empirical Trotter numbers against this counting bound for different values of $n$ at $\alpha = 0$ and $\alpha = 4$.
The figure shows that even our counting bound is tighter than the previous estimates at both values of $\alpha$. We leave a thorough study of the efficient numerical implementation of our bound as a subject for future work.

In addition, we plot in \fig{pfbound-power-law-exponents} the scaling exponents of the Trotter numbers as functions of $n$ at different values of $\alpha\in[0,4]$.
While the scaling exponent of the analytic bound in Ref.~\cite{CMNRS18} is loose and independent of $\alpha$, our bound appears to correctly capture the scaling of Trotter number at all values of $\alpha$.
In particular, it shows that the scaling exponent decreases with $\alpha$---indicating fewer resources needed to simulate faster decaying interactions---and approaches the value for simulating nearest-neighbor interactions at large $\alpha$.

\begin{figure}[t]
	\centering
	\begin{subfigure}{.5\linewidth}
		\resizebox{.95\textwidth}{!}{
			\begin{tikzpicture}
			\begin{axis}[
			log x ticks with fixed point,
			xtick={10,100},
			xmode=log,
			ymode=log,
			xmin = 3,
			xmax = 300,
			ymin=10^1,
			ymax=10^12,
			width=12cm,
			ymajorgrids=true,
			yminorgrids=true,
			legend style={at={(0.02,0.98)},anchor=north west},
			xlabel={$n$},
			ylabel={$r$}, ylabel near ticks,
			legend cell align={left},
			clip=false
			]

			\addlegendimage{empty legend}
			\addlegendentry[yshift=0pt]{\hspace{-.4cm}
			\textbf{Power-law $\alpha = 0$}}

			\addplot[only marks,mark=square*,color=red] coordinates {
				(10,683492)
				(11,980607)
				(16,4040834)
				(23,15871513)
				(32,55007474)
				(45,198192734)
				(64,744293885)
				(91,2790519047)
				(128,10041413052)
				(181,36845915463)
				(256,135278265598)
			};
			\addlegendentry{Analytic \cite{CMNRS18}}

			\addplot[only marks,mark=square*,color=blue,error bars/.cd,y dir=both,y explicit] coordinates {
				(10.,85000) +- (0,2060.96009180188)
				(11.,120468.800000000) +- (0,3347.21962231342)
				(16.,496058.600000000) +- (0,5220.44474159051)
				(23.,1923301.20000000) +- (0,3520.03234360141)
				(32.,6681200.20000000) +- (0,27385.6615804694)
				(45.,23859805.2000000) +- (0,47387.3234388692)
				(64.,89508455.8000000) +- (0,142973.959654197)
				(91.,335501843.200000) +- (0,233017.926575832)
				(128.,1205555294.40000) +- (0,622074.416221886)
				(181.,4419630596.80000) +- (0,1379160.76385558)
				(256.,16220417081.0000) +- (0,1888521.59401435)
			};
			\addlegendentry{$1$-norm (\lem{trotter_error_one_norm_scaling})}

			\addplot[only marks,mark=square*,color=looseclr,error bars/.cd,y dir=both,y explicit] coordinates {
				(10.,28778.3) +- (0,0)
				(11.,38215.2) +- (0,0)
				(16.,114348.) +- (0,0)
				(23.,323591.) +- (0,0)
				(32.,823998.) +- (0,0)
				(45.,2145001.) +- (0,0)
				(64.,5729269.) +- (0,0)
				(91.,15228883.) +- (0,0)
				(128.,39173729.) +- (0,0)
				(181.,102052601.) +- (0,0)
				(256.,265651944.) +- (0,0)
			};
			\addlegendentry{Counting bound [\cref{eq:power-law-pf}]}

			\addplot[only marks,mark=square*,color=bndclr,error bars/.cd,y dir=both,y explicit] coordinates {
				(5.,732.884) +- (0,20.4298)
				(6.,1342.21) +- (0,8.25211)
				(7.,2040.18) +- (0,14.7983)
				(8.,3070.58) +- (0,14.2236)
				(9.,4190.92) +- (0,35.5577)
				(10.,5609.49) +- (0,2.58875)
				(11.,7310.94) +- (0,18.5933)
			};
			\addlegendentry{Our bound}

			\addplot[only marks,mark=square*,color=empclr,error bars/.cd,y dir=both,y explicit] coordinates {
				(5.,72.4) +- (0,6.50385)
				(6.,129.4) +- (0,11.6103)
				(7.,199.4) +- (0,13.8492)
				(8.,297.6) +- (0,25.6281)
				(9.,400.) +- (0,28.3373)
				(10.,552.6) +- (0,45.0533)
				(11.,717.2) +- (0,59.0271)
			};
			\addlegendentry{Empirical}

			\addplot[
			color = black,
			mark = none
			]	coordinates {
				( 3, 7.446008953978934e+03 )
				( 300, 2.472033862130378e+11 )
			}
			node[right,pos=1.01] {$r=\cO{n^{3.75}}$};

			\addplot[
			color = black,
			mark = none
			]	coordinates {
				( 3, 9.270537091830635e+02 )
				( 300, 2.941819043177504e+10 )
			}
			node[right,pos=1.01] {$r=\cO{n^{3.75}}$};

			\addplot[
			color = black,
			mark = none
			]	coordinates {
				( 3, 17.0912 )
				( 300, 1.01249e+7 )
			}
			node[right,pos=1.01] {$r=\cO{n^{2.89}}$};

			\addplot[
			color = black,
			mark = none
			]	coordinates {
				( 3, 174.861 )
				( 300, 1.05318e+8 )
			}
			node[right,pos=1.01] {$r=\cO{n^{2.84}}$};

			\addplot[
			color = black,
			mark = none
			]	coordinates {
				( 3, 924.864 )
				( 300, 4.48112e+8 )
			}
			node[right,pos=1.01] {$r=\cO{n^{2.84}}$};

			\end{axis}
			\end{tikzpicture}
		}
	\end{subfigure}%
	~
	\begin{subfigure}{.5\linewidth}
		\resizebox{.95\textwidth}{!}{
			\begin{tikzpicture}
			\begin{axis}[
			log x ticks with fixed point,
			xtick={10,100},
			xmode=log,
			ymode=log,
			xmin = 3,
			xmax = 300,
			ymin=10^1,
			ymax=10^12,
			width=12cm,
			ymajorgrids=true,
			yminorgrids=true,
			legend style={at={(0.02,0.98)},anchor=north west},
			xlabel={$n$},
			ylabel={$r$}, ylabel near ticks,
			legend cell align={left},
			clip=false
			]

			\addlegendimage{empty legend}
			\addlegendentry[yshift=0pt]{\hspace{-.4cm}
				\textbf{Power-law $\alpha = 4$}}

			\addplot[only marks,mark=square*,color=red] coordinates {
				(10,683492)
				(11,980607)
				(16,4040834)
				(23,15871513)
				(32,55007474)
				(45,198192734)
				(64,744293885)
				(91,2790519047)
				(128,10041413052)
				(181,36845915463)
				(256,135278265598)
			};
			\addlegendentry{Analytic \cite{CMNRS18}}

			\addplot[only marks,mark=square*,color=blue,error bars/.cd,y dir=both,y explicit] coordinates {
				(10.,18913.6000000000) +- (0,680.995447855564)
				(11.,21784.6000000000) +- (0,1753.39450780479)
				(16.,61246.2000000000) +- (0,3556.32348078743)
				(23.,147230.200000000) +- (0,12054.4694491296)
				(32.,348796.200000000) +- (0,8317.50745115386)
				(45.,818638.800000000) +- (0,31711.0553072584)
				(64.,2049576) +- (0,50987.8100775077)
				(91.,4928611.60000000) +- (0,137565.090452847)
				(128.,11664013.2000000) +- (0,114699.377464309)
				(181.,27831397.2000000) +- (0,287070.041750615)
				(256.,66314375.6000000) +- (0,887453.412779679)
			};
			\addlegendentry{$1$-norm (\lem{trotter_error_one_norm_scaling})}

			\addplot[only marks,mark=square*,color=looseclr,error bars/.cd,y dir=both,y explicit] coordinates {
				(10.,1782.61) +- (0,0)
				(11.,2055.5) +- (0,0)
				(16.,3605.38) +- (0,0)
				(23.,6226.75) +- (0,0)
				(32.,10179.8) +- (0,0)
				(45.,17025.6) +- (0,0)
				(64.,28830.1) +- (0,0)
				(91.,48927.9) +- (0,0)
				(128.,81561.6) +- (0,0)
				(181.,137012.) +- (0,0)
				(256.,230528.) +- (0,0)
			};
			\addlegendentry{Counting bound [\cref{eq:power-law-pf}]}

			\addplot[only marks,mark=square*,color=bndclr,error bars/.cd,y dir=both,y explicit] coordinates {
				(5.,278.005) +- (0,13.2403)
				(6.,385.801) +- (0,11.5145)
				(7.,505.685) +- (0,22.7519)
				(8.,619.989) +- (0,11.2582)
				(9.,752.517) +- (0,22.4373)
				(10.,885.729) +- (0,32.2218)
				(11.,1011.49) +- (0,12.2703)
			};
			\addlegendentry{Our bound}

			\addplot[only marks,mark=square*,color=empclr,error bars/.cd,y dir=both,y explicit] coordinates {
				(5.,41.2) +- (0,2.77489)
				(6.,56.4) +- (0,0.547723)
				(7.,71.8) +- (0,2.16795)
				(8.,91.8) +- (0,6.37966)
				(9.,111.2) +- (0,5.21536)
				(10.,129.) +- (0,6.48074)
				(11.,151.2) +- (0,7.59605)
			};
			\addlegendentry{Empirical}

			\addplot[
			color = black,
			mark = none
			]	coordinates {
				( 3, 7.446008953978934e+03 )
				( 300, 2.472033862130378e+11 )
			}
			node[right,pos=1.01] {$r=\cO{n^{3.75}}$};

			\addplot[
			color = black,
			mark = none
			]	coordinates {
				( 3, 8.609301558039709e+02 )
				( 300, 1.005245258332178e+08 )
			}
			node[right,pos=1.01] {$r=\cO{n^{2.53}}$};

			\addplot[
			color = black,
			mark = none
			]	coordinates {
				( 3,17.8835 )
				( 300,35705.9)
			}
			node[right,pos=1.01] {$r=\cO{n^{1.65}}$};

			\addplot[
			color = black,
			mark = none
			]	coordinates {
				( 3,123.296 )
				( 300,232629.)
			}
			node[right,pos=1.01] {$r=\cO{n^{1.64}}$};

			\addplot[
			color = black,
			mark = none
			]	coordinates {
				( 3,293.059 )
				( 300,292487.)
			}
			node[above right,pos=1.01] {$r=\cO{n^{1.5}}$};

			\end{axis}
			\end{tikzpicture}
		}
	\end{subfigure}%
	\caption{ Comparison of $r$ for the power-law Heisenberg model using the analytic bound \cite[Proposition F.4]{CMNRS18}, $1$-norm bound \lem{trotter_error_one_norm_scaling}, a bound from counting argument \eq{power-law-pf}, and our bound \prop{pf4_bound_3term}. Error bars are omitted as they are negligibly small on the plot. Straight lines show power-law fits to the data. Note that the exponent for the empirical data is based on brute-force simulations of small systems, and thus may not precisely capture the true asymptotic scaling due to finite-size effects. \label{fig:pfbound-power-law}}
\end{figure}
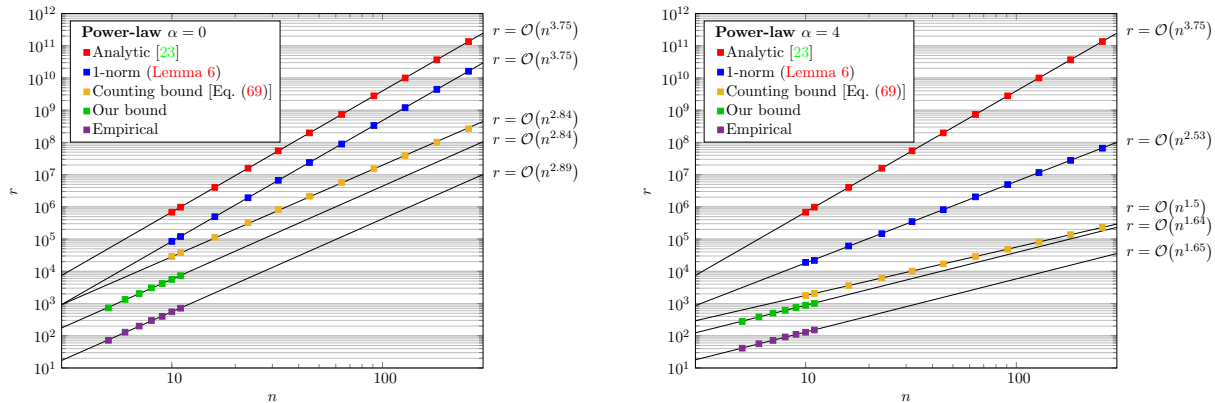


\begin{figure}[t]
	\centering
	\begin{subfigure}{.525\linewidth}
		\resizebox{.95\textwidth}{!}{
			\begin{tikzpicture}
			\begin{axis}[
			xtick={0,1,2,3,4},
			xmode=normal,
			ymode=normal,
			xmin = 0,
			xmax = 4,
			ymin=1.25,
			ymax=4,
			width=12cm,
			ymajorgrids=true,
			yminorgrids=true,
			legend style={at={(0.98,0.78)},anchor=north east},
			xlabel={$\alpha$},
			ylabel={Scaling exponents of the Trotter numbers}, ylabel near ticks,
			legend cell align={left},
			clip=false
			]

			\addlegendimage{empty legend}
			\addlegendentry[yshift=0pt]{\hspace{-.4cm}
			\textbf{Power-law $\alpha = 0$}}

			\addplot[
			color = red,
			style = dashdotted,
			mark = none
			]	coordinates {
				( 0,3.76)
				( 4,3.76)
			};
			\addlegendentry{Analytic~\cite{CMNRS18}}

			\addplot[only marks,mark=square*,color=bndclr,error bars/.cd,y dir=both,y explicit] coordinates {
				(0.,2.86454) +- (0,0)
				(0.1,2.74736) +- (0,0)
				(0.2,2.63258) +- (0,0)
				(0.3,2.52044) +- (0,0)
				(0.4,2.41193) +- (0,0)
				(0.5,2.30751) +- (0,0)
				(0.6,2.20833) +- (0,0)
				(0.7,2.11507) +- (0,0)
				(0.8,2.02843) +- (0,0)
				(0.9,1.9495) +- (0,0)
				(1.,1.87862) +- (0,0)
				(1.1,1.81584) +- (0,0)
				(1.2,1.76143) +- (0,0)
				(1.3,1.71496) +- (0,0)
				(1.4,1.67556) +- (0,0)
				(1.5,1.64311) +- (0,0)
				(1.6,1.61611) +- (0,0)
				(1.7,1.59376) +- (0,0)
				(1.8,1.57611) +- (0,0)
				(1.9,1.56152) +- (0,0)
				(2.,1.55039) +- (0,0)
				(2.1,1.54059) +- (0,0)
				(2.2,1.5331) +- (0,0)
				(2.3,1.52696) +- (0,0)
				(2.4,1.52249) +- (0,0)
				(2.5,1.51837) +- (0,0)
				(2.6,1.51536) +- (0,0)
				(2.7,1.51307) +- (0,0)
				(2.8,1.51073) +- (0,0)
				(2.9,1.5095) +- (0,0)
				(3.,1.50837) +- (0,0)
				(3.1,1.5074) +- (0,0)
				(3.2,1.5059) +- (0,0)
				(3.3,1.5058) +- (0,0)
				(3.4,1.50472) +- (0,0)
				(3.5,1.50427) +- (0,0)
				(3.6,1.50467) +- (0,0)
				(3.7,1.50454) +- (0,0)
				(3.8,1.50337) +- (0,0)
				(3.9,1.50358) +- (0,0)
				(4.,1.50321) +- (0,0)
			};
			\addlegendentry{Our Bound}

			\addplot[only marks,mark=square*,color=empclr,error bars/.cd,y dir=both,y explicit] coordinates {
				(0.,2.88694) +- (0,0.115278)
				(0.1,2.69508) +- (0,0.108969)
				(0.2,2.6383) +- (0,0.0868982)
				(0.3,2.43989) +- (0,0.0579215)
				(0.4,2.4013) +- (0,0.151287)
				(0.5,2.31783) +- (0,0.0672562)
				(0.6,2.27619) +- (0,0.0763851)
				(0.7,2.2184) +- (0,0.0567153)
				(0.8,2.20326) +- (0,0.120509)
				(0.9,2.0731) +- (0,0.0832915)
				(1.,2.04945) +- (0,0.0960933)
				(1.1,1.97028) +- (0,0.0280553)
				(1.2,1.93393) +- (0,0.0905912)
				(1.3,1.87117) +- (0,0.0689264)
				(1.4,1.7873) +- (0,0.0388334)
				(1.5,1.78035) +- (0,0.0795774)
				(1.6,1.73421) +- (0,0.0346347)
				(1.7,1.71697) +- (0,0.0788095)
				(1.8,1.67486) +- (0,0.0374551)
				(1.9,1.68339) +- (0,0.0453289)
				(2.,1.59375) +- (0,0.0196314)
				(2.1,1.55765) +- (0,0.0338484)
				(2.2,1.57606) +- (0,0.106149)
				(2.3,1.60516) +- (0,0.0638376)
				(2.4,1.55936) +- (0,0.0811354)
				(2.5,1.55992) +- (0,0.0548177)
				(2.6,1.54867) +- (0,0.0290067)
				(2.7,1.61791) +- (0,0.0736941)
				(2.8,1.59342) +- (0,0.0533627)
				(2.9,1.58904) +- (0,0.117669)
				(3.,1.5534) +- (0,0.0574723)
				(3.1,1.55667) +- (0,0.106389)
				(3.2,1.6106) +- (0,0.0699574)
				(3.3,1.60269) +- (0,0.0470481)
				(3.4,1.62264) +- (0,0.0237255)
				(3.5,1.65281) +- (0,0.0436646)
				(3.6,1.63904) +- (0,0.0341563)
				(3.7,1.58655) +- (0,0.0679165)
				(3.8,1.66013) +- (0,0.0703377)
				(3.9,1.57454) +- (0,0.058028)
				(4.,1.65013) +- (0,0.0718613)
			};
			\addlegendentry{Empirical}

			\addplot[
			color = blue,
			style = dashed,
			mark = none
			]	coordinates {
				( 0,1.5)
				( 4,1.5)
			};
			\addlegendentry{$\alpha\rightarrow\infty$ limit}
			\end{axis}
			\end{tikzpicture}
		}
	\end{subfigure}%
	~

	\caption{ Comparison of the empirical scaling exponents of the Trotter numbers (purple squares) against our bound (green squares) as functions of the system size at different values of the power-law exponent $\alpha$.
	The error bars of the empirical values represent the standard deviation of the fitted exponents (See \fig{pfbound-power-law}) over five instances of the random field $h_j$.
	The bound is derived from the counting argument in \cref{eq:alpha_comm_one_norm} and therefore has no standard deviation.
	We attribute the systematic difference between our bound and the empirical values to the fact that we only compute the empirical Trotter numbers up to $n = 11$, which may not capture the precise asymptotic scaling in the large-$n$ limit.
	We also include the scaling exponent of the analytic bound in Ref.~\cite{CMNRS18} (red dash-dotted line) as well as the theoretical exponent in the limit $\alpha\rightarrow\infty$, i.e. nearest-neighbor interactions (blue dashed line), for references.
	 \label{fig:pfbound-power-law-exponents}}
\end{figure}
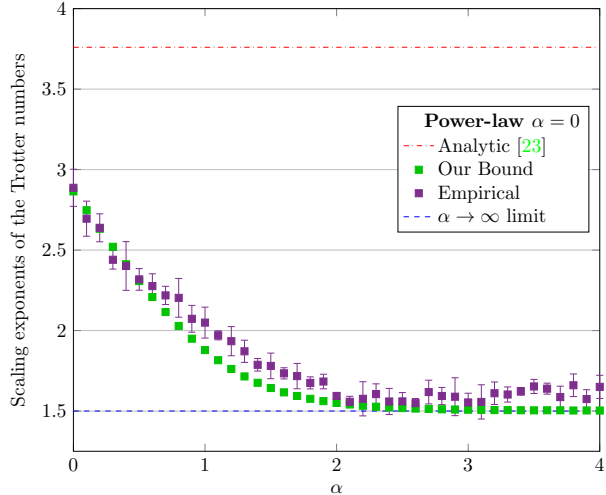


\section{Discussion}
\label{sec:discussion}

We have developed a general theory of Trotter error and identified a host of applications to simulating quantum dynamics, local observables, and quantum Monte Carlo methods. We work with arbitrary finite-dimensional operators as opposed to anti-Hermitian ones, which makes our theory applicable to both real- and imaginary-time evolutions. We consider Trotter error of various types, including additive error, multiplicative error, and error that appears in the exponent. For each type, we apply the correct order condition to cancel lower-order terms, and represent higher-order ones as explicit nested commutators. The list of applications presented herein is not intended to be exhaustive and we believe our techniques can uncover more speedups of the product-formula algorithm that were previously unknown \cite{XSSS20,Shaw2020quantumalgorithms}.

Compared to the analysis of other simulation algorithms such as the truncated Taylor-series algorithm~\cite{BCCKS14} and the qubitization approach~\cite{LC16}, the derivation of our Trotter error theory is considerably more involved. However, the resulting error bounds are succinct and easy to evaluate.  \thm{trotter_error_comm_scaling} shows that Trotter error incurred by decomposing the evolution generated by $H=\sum_{\gamma=1}^{\Gamma}H_\gamma$ depends asymptotically on the quantity $\acommtilde=\sum_{\gamma_1,\gamma_2,\ldots,\gamma_{p+1}}\norm{\big[H_{\gamma_{p+1}},\cdots\big[H_{\gamma_2},H_{\gamma_1}\big]\big]}$, which can be computed by induction as for the second-quantized plane-wave electronic structure, $k$-local Hamiltonians, rapidly decaying interacted systems, clustered Hamiltonians, transverse field Ising model, and quantum ferromagnetic spin systems. We further show how to improve the analysis to find error bounds with small constant prefactors. Numerical simulation suggests that our higher-order error bounds are close to tight for systems with nearest-neighbor and power-law interactions, and we hope future work can explore their tightness for other systems.

Our result shows that high-order product formulas can be advantageous for simulating many physical systems. Interestingly, we can often achieve this advantage without using a formula of very large order. For $d$-dimensional power-law interactions with exponent $\alpha>2d$, we have shown that the $p$th-order product-formula algorithm has gate complexity $\cO{(nt)^{1+d/(\alpha-d)+1/p}}$, whereas the state-of-the-art Lieb-Robinson-based approach requires $\tildecO{(nt)^{1+2d/(\alpha-d)}}$ gates. Product formulas can thus scale better if $p\geq (\alpha-d)/d$, which is small for various physical systems such as the dipole-dipole interactions ($\alpha=3$) and the van der Waals interactions ($\alpha=6$). For other systems such as nearest-neighbor interactions and electronic-structure Hamiltonians, product formulas do not exactly match the state-of-the-art result in terms of the asymptotic scaling, but they are still advantageous for simulating systems of small sizes \cite{CMNRS18,Kivlichan19}.

The complexity of the product-formula approach is determined by both the Trotter number (or Trotter error) and the cost per Trotter step. A naive implementation of each Trotter step exponentiates all the terms in the Hamiltonian, which has a cost that scales with the total number of terms. However, this worst-case complexity can be avoided by truncating the original Hamiltonian, as we have demonstrated in the simulation of rapidly decaying power-law Hamiltonians. Recent studies have proposed other techniques for implementing Trotter steps \cite{WBCHT14,Kivlichan19,AA17,KMWGACB18,CBC20}. Those techniques can be applied in combination with our Trotter error analysis to further speed up the product-formula algorithm.

We have restricted our attention to the evolutions generated by time-independent operators. In the more general case, we have an operator-valued function $\mathscr{H}(\tau)=\sum_{\gamma=1}^{\Gamma}\mathscr{H}_\gamma(\tau)$ and our goal is to simulate the time-ordered evolution $\expT\big(\int_{0}^{t}\mathrm{d}\tau\sum_{\gamma=1}^{\Gamma}\mathscr{H}_\gamma(\tau)\big)$ \cite{WBHS10,PQSV11,FractionalQuery14,BCCKS14,Kieferova18,LW18,BCSWW19}. Under certain smoothness assumptions, Reference \cite{WBHS10} shows that this evolution can be simulated using product formulas, although their analysis does not exploit the commutativity of operator summands. We believe our approach can be extended to give improved analysis for time-dependent Hamiltonian simulation, but we leave a detailed study for future work \cite{AFL20}.

Previous work considered several generalized product formulas, such as ones based on the divide-and-conquer construction \cite{HP18}, the randomized construction \cite{COS18,OWC19}, and the linear-combination-of-unitaries construction \cite{LKW19,FSKKE21,childs2012hamiltonian}. The common underlying idea is to approximate the ideal evolution to $p$th order using formulas of order $q_k$, where $q_k\leq p$. Our theory can be applied to represent the $q_k$th-order Trotter error in terms of nested commutators, thus improving the analyses of \cite{HP18,COS18,OWC19,LKW19,FSKKE21,childs2012hamiltonian}. This leads to a better understanding of these generalized formulas and justifies their potential utility in quantum simulation.

Several other questions related to our theory deserve further investigation. For example, the spectral-norm error bound computed here would be overly pessimistic if we simulate with a low-energy initial state. It would then be beneficial to change the error metric to avoid the worst-case error propagation \cite{SS20,SHC20,AFL20}. Our analysis has also assumed an operator decomposition $H=\sum_{\gamma=1}^{\Gamma}H_\gamma$ given a prior, but one may instead seek an alternative decomposition to maximize the commutativity of operator summands. We focus on the error analysis within each Trotter step and apply the triangle inequality across different steps, which may be improved upon as hinted in \cite{HHZ19,SOEHHHZ19,TSCT21,HLMJH20}. Finally, we have considered an idealized setting, and we hope future work could take the effect of noise into account \cite{MYB19}.

Product formulas provide arguably the most straightforward approach to simulating quantum systems. This approach is empirically advantageous and is often the method of choice for near-term demonstration of quantum simulation. Despite their experimental success, the error scaling of product formulas was poorly understood and, prior to our work, their advantage was only rigorously analyzed for a restricted collection of systems. The theory developed here represents progress toward a precise characterization of Trotter error, which we hope will bridge the gap between theoretical investigation and experimental realization of quantum simulation.

\section*{Acknowledgements}
Y.S.\ thanks Ryan Babbush, Fernando G.S.L.\ Brand\~{a}o, Xun Gao, Jens Eisert, and Brian Swingle for inspiring discussions. He also received useful feedback during his visits to the Institute for Quantum Computing, Center for Computation and Technology, and Lawrence Berkeley National Laboratory. We thank Daochen Wang, Lin Lin, ‪Dong An‬, and anonymous referees for their comments on an earlier draft.

This work was supported in part by the U.S.\ Department of Energy (DOE), Office of Science, Office of Advanced Scientific Computing Research, Quantum Algorithms Teams and Accelerated Research in Quantum Computing programs; by the Army Research Office (MURI Grant No.\ W911NF-16-1-0349); and by the National Science Foundation (Grant No.\ CCF-1813814). In addition, Y.S.\ is supported by the Google Ph.D.\ Fellowship program, the National Science Foundation RAISE-TAQS 1839204 and Amazon Web Services, AWS Quantum Program. The Institute for Quantum Information and Matter is an NSF Physics Frontiers Center PHY-1733907. M.C.T.\ also acknowledges funding by the DOE BES QIS program (Award No.\ DE-SC0019449), the NSF PFCQC program, the DOE ASCR Quantum Testbed Pathfinder program (Award No.\ DE-SC0019040), AFOSR, ARL CDQI, and NSF PFC at JQI. N.W.\ is supported in part by a Google Quantum Research Award. S.C.Z.\ acknowledges support from NSF CAREER Grant No.\ CCF-1845125.
\appendix
\section{Error types}
\label{append:type}

In this appendix, we consider different types of Trotter error for a general product formula introduced in \sec{theory_type}. In particular, we will prove \thm{error_type} that gives explicit expressions for three types of Trotter error: the additive error, the multiplicative error, and the error that appears in the exponent of a time-ordered exponential. These types are equivalent for analyzing the complexity of simulating quantum dynamics and local observables, but the latter two are more versatile for quantum Monte Carlo simulation.

Let $H=\sum_{\gamma=1}^{\Gamma}H_\gamma$ be an operator with $\Gamma$ summands. We decompose the evolution $e^{tH}=e^{t\sum_{\gamma=1}^{\Gamma}H_\gamma}$ using a general product formula $\mathscr{S}(t)=\prod_{\upsilon=1}^{\Upsilon}\prod_{\gamma=1}^{\Gamma}e^{ta_{(\upsilon,\gamma)}H_{\pi_{\upsilon}(\gamma)}}$. We impose the lexicographical order on the tuples $(\upsilon,\gamma)$ as in \sec{theory_type}, so that
\begin{align}
\mathscr{S}(t) = \prod_{(\upsilon,\gamma)}^{\longleftarrow}e^{ta_{(\upsilon,\gamma)}H_{\pi_{\upsilon}(\gamma)}}.
\end{align}

To compute the additive error, we construct the differential equation
\begin{equation}
\frac{\mathrm{d}}{\mathrm{d}t}\mathscr{S}(t)=H\mathscr{S}(t)+\mathscr{R}(t),
\end{equation}
with initial condition $\mathscr{S}(0)=I$, where
\begin{equation}
\begin{aligned}
\mathscr{R}(t)
:=&\sum_{(\upsilon,\gamma)}\prod_{(\upsilon',\gamma')\succ(\upsilon,\gamma)}^{\longleftarrow}e^{ta_{(\upsilon',\gamma')}H_{\pi_{\upsilon'}(\gamma')}}\big(a_{(\upsilon,\gamma)}H_{\pi_{\upsilon}(\gamma)}\big)\prod_{(\upsilon',\gamma')\preceq(\upsilon,\gamma)}^{\longleftarrow}e^{ta_{(\upsilon',\gamma')}H_{\pi_{\upsilon'}(\gamma')}}\\
&-H\prod_{(\upsilon',\gamma')}^{\longleftarrow}e^{ta_{(\upsilon',\gamma')}H_{\pi_{\upsilon'}(\gamma')}}.
\end{aligned}
\end{equation}
By the variation-of-parameters formula (\lem{td_Duhamel}), $\mathscr{S}(t)-e^{tH}=\int_{0}^{t}\mathrm{d}\tau\ e^{(t-\tau)H}\mathscr{R}(\tau)$, so we obtain the additive error
\begin{equation}
\mathscr{A}(t):=\int_{0}^{t}\mathrm{d}\tau\ e^{(t-\tau)H}\mathscr{R}(\tau).
\end{equation}
This suffices if our purpose is to only compute the additive error operator. However, for the later discussion in \append{rep}, it is convenient to further rewrite
\begin{equation}
\mathscr{A}(t)=\int_{0}^{t}\mathrm{d}\tau\ e^{(t-\tau)H}\mathscr{S}(\tau)\mathscr{T}(\tau),
\end{equation}
where
\begin{equation}
\begin{aligned}
\mathscr{T}(\tau)
:=&\sum_{(\upsilon,\gamma)}\prod_{(\upsilon',\gamma')\prec(\upsilon,\gamma)}^{\longrightarrow}e^{-\tau a_{(\upsilon',\gamma')}H_{\pi_{\upsilon'}(\gamma')}}\big(a_{(\upsilon,\gamma)}H_{\pi_{\upsilon}(\gamma)}\big)\prod_{(\upsilon',\gamma')\prec(\upsilon,\gamma)}^{\longleftarrow}e^{\tau a_{(\upsilon',\gamma')}H_{\pi_{\upsilon'}(\gamma')}}\\
&-\prod_{(\upsilon',\gamma')}^{\longrightarrow}e^{-\tau a_{(\upsilon',\gamma')}H_{\pi_{\upsilon'}(\gamma')}}H\prod_{(\upsilon',\gamma')}^{\longleftarrow}e^{\tau a_{(\upsilon',\gamma')}H_{\pi_{\upsilon'}(\gamma')}}.
\end{aligned}
\end{equation}
Note that we have rewritten part of the error operator as a linear combination of conjugation of matrix exponentials. In \append{rep}, we apply the correct order condition to further represent it as nested commutators of the operator summands $H_\gamma$.

For the exponentiated type of Trotter error, we aim to construct an operator-valued function $\mathscr{E}(t)$ such that
\begin{equation}
\mathscr{S}(t)=\expT\bigg(\int_{0}^{t}\mathrm{d}\tau\big(H+\mathscr{E}(\tau)\big)\bigg).
\end{equation}
To do this, we differentiate the product formula $\mathscr{S}(t)$ and obtain
\begin{equation}
\begin{aligned}
\frac{\mathrm{d}}{\mathrm{d}t}\mathscr{S}(t)
&=\sum_{(\upsilon,\gamma)}\prod_{(\upsilon',\gamma')\succ(\upsilon,\gamma)}^{\longleftarrow}e^{ta_{(\upsilon',\gamma')}H_{\pi_{\upsilon'}(\gamma')}}\big(a_{(\upsilon,\gamma)}H_{\pi_{\upsilon}(\gamma)}\big)\prod_{(\upsilon',\gamma')\preceq(\upsilon,\gamma)}^{\longleftarrow}e^{ta_{(\upsilon',\gamma')}H_{\pi_{\upsilon'}(\gamma')}}\\
&=\mathscr{F}(t)\mathscr{S}(t),
\end{aligned}
\end{equation}
where
\begin{equation}
\mathscr{F}(t):=\sum_{(\upsilon,\gamma)}\prod_{(\upsilon',\gamma')\succ(\upsilon,\gamma)}^{\longleftarrow}e^{ta_{(\upsilon',\gamma')}H_{\pi_{\upsilon'}(\gamma')}}\big(a_{(\upsilon,\gamma)}H_{\pi_{\upsilon}(\gamma)}\big)\prod_{(\upsilon',\gamma')\succ(\upsilon,\gamma)}^{\longrightarrow}e^{-ta_{(\upsilon',\gamma')}H_{\pi_{\upsilon'}(\gamma')}}.
\end{equation}
Applying the fundamental theorem of time-ordered evolution (\lem{fte}), we have
\begin{equation}
\mathscr{S}(t)=\expT\bigg(\int_{0}^{t}\mathrm{d}\tau\ \mathscr{F}(\tau)\bigg),
\end{equation}
which gives the exponentiated error
\begin{equation}
\mathscr{E}(t):=\mathscr{F}(t)-H.
\end{equation}

From the exponentiated type of Trotter error, we can obtain the multiplicative error by switching to the interaction picture. Specifically, we apply \lem{interaction_picture} and get
\begin{equation}
\mathscr{S}(t)=\expT\bigg(\int_{0}^{t}\mathrm{d}\tau \big(H+\mathscr{E}(\tau)\big)\bigg)
=e^{tH}\expT\bigg(\int_{0}^{t}\mathrm{d}\tau\ e^{-\tau H}\mathscr{E}(\tau)e^{\tau H}\bigg).
\end{equation}
Then, the operator-valued function
\begin{equation}
\mathscr{M}(t):=\expT\bigg(\int_{0}^{t}\mathrm{d}\tau\ e^{-\tau H}\mathscr{E}(\tau)e^{\tau H}\bigg)-I
\end{equation}
is the multiplicative error of the product formula. We have thus established \thm{error_type}, which we restate below.

\thmtype*

\section{Order conditions}
\label{append:order}

In this appendix, we continue the discussion in \sec{theory_order} about order conditions of Trotter error. We show how to use these conditions to cancel low-order terms of the Taylor series. Toward the end of this section, we establish \thm{error_order_cond} that gives order conditions for the additive, multiplicative, and exponentiated Trotter error. We apply these conditions to prove the main result on the commutator scaling of Trotter error.

Recall from \sec{theory_order} that the order condition of an operator-valued function $\mathscr{F}(\tau)$ represents the rate at which $\mathscr{F}(\tau)$ approaches zero when $\tau\rightarrow0$. Formally, given a continuous operator-valued function $\mathscr{F}(\tau)$ defined on $\R$, we write $\mathscr{F}(\tau)=O(\tau^{p})$ with nonnegative integer $p$ if there exist constants $c,t_0>0$, independent of $\tau$, such that $\norm{\mathscr{F}(\tau)}\leq c\abs{\tau}^p$ whenever $\abs{\tau}\leq t_0$. To verify this, it suffices to check that the limit
\begin{equation}
	\lim\limits_{\tau\rightarrow 0}\frac{\norm{\mathscr{F}(\tau)}}{\abs{\tau}^{p}}
\end{equation}
exists.

As aforementioned, our approach uses the order condition $\mathscr{F}(\tau)=\OO{\tau^p}$ to argue that terms with order $1,\tau,\ldots,\tau^{p-1}$ vanish in the Taylor series of $\mathscr{F}(\tau)$. This argument is rigorized in \cite[Lemma 6]{CS19}, which we restate and prove for completeness.

\begin{lemma}[Derivative condition]
	\label{lem:order_cond_deriv}
	Any continuous operator-valued function $\mathscr{F}(\tau)$ defined on $\R$ satisfies the order condition
	\begin{equation}
		\mathscr{F}(\tau)=\OO{1}.
	\end{equation}
	Furthermore, if $\mathscr{F}(\tau)$ has $p$ continuous derivatives for some positive integer $p$, then the following two conditions are equivalent:
	\begin{enumerate}
		\item $\mathscr{F}(\tau)=O(\tau^{p})$; and
		\item $\mathscr{F}(0)=\mathscr{F}'(0)=\cdots=\mathscr{F}^{(p-1)}(0)=0$.
	\end{enumerate}
\end{lemma}
\begin{proof}
	The continuity of $\mathscr{F}(\tau)$ at $\tau=0$ implies $\mathscr{F}(\tau)=\OO{1}$ by definition. Assume that $\mathscr{F}(\tau)$, $\mathscr{F}'(\tau)$,..., $\mathscr{F}^{(p)}(\tau)$ exist and are continuous. If Condition $2$ holds, we have
	\begin{equation}
		\lim_{\tau\rightarrow 0}\frac{\norm{\mathscr{F}(\tau)}}{\abs{\tau}^{p}}
		=\norm{\lim_{\tau\rightarrow 0}\frac{\mathscr{F}(\tau)}{\tau^{p}}}
		=\frac{\norm{\mathscr{F}^{(p)}(0)}}{p!}
	\end{equation}
	by the L'H\^{o}pital's rule. This proves that Condition $1$ holds.
	
	Given Condition $1$, we have by definition that
	\begin{equation}
		\norm{\mathscr{F}(\tau)}\leq c\abs{\tau}^{p}
	\end{equation}
	for some $c,t_0>0$ and all $\abs{\tau}\leq t_0$. Suppose by contradiction that Condition $2$ is not true. Then we let $0\leq j\leq p-1$ be the first integer for which $\mathscr{F}^{(j)}(0)\neq 0$. We use Taylor's theorem to order $j$ to get
	\begin{equation}
		\mathscr{F}(\tau)=\mathscr{F}^{(j)}(0)\frac{\tau^j}{j!}+\int_{0}^{\tau}\mathrm{d}\tau_2\ \mathscr{F}^{(j+1)}(\tau-\tau_2)\frac{\tau_2^j}{j!},
	\end{equation}
	which implies
	\begin{equation}
		\norm{\mathscr{F}(\tau)}\geq\norm{\mathscr{F}^{(j)}(0)}\frac{\abs{\tau}^j}{j!}-\max_{\abs{\tau_2}\leq\abs{\tau}}\norm{\mathscr{F}^{(j+1)}(\tau_2)}\frac{\abs{\tau}^{j+1}}{(j+1)!}
	\end{equation}
	by the triangle inequality. We combine the above inequalities and divide both sides by $\abs{\tau}^j$. Taking the limit $\tau\rightarrow 0$ gives the contradiction $\norm{\mathscr{F}^{(j)}(0)}\leq 0$.
\end{proof}

\lem{order_cond_deriv} provides a direct approach to computing order conditions for functions of real variables. This works for simple examples such as the power functions $f(\tau)=\tau^p=\OO{\tau^p}$. Another example that we will use in our analysis is the integration of a monomial, such as
\begin{equation}
\int_{0}^{\tau}\mathrm{d}\tau_1\int_{0}^{\tau_1}\mathrm{d}\tau_2\int_{0}^{\tau_1}\mathrm{d}\tau_3\int_{0}^{\tau_2}\mathrm{d}\tau_4\ \tau_1^3\tau_2\tau_3^4\tau_4^5.
\end{equation}
As the following lemma shows, we can directly evaluate such an integral and compute the order condition of the resulting power function.
\begin{lemma}[Integration of a monomial]
	\label{lem:monomial}
	The integration of a monomial $\tau_1^{p_1}\cdots\tau_\gamma^{p_\gamma}\cdots\tau_\Gamma^{p_\Gamma}$ evaluates as
	\begin{equation}
	\int_{0}^{\tau}\mathrm{d}\tau_1\cdots\int_{0}^{\tau_{<\gamma}}\mathrm{d}\tau_\gamma\cdots\int_{0}^{\tau_{<\Gamma}}\mathrm{d}\tau_\Gamma\
	\tau_1^{p_1}\cdots\tau_\gamma^{p_\gamma}\cdots\tau_\Gamma^{p_\Gamma}
	=c t^{p_1+\cdots+p_\Gamma+\Gamma}=\OO{t^{p_1+\cdots+p_\Gamma+\Gamma}},
	\end{equation}
	where $\tau_{<\gamma}\in\{\tau,\tau_1,\ldots,\tau_{\gamma-1}\}$ and $c$ is a constant that depends on nonnegative integers $p_1,\ldots,p_\Gamma$.
\end{lemma}
\begin{proof}
	We induct on the value of $\Gamma$. The claim trivially holds when $\Gamma=1$. Suppose that it is true for $\Gamma$. For $\Gamma+1$, we have
	\begin{equation}
	\int_{0}^{\tau}\mathrm{d}\tau_1\cdots\int_{0}^{\tau_{<\Gamma+1}}\mathrm{d}\tau_{\Gamma+1}\
	\tau_1^{p_1}\cdots\tau_{\Gamma+1}^{p_{\Gamma+1}}
	=\int_{0}^{\tau}\mathrm{d}\tau_1\cdots\int_{0}^{\tau_{<\Gamma}}\mathrm{d}\tau_\Gamma\
	\frac{\tau_1^{q_1}\cdots\tau_\Gamma^{q_\Gamma}}{p_{\Gamma+1}+1},
	\end{equation}
	where $q_1+\cdots+q_\Gamma=p_1+\cdots+p_{\Gamma+1}+1$. The claim then follows from the inductive hypothesis.
\end{proof}

For most of our analysis, however, a direct calculation of order conditions is inefficient. In particular, a ($2k$)th-order Suzuki formula contains $2\cdot 5^{k-1}$ matrix exponentials and a direct analysis becomes prohibitive when $k$ is large. Instead, we follow standard rules of order conditions to compute them indirectly, some of which are summarized below:
\begin{proposition}[Rules of order conditions]
	\label{prop:order_cond_rule}
	Let $\mathscr{F}(\tau)$ and $\mathscr{G}(\tau)$ be operator-valued functions defined on $\R$ that are infinitely differentiable. Let $p$ and $q$ be nonnegative integers. The following rules of order conditions hold:
	\begin{enumerate}
		\item Addition: if $\mathscr{F}(\tau)=O(\tau^p)$ and $\mathscr{G}(\tau)=O(\tau^q)$, then $\mathscr{F}(\tau)+\mathscr{G}(\tau)=O(\tau^{\min(p,q)})$;
		\item Multiplication: if $\mathscr{F}(\tau)=O(\tau^p)$ and $\mathscr{G}(\tau)=O(\tau^q)$, then $\mathscr{F}(\tau)\mathscr{G}(\tau)=O(\tau^{p+q})$;
		\item Differentiation: $\mathscr{F}(\tau)=O(\tau^{p+1})$ if and only if $\mathscr{F}(0)=0$ and $\mathscr{F}'(\tau)=O(\tau^p)$;
		\item Integration: $\mathscr{F}(\tau)=O(\tau^p)$ if and only if $\int_{0}^{t}\mathrm{d}\tau\mathscr{F}(\tau)=O(t^{p+1})$; and
		\item Exponentiation: $\mathscr{F}(\tau)=\mathscr{G}(\tau)+O(\tau^p)$ if and only if $\expT\big(\int_{0}^{t}\mathrm{d}\tau\mathscr{F}(\tau)\big)=\expT\big(\int_{0}^{t}\mathrm{d}\tau \mathscr{G}(\tau)\big)+O(t^{p+1})$.
	\end{enumerate}
\end{proposition}
\begin{proof}
	We only prove the exponentiation rule, as the other rules follow directly from \lem{order_cond_deriv}. Suppose that $\expT\big(\int_{0}^{t}\mathrm{d}\tau\ \mathscr{F}(\tau)\big)=\expT\big(\int_{0}^{t}\mathrm{d}\tau\ \mathscr{G}(\tau)\big)+O(t^{p+1})$. To prove $\mathscr{F}(\tau)=\mathscr{G}(\tau)+O(\tau^p)$, it suffices to show that $\mathscr{F}^{(q)}(0)=\mathscr{G}^{(q)}(0)$ for $q=0,\ldots,p-1$.
	
	We prove this by induction. By the differentiation rule, we have
	\begin{equation}
	\mathscr{F}(t)\expT\bigg(\int_{0}^{t}\mathrm{d}\tau\ \mathscr{F}(\tau)\bigg)=\mathscr{G}(t)\expT\bigg(\int_{0}^{t}\mathrm{d}\tau\ \mathscr{G}(\tau)\bigg)+O(t^{p}),
	\end{equation}
	so \lem{order_cond_deriv} implies $\mathscr{F}(0)=\mathscr{G}(0)$. This proves the claim in the base case. Now assume that $\mathscr{F}^{(l)}(0)=\mathscr{G}^{(l)}(0)$ holds for $l=0,\ldots q$, where $q<p-1$. By \lem{order_cond_deriv} and the general Leibniz rule,
	\begin{equation}
	\sum_{l=0}^{q+1}\binom{q+1}{l}\mathscr{F}^{q+1-l}(0)\expT^{(l)}\bigg(\int_{0}^{0}\mathrm{d}\tau\ \mathscr{F}(\tau)\bigg)
	=\sum_{l=0}^{q+1}\binom{q+1}{l}\mathscr{G}^{q+1-l}(0)\expT^{(l)}\bigg(\int_{0}^{0}\mathrm{d}\tau\ \mathscr{G}(\tau)\bigg).
	\end{equation}
	\lem{order_cond_deriv} also implies $\expT^{(l)}\big(\int_{0}^{0}\mathrm{d}\tau\ \mathscr{F}(\tau)\big)=\expT^{(l)}\big(\int_{0}^{0}\mathrm{d}\tau\ \mathscr{G}(\tau)\big)$ for $l=0,\ldots,q+1$. So the above equation simplifies to
	\begin{equation}
	\mathscr{F}^{(q+1)}(0)
	=\mathscr{G}^{(q+1)}(0).
	\end{equation}
	This completes the inductive step.
	
	For the reverse direction, we want to prove $\expT\big(\int_{0}^{t}\mathrm{d}\tau\ \mathscr{F}(\tau)\big)=\expT\big(\int_{0}^{t}\mathrm{d}\tau\ \mathscr{G}(\tau)\big)+O(t^{p+1})$ assuming that $\mathscr{F}(\tau)=\mathscr{G}(\tau)+O(\tau^p)$. Equivalently, we want to show that $\expT^{(q+1)}\big(\int_{0}^{0}\mathrm{d}\tau\ \mathscr{F}(\tau)\big)=\expT^{(q+1)}\big(\int_{0}^{0}\mathrm{d}\tau\ \mathscr{G}(\tau)\big)$ for $q=0,\ldots,p-1$ given that $\mathscr{F}^{(q)}(0)=\mathscr{G}^{(q)}(0)$. This can be proved by induction and by applying the Leibniz rule in a similar way as above. Specifically, the base case follows from
	\begin{equation}
	\expT^{(1)}\bigg(\int_{0}^{0}\mathrm{d}\tau\ \mathscr{F}(\tau)\bigg)
	=\mathscr{F}(0)=\mathscr{G}(0)
	=\expT^{(1)}\bigg(\int_{0}^{0}\mathrm{d}\tau\ \mathscr{G}(\tau)\bigg)
	\end{equation}
	and the inductive step follows from
	\begin{equation}
	\begin{aligned}
	\expT^{(q+1)}\bigg(\int_{0}^{0}\mathrm{d}\tau\ \mathscr{F}(\tau)\bigg)
	&=\sum_{l=0}^{q}\binom{q}{l}\mathscr{F}^{(q-l)}(0)\expT^{(l)}\bigg(\int_{0}^{0}\mathrm{d}\tau\ \mathscr{F}(\tau)\bigg)\\
	&=\sum_{l=0}^{q}\binom{q}{l}\mathscr{G}^{(q-l)}(0)\expT^{(l)}\bigg(\int_{0}^{0}\mathrm{d}\tau\ \mathscr{G}(\tau)\bigg)\\
	&=\expT^{(q+1)}\bigg(\int_{0}^{0}\mathrm{d}\tau\ \mathscr{G}(\tau)\bigg).
	\end{aligned}
	\end{equation}
\end{proof}

We now compute order conditions for the additive, multiplicative, and exponentiated Trotter error.
\thmorder*
\begin{proof}
	Suppose that $\mathscr{T}(\tau)=O(\tau^p)$. We apply the multiplication rule of \prop{order_cond_rule} to get $e^{(t-\tau)H}\mathscr{S}(\tau)\mathscr{T}(\tau)=O(\tau^p)$. A further application of the integration rule gives $\mathscr{S}(t)-e^{tH}=\int_{0}^{t}\mathrm{d}\tau\ e^{(t-\tau)H}\mathscr{S}(\tau)\mathscr{T}(\tau)=O(t^{p+1})$.
	
	Conversely, let $\mathscr{S}(t)=e^{tH}+O(t^{p+1})$. This implies $\int_{0}^{t}\mathrm{d}\tau\ e^{(t-\tau)H}\mathscr{S}(\tau)\mathscr{T}(\tau)=O(t^{p+1})$. Applying the integration rule and the multiplication rule gives $\mathscr{S}(\tau)\mathscr{T}(\tau)=O(\tau^p)$. Note that $\mathscr{S}(t)=e^{tH}+O(t^{p+1})=I+O(t)$ implies that the operator-valued function $\mathscr{S}(t)$ is invertible for sufficiently small $t$ and, since $\frac{\mathrm{d}}{\mathrm{d}t}\mathscr{S}^{-1}(t)=-\mathscr{S}^{-1}(t)\mathscr{S}'(t)\mathscr{S}^{-1}(t)$, the inverse function $\mathscr{S}^{-1}(t)$ is infinitely differentiable. Applying the multiplication rule gives $\mathscr{T}(\tau)=O(\tau^p)$, which establishes the equivalence of Conditions $1$ and $2$.
	
	Note that $\mathscr{S}(t)=e^{tH}+O(t^{p+1})$ is equivalent to $\expT\big(\int_{0}^{t}\mathrm{d}\tau (H+\mathscr{E}(\tau))\big)=e^{tH}+O(t^{p+1})$, which is further equivalent to $H+\mathscr{E}(\tau)=H+O(\tau^p)$ by the exponentiation rule. Canceling $H$ from both sides proves the equivalence of Conditions $1$ and $3$.
	
	Finally, note that $\mathscr{S}(t)=e^{tH}(I+\mathscr{M}(t))=e^{tH}+O(t^{p+1})$ can be simplified to $e^{tH}\mathscr{M}(t)=O(t^{p+1})$. The equivalence of Conditions $1$ and $4$ then follows from the multiplication rule.
\end{proof}

\section{Error representations}
\label{append:rep}

We now prove \thm{trotter_error_comm_scaling} that establishes the commutator scaling of Trotter error. The proof is sketched in \sec{theory_rep} and will be detailed here. For simplicity, we only discuss the additive error, although the analysis can be easily adapted to handle the multiplicative error and the exponentiated error.

Let $H=\sum_{\gamma=1}^{\Gamma}H_\gamma$ be an operator that generates the evolution $e^{tH}=e^{t\sum_{\gamma=1}^{\Gamma}H_\gamma}$. Let $\mathscr{S}(t)=\prod_{\upsilon=1}^{\Upsilon}\prod_{\gamma=1}^{\Gamma}e^{ta_{(\upsilon,\gamma)}H_{\pi_{\upsilon}(\gamma)}}$ be a $p$th-order product formula as in \sec{prelim_pf}. We know from \thm{error_type} that the Trotter error can be expressed in an additive form as $\mathscr{S}(t)=e^{tH}+\int_{0}^{t}\mathrm{d}\tau\ e^{(t-\tau)H}\mathscr{S}(\tau)\mathscr{T}(\tau)$, where
\begin{equation}
\begin{aligned}
\mathscr{T}(\tau)
=&\sum_{(\upsilon,\gamma)}\prod_{(\upsilon',\gamma')\prec(\upsilon,\gamma)}^{\longrightarrow}e^{-\tau a_{(\upsilon',\gamma')}H_{\pi_{\upsilon'}(\gamma')}}\big(a_{(\upsilon,\gamma)}H_{\pi_{\upsilon}(\gamma)}\big)\prod_{(\upsilon',\gamma')\prec(\upsilon,\gamma)}^{\longleftarrow}e^{\tau a_{(\upsilon',\gamma')}H_{\pi_{\upsilon'}(\gamma')}}\\
&-\prod_{(\upsilon',\gamma')}^{\longrightarrow}e^{-\tau a_{(\upsilon',\gamma')}H_{\pi_{\upsilon'}(\gamma')}}H\prod_{(\upsilon',\gamma')}^{\longleftarrow}e^{\tau a_{(\upsilon',\gamma')}H_{\pi_{\upsilon'}(\gamma')}}.
\end{aligned}
\end{equation}
Furthermore, \thm{error_order_cond} implies that the operator-valued function $\mathscr{T}(\tau)$ satisfies the order condition $\mathscr{T}(\tau)=O(\tau^p)$.

We now apply \thm{comm_exp_conj} to expand every conjugation of matrix exponentials in $\mathscr{T}(\tau)$. In doing so, we only keep track of terms of order $O(\tau^p)$, as those terms corresponding to $1,\tau,\ldots,\tau^{p-1}$ will vanish due to the order condition. We obtain
\begin{equation}
\begin{aligned}
\norm{\mathscr{T}(\tau)}
&\leq\sum_{(\upsilon,\gamma)}\acomm\bigg(\overrightarrow{\big\{H_{\pi_{\upsilon'}(\gamma')},(\upsilon',\gamma')\prec(\upsilon,\gamma)\big\}},H_{\pi_{\upsilon}(\gamma)}\bigg)\frac{\tau^p}{p!}
\exp\bigg(2\tau \sum_{(\upsilon',\gamma')\prec(\upsilon,\gamma)}\norm{H_{\pi_{\upsilon'}(\gamma')}}\bigg)\\
&\quad +\acomm\bigg(\overrightarrow{\big\{H_{\pi_{\upsilon'}(\gamma')}\big\}},H\bigg)\frac{\tau^p}{p!}
\exp\bigg(2\tau \sum_{(\upsilon',\gamma')}\norm{H_{\pi_{\upsilon'}(\gamma')}}\bigg),
\end{aligned}
\end{equation}
where $\overrightarrow{\{\}}$ denotes an ordered list where elements have increasing indices from left to right. This is further bounded by
\begin{equation}
\begin{aligned}
\norm{\mathscr{T}(\tau)}
&\leq2\sum_{(\upsilon,\gamma)}\acomm\bigg(\overrightarrow{\big\{H_{\pi_{\upsilon'}(\gamma')}\big\}},H_{\pi_{\upsilon}(\gamma)}\bigg)\frac{\tau^p}{p!}
\exp\bigg(2\tau \sum_{(\upsilon',\gamma')}\norm{H_{\pi_{\upsilon'}(\gamma')}}\bigg)\\
&=2\Upsilon\sum_{\gamma=1}^{\Gamma}\acomm\bigg(\overrightarrow{\big\{H_{\pi_{\upsilon'}(\gamma')}\big\}},H_{\gamma}\bigg)\frac{\tau^p}{p!}
\exp\bigg(2\tau \Upsilon\sum_{\gamma'=1}^{\Gamma}\norm{H_{\gamma'}}\bigg).
\end{aligned}
\end{equation}
After a final integration over $\tau$, we have
\begin{equation}
\begin{aligned}
\norm{\mathscr{S}(t)-e^{tH}}
&\leq\int_{0}^{t}\mathrm{d}\tau\norm{e^{(t-\tau)H}\mathscr{S}(\tau)\mathscr{T}(\tau)}\\
&\leq2\Upsilon\sum_{\gamma=1}^{\Gamma}\acomm\bigg(\overrightarrow{\big\{H_{\pi_{\upsilon'}(\gamma')}\big\}},H_{\gamma}\bigg)\frac{t^{p+1}}{(p+1)!}
\exp\bigg(4t \Upsilon\sum_{\gamma'=1}^{\Gamma}\norm{H_{\gamma'}}\bigg)
\end{aligned}
\end{equation}
with prefactor $4$ in the exponent. This bound holds for arbitrary operators $H_\gamma$. If the operator summands are anti-Hermitian, the bound can be further tightened to
\begin{equation}
\norm{\mathscr{S}(t)-e^{tH}}
\leq2\Upsilon\sum_{\gamma=1}^{\Gamma}\acomm\bigg(\overrightarrow{\big\{H_{\pi_{\upsilon'}(\gamma')}\big\}},H_{\gamma}\bigg)\frac{t^{p+1}}{(p+1)!}.
\end{equation}

In the following, we show that the prefactor in the exponent can be improved from $4$ to $2$ by a more careful analysis. This can be achieved using the expansion in \thm{comm_exp_conj} without invoking the triangle inequality. After canceling low-order terms, we have
\begingroup
\allowdisplaybreaks
\begin{align*}
	\mathscr{T}(\tau)&=
	\sum_{(\upsilon,\gamma)}\sum_{(\upsilon',\gamma')\prec(\upsilon,\gamma)}\prod_{(\upsilon'',\gamma'')\prec(\upsilon',\gamma')}^{\longrightarrow}e^{-\tau a_{(\upsilon'',\gamma'')}H_{(\upsilon'',\gamma'')}}\sum_{\substack{q_{(\upsilon',\gamma')}+\cdots+q_{(\upsilon,\gamma)-1}=p\\q_{(\upsilon',\gamma')}\neq 0}}
	\int_{0}^{\tau}\mathrm{d}\tau_2\ e^{-\tau_2 a_{(\upsilon',\gamma')}H_{(\upsilon',\gamma')}}\\
	&\qquad\cdot\ad_{-a_{(\upsilon',\gamma')}H_{(\upsilon',\gamma')}}^{q_{(\upsilon',\gamma')}}
	\cdots\ad_{-a_{(\upsilon,\gamma)-1}H_{(\upsilon,\gamma)-1}}^{q_{(\upsilon,\gamma)-1}}\big(a_{(\upsilon,\gamma)}H_{(\upsilon,\gamma)}\big)
	\frac{(\tau-\tau_2)^{q_{(\upsilon',\gamma')}-1}\tau^{q_{(\upsilon',\gamma')+1}+\cdots+q_{(\upsilon,\gamma)-1}}}{(q_{(\upsilon',\gamma')}-1)!q_{(\upsilon',\gamma')+1}!\cdots q_{(\upsilon,\gamma)-1}!}\\
	&\qquad\cdot e^{\tau_2 a_{(\upsilon',\gamma')}H_{(\upsilon',\gamma')}}
	\prod_{(\upsilon'',\gamma'')\prec(\upsilon',\gamma')}^{\longleftarrow}e^{\tau a_{(\upsilon'',\gamma'')}H_{(\upsilon'',\gamma'')}}\\
	&\quad-\sum_{(\upsilon',\gamma')}\prod_{(\upsilon'',\gamma'')\prec(\upsilon',\gamma')}^{\longrightarrow}e^{-\tau a_{(\upsilon'',\gamma'')}H_{(\upsilon'',\gamma'')}}\sum_{\substack{q_{(\upsilon',\gamma')}+\cdots+q_{(\Upsilon,\Gamma)}=p\\q_{(\upsilon',\gamma')}\neq 0}}
	\int_{0}^{\tau}\mathrm{d}\tau_2\ e^{-\tau_2 a_{(\upsilon',\gamma')}H_{(\upsilon',\gamma')}}\\
	&\qquad\cdot\ad_{-a_{(\upsilon',\gamma')}H_{(\upsilon',\gamma')}}^{q_{(\upsilon',\gamma')}}
	\cdots\ad_{-a_{(\Upsilon,\Gamma)}H_{(\Upsilon,\Gamma)}}^{q_{(\Upsilon,\Gamma)}}\big(H\big)
	\frac{(\tau-\tau_2)^{q_{(\upsilon',\gamma')}-1}\tau^{q_{(\upsilon',\gamma')+1}+\cdots+q_{(\Upsilon,\Gamma)}}}{(q_{(\upsilon',\gamma')}-1)!q_{(\upsilon',\gamma')+1}!\cdots q_{(\Upsilon,\Gamma)}!}\\
	&\qquad\cdot e^{\tau_2 a_{(\upsilon',\gamma')}H_{(\upsilon',\gamma')}}
	\prod_{(\upsilon'',\gamma'')\prec(\upsilon',\gamma')}^{\longleftarrow}e^{\tau a_{(\upsilon'',\gamma'')}H_{(\upsilon'',\gamma'')}},
\end{align*}
\endgroup
where we have temporarily defined $e^{\tau a_{(\upsilon,\gamma)}H_{(\upsilon,\gamma)}}:=e^{\tau a_{(\upsilon,\gamma)}H_{\pi_{\upsilon}(\gamma)}}$. This implies that Trotter error can be expressed in an additive form as $\mathscr{S}(t)-e^{tH}=\int_{0}^{t}\mathrm{d}\tau\ e^{(t-\tau)H}\mathscr{R}(\tau)$, where
\begingroup
\allowdisplaybreaks
\begin{align*}
\mathscr{R}(\tau)&=
\sum_{(\upsilon,\gamma)}\sum_{(\upsilon',\gamma')\prec(\upsilon,\gamma)}\prod_{(\upsilon'',\gamma'')\succ(\upsilon',\gamma')}^{\longleftarrow}e^{\tau a_{(\upsilon'',\gamma'')}H_{(\upsilon'',\gamma'')}}\sum_{\substack{q_{(\upsilon',\gamma')}+\cdots+q_{(\upsilon,\gamma)-1}=p\\q_{(\upsilon',\gamma')}\neq 0}}
\int_{0}^{\tau}\mathrm{d}\tau_2\ e^{(\tau-\tau_2) a_{(\upsilon',\gamma')}H_{(\upsilon',\gamma')}}\\
&\qquad\cdot\ad_{-a_{(\upsilon',\gamma')}H_{(\upsilon',\gamma')}}^{q_{(\upsilon',\gamma')}}
\cdots\ad_{-a_{(\upsilon,\gamma)-1}H_{(\upsilon,\gamma)-1}}^{q_{(\upsilon,\gamma)-1}}\big(a_{(\upsilon,\gamma)}H_{(\upsilon,\gamma)}\big)
\frac{(\tau-\tau_2)^{q_{(\upsilon',\gamma')}-1}\tau^{q_{(\upsilon',\gamma')+1}+\cdots+q_{(\upsilon,\gamma)-1}}}{(q_{(\upsilon',\gamma')}-1)!q_{(\upsilon',\gamma')+1}!\cdots q_{(\upsilon,\gamma)-1}!}\\
&\qquad\cdot e^{\tau_2 a_{(\upsilon',\gamma')}H_{(\upsilon',\gamma')}}
\prod_{(\upsilon'',\gamma'')\prec(\upsilon',\gamma')}^{\longleftarrow}e^{\tau a_{(\upsilon'',\gamma'')}H_{(\upsilon'',\gamma'')}}\\
&\quad-\sum_{(\upsilon',\gamma')}\prod_{(\upsilon'',\gamma'')\succ(\upsilon',\gamma')}^{\longleftarrow}e^{\tau a_{(\upsilon'',\gamma'')}H_{(\upsilon'',\gamma'')}}\sum_{\substack{q_{(\upsilon',\gamma')}+\cdots+q_{(\Upsilon,\Gamma)}=p\\q_{(\upsilon',\gamma')}\neq 0}}
\int_{0}^{\tau}\mathrm{d}\tau_2\ e^{(\tau-\tau_2) a_{(\upsilon',\gamma')}H_{(\upsilon',\gamma')}}\\
&\qquad\cdot\ad_{-a_{(\upsilon',\gamma')}H_{(\upsilon',\gamma')}}^{q_{(\upsilon',\gamma')}}
\cdots\ad_{-a_{(\Upsilon,\Gamma)}H_{(\Upsilon,\Gamma)}}^{q_{(\Upsilon,\Gamma)}}\big(H\big)
\frac{(\tau-\tau_2)^{q_{(\upsilon',\gamma')}-1}\tau^{q_{(\upsilon',\gamma')+1}+\cdots+q_{(\Upsilon,\Gamma)}}}{(q_{(\upsilon',\gamma')}-1)!q_{(\upsilon',\gamma')+1}!\cdots q_{(\Upsilon,\Gamma)}!}\\
&\qquad\cdot e^{\tau_2 a_{(\upsilon',\gamma')}H_{(\upsilon',\gamma')}}
\prod_{(\upsilon'',\gamma'')\prec(\upsilon',\gamma')}^{\longleftarrow}e^{\tau a_{(\upsilon'',\gamma'')}H_{(\upsilon'',\gamma'')}}.
\end{align*}
\endgroup
By the assumption of \thm{trotter_error_comm_scaling}, we have $t\geq\tau\geq\tau_2\geq0$, which implies
\begingroup
\allowdisplaybreaks
\begin{align*}
\norm{\mathscr{R}(\tau)}&\leq
\sum_{(\upsilon,\gamma)}\sum_{(\upsilon',\gamma')\prec(\upsilon,\gamma)}\sum_{\substack{q_{(\upsilon',\gamma')}+\cdots+q_{(\upsilon,\gamma)-1}=p\\q_{(\upsilon',\gamma')}\neq 0}}
\int_{0}^{\tau}\mathrm{d}\tau_2\ \frac{(\tau-\tau_2)^{q_{(\upsilon',\gamma')}-1}\tau^{q_{(\upsilon',\gamma')+1}+\cdots+q_{(\upsilon,\gamma)-1}}}{(q_{(\upsilon',\gamma')}-1)!q_{(\upsilon',\gamma')+1}!\cdots q_{(\upsilon,\gamma)-1}!}\\
&\qquad\cdot\norm{\ad_{-a_{(\upsilon',\gamma')}H_{(\upsilon',\gamma')}}^{q_{(\upsilon',\gamma')}}
\cdots\ad_{-a_{(\upsilon,\gamma)-1}H_{(\upsilon,\gamma)-1}}^{q_{(\upsilon,\gamma)-1}}\big(a_{(\upsilon,\gamma)}H_{(\upsilon,\gamma)}\big)}
\exp\bigg(\tau \sum_{(\upsilon'',\gamma'')}\norm{H_{\pi_{\upsilon''}(\gamma'')}}\bigg)\\
&\quad+\sum_{(\upsilon',\gamma')}\sum_{\substack{q_{(\upsilon',\gamma')}+\cdots+q_{(\Upsilon,\Gamma)}=p\\q_{(\upsilon',\gamma')}\neq 0}}
\int_{0}^{\tau}\mathrm{d}\tau_2\ \frac{(\tau-\tau_2)^{q_{(\upsilon',\gamma')}-1}\tau^{q_{(\upsilon',\gamma')+1}+\cdots+q_{(\Upsilon,\Gamma)}}}{(q_{(\upsilon',\gamma')}-1)!q_{(\upsilon',\gamma')+1}!\cdots q_{(\Upsilon,\Gamma)}!}\\
&\qquad\cdot\norm{\ad_{-a_{(\upsilon',\gamma')}H_{(\upsilon',\gamma')}}^{q_{(\upsilon',\gamma')}}
	\cdots\ad_{-a_{(\Upsilon,\Gamma)}H_{(\Upsilon,\Gamma)}}^{q_{(\Upsilon,\Gamma)}}\big(H\big)}
\exp\bigg(\tau \sum_{(\upsilon'',\gamma'')}\norm{H_{\pi_{\upsilon''}(\gamma'')}}\bigg)\\
&\leq\sum_{(\upsilon,\gamma)}\acomm\bigg(\overrightarrow{\big\{H_{\pi_{\upsilon'}(\gamma')},(\upsilon',\gamma')\prec(\upsilon,\gamma)\big\}},H_{\pi_{\upsilon}(\gamma)}\bigg)\frac{\tau^p}{p!}
\exp\bigg(\tau \sum_{(\upsilon'',\gamma'')}\norm{H_{\pi_{\upsilon''}(\gamma'')}}\bigg)\\
&\quad +\acomm\bigg(\overrightarrow{\big\{H_{\pi_{\upsilon'}(\gamma')}\big\}},H\bigg)\frac{\tau^p}{p!}
\exp\bigg(\tau \sum_{(\upsilon'',\gamma'')}\norm{H_{\pi_{\upsilon''}(\gamma'')}}\bigg).
\end{align*}
\endgroup
The remaining analysis proceeds as above.

For the multiplicative error, we have
\begingroup
\allowdisplaybreaks
\begin{align*}
	\mathscr{M}(t)&=\expT\bigg(\int_{0}^{t}\mathrm{d}\tau\ e^{-\tau H}\mathscr{E}(\tau)e^{\tau H}\bigg)-I\\
	&=\int_{0}^{t}\mathrm{d}\tau\ e^{-\tau H}\mathscr{E}(\tau)e^{\tau H}
	\expT\bigg(\int_{0}^{\tau}\mathrm{d}\tau_2\ e^{-\tau_2 H}\mathscr{E}(\tau_2)e^{\tau_2 H}\bigg)\\
	&=\int_{0}^{t}\mathrm{d}\tau\  e^{-\tau H}\mathscr{E}(\tau)\mathscr{S}(\tau),
\end{align*}
\endgroup
where the first equality follows from \thm{error_type}, the second equality follows from the integral equation \eq{int_eq}, and the third equality follows from the definition of multiplicative error. Using the explicit expression of $\mathscr{E}(\tau)$ in \thm{error_type}, we obtain a bound on the multiplicative error similar to the additive bound.

Note that our analysis depends on $\pi_{\upsilon'}$, the ordering of operator summands in stage $\upsilon'$ of the product formula. In the following, we prove an asymptotic bound that removes this ordering constraint. The resulting bound is independent of the definition of product formula and may thus be easier to compute in practice. Our analysis here is not tight in terms of the constant prefactor, but it is sufficient to establish the commutator scaling in \thm{trotter_error_comm_scaling}.

Recall from \thm{comm_exp_conj} that
\begin{equation}
\acomm\bigg(\overrightarrow{\big\{H_{\pi_{\upsilon'}(\gamma')}\big\}},H_{\gamma}\bigg)
=\sum_{q_{(1,1)}+\cdots+q_{(\Upsilon,\Gamma)}=p}\binom{p}{q_{(1,1)}\ \cdots\ q_{(\Upsilon,\Gamma)}}\norm{\ad_{H_{\pi_{1}(1)}}^{q_{(1,1)}}\cdots\ad_{H_{\pi_{\Upsilon}(\Gamma)}}^{q_{(\Upsilon,\Gamma)}}(H_\gamma)},
\end{equation}
which is upper bounded by $p!$ times $\sum_{q_{(1,1)}+\cdots+q_{(\Upsilon,\Gamma)}=p}\norm{\ad_{H_{\pi_{1}(1)}}^{q_{(1,1)}}\cdots\ad_{H_{\pi_{\Upsilon}(\Gamma)}}^{q_{(\Upsilon,\Gamma)}}(H_\gamma)}$. Fixing the value of $\gamma$, we claim that
\begin{equation}
\sum_{q_{(1,1)}+\cdots+q_{(\Upsilon,\Gamma)}=p}\norm{\ad_{H_{\pi_{1}(1)}}^{q_{(1,1)}}\cdots\ad_{H_{\pi_{\Upsilon}(\Gamma)}}^{q_{(\Upsilon,\Gamma)}}(H_\gamma)}
\leq \Upsilon^p\sum_{\gamma_{p+1}=1}^{\Gamma}\cdots\sum_{\gamma_2=1}^{\Gamma}\norm{\big[H_{\gamma_{p+1}},\cdots\big[H_{\gamma_2},H_\gamma\big]\big]}.
\end{equation}
This can be seen as follows. Every nested commutator on the left-hand side has $p$ nesting layers and must thus be of the form on the right. Conversely, we fix one term $\norm{\big[H_{\gamma_{p+1}},\cdots\big[H_{\gamma_2},H_\gamma\big]\big]}$ from the right and bound the number of times this term might appear on the left. Each operator $H_{\gamma_2},\ldots,H_{\gamma_{p+1}}$ can appear in $\Upsilon$ possible stages and hence there are $\Upsilon^p$ possibilities in total. When the stages are fixed, this will uniquely determine one term $\norm{\ad_{H_{\pi_{1}(1)}}^{q_{(1,1)}}\cdots\ad_{H_{\pi_{\Upsilon}(\Gamma)}}^{q_{(\Upsilon,\Gamma)}}(H_\gamma)}$ on the left. We have thus established the commutator scaling of Trotter error.

\thmcomm*

\section{Simulating second-quantized electronic structure}
\label{append:electron}

In this section, we use product formulas to simulate the second-quantized plane-wave electronic structure
\begin{equation}
\label{eq:plane_wave_dual2}
\begin{aligned}
H&=\underbrace{\frac{1}{2n}\sum_{j,k,\nu}\kappa_{\nu}^2\cos[\kappa_{\nu}\cdot r_{k-j}]A_{j}^\dagger A_{k}}_{T}\\
&\quad\underbrace{-\frac{4\pi}{\omega}\sum_{j,\iota,\nu\neq 0}\frac{\zeta_\iota\cos[\kappa_{\nu}\cdot(\widetilde{r}_\iota-r_j)]}{\kappa_{\nu}^2}N_{j}}_{U}
\underbrace{+\frac{2\pi}{\omega}\sum_{\substack{j\neq k\\\nu\neq 0}}\frac{\cos[\kappa_{\nu}\cdot r_{j-k}]}{\kappa_{\nu}^2}N_{j}N_{k}}_{V},
\end{aligned}
\end{equation}
where $j,k$ range over all $n$ orbitals, $\omega$ is the volume of the computational cell, and we consider the constant density case where $n/\omega=\cO{1}$. Here, $\kappa_{\nu}=2\pi\nu/\omega^{1/3}$ are $n$ vectors of the plane-wave frequencies, where $\nu$ are three-dimensional vectors of integers with elements in $[-n^{1/3},n^{1/3}]$, $r_j$ are the positions of electrons, $\zeta_\iota$ are nuclei charges such that $\sum_{\iota}|\zeta_\iota|=\cO{n}$, and $\widetilde{r}_\iota$ are the nuclear coordinates. $A_j^\dagger$ and $A_k$ are the creation and annihilation operators, and $N_{j}=A_{j}^\dagger A_{j}$ are the number operators.

Following the analysis in \sec{app_dqs}, we need to bound the spectral norm of the nested commutators $[H_{\gamma_{p+1}},\cdots[H_{\gamma_2},H_{\gamma_1}]]$, where $H_\gamma\in\{T,U,V\}$. This can be done by induction. In the base case, we need to estimate the norm of the kinetic operator $T$ and the potential operators $U$ and $V$. For readability, we use the abbreviated representation
\begin{equation}
T=\sum_{j,k}t_{j,k}A_j^\dagger A_k,\qquad
U=\sum_{j}u_j N_j,\qquad
V=\sum_{j,k}v_{j,k} N_j N_k.
\end{equation}
Since $\norm{A_j^\dagger}=\norm{A_j}=\norm{N_j}=1$, we can apply the triangle inequality and upper bound $\norm{T}$, $\norm{U}$, and $\norm{V}$ by the vector $1$-norm $\norm{\vec{t}}_1$, $\norm{\vec{u}}_1$, and $\norm{\vec{v}}_1$. We analyze this in \prop{tuv_base}.
\begin{lemma}[{\cite[(F6) and (F13)]{BWMMNC18}}]
	\label{lem:sum_scaling}
	Let an electronic-structure Hamiltonian be given as in \eq{plane_wave_dual}. The following asymptotic analyses hold:
	\begin{enumerate}
		\item 
		\begin{equation}
		\sum_{\nu\neq 0}\frac{1}{\kappa_{\nu}^2}=\cO{n}.
		\end{equation}
		\item For any fixed $j$,
		\begin{equation}
		\sum_{\nu}\kappa_{\nu}^2\cos[\kappa_{\nu}\cdot r_j]=\cO{1}.
		\end{equation}
		\item 
		\begin{equation}
		\sum_{\iota}\abs{\zeta_\iota}=\cO{n}.
		\end{equation}
	\end{enumerate}
\end{lemma}
\begin{proposition}
	\label{prop:tuv_base}
	Let an electronic-structure Hamiltonian be given as in \eq{plane_wave_dual}. We have the following bounds on the vector $1$-norm and $\infty$-norm of the coefficients of the kinetic operator and the potential operators:
	\begin{equation}
	\begin{aligned}
	\norm{\vec{t}}_\infty&=\cO{\frac{1}{n}},\qquad &\norm{\vec{t}}_1&=\cO{n},\\
	\norm{\vec{u}}_\infty&=\cO{n},\qquad &\norm{\vec{u}}_1&=\cO{n^2},\\
	\norm{\vec{v}}_\infty&=\cO{1},\qquad &\norm{\vec{v}}_1&=\cO{n^2}.
	\end{aligned}
	\end{equation}
\end{proposition}
\begin{proof}
	The claims about the asymptotic scaling of $\norm{\vec{t}}_\infty$, $\norm{\vec{u}}_\infty$, and $\norm{\vec{v}}_\infty$ follow from \lem{sum_scaling}. We then obtain the scaling of the vector $1$-norm from the triangle inequality.
\end{proof}

For the inductive step, we consider a general second-quantized operator of the form
\begin{equation}
W=\sum_{\vec{j},\vec{k},\vec{l}}w_{\vec{j},\vec{k},\vec{l}}\
\underbrace{\cdots\big(A_{j_x}^\dagger A_{k_x}\big)\cdots (N_{l_y})\cdots}_{\text{at most }q\text{ operators}}
\end{equation}
where $\vec{j}$, $\vec{k}$, and $\vec{l}$ denote vectors of orbitals, with total length at most $q$. We keep track of the number of $A_{j_x}^\dagger A_{k_x}$ and $N_{l_y}$ in each summand; the largest such number $q$ is called the ``layer'' of $W$. We compute the commutator between the kinetic/potential operator and a general second-quantized operator in \prop{tuv_induction}.

\begin{lemma}[Commutation rules of second-quantized operators]
	\label{lem:commutation_rule}
	The following commutation rules hold for second-quantized operators:
	\begin{equation}
	\begin{aligned}
	\big[A_j^\dagger A_k,A_l^\dagger A_m\big]
	&=\delta_{kl}A_j^\dagger A_m
	-\delta_{jm}A_l^\dagger A_k,\\
	\big[A_j^\dagger A_k,N_l\big]
	&=\delta_{kl}A_j^\dagger A_l
	-\delta_{jl}A_l^\dagger A_k,\\
	\big[A_j^\dagger A_k, N_{l}N_{m}\big]
	&=\big(\delta_{kl}A_j^\dagger A_{l}-\delta_{jl}A_{l}^\dagger A_k\big)N_{m}
	+N_{l}\big(\delta_{km}A_j^\dagger A_{m}-\delta_{jm}A_{m}^\dagger A_k\big),
	\end{aligned}
	\end{equation}
	where $\delta_{kl}$ is the Kronecker-delta function.
\end{lemma}
\begin{proof}
	The first rule is proved by \cite[(1.8.14)]{helgaker2014molecular}. The other rules follow from the definition of the number operator $N_l=A_l^\dagger A_l$ and the commutation relation $\big[AB,C\big]=A\big[B,C\big]+\big[A,C\big]B$ for any operators $A$, $B$, and $C$.
	
\end{proof}
\begin{proposition}
	\label{prop:tuv_induction}
	Let an electronic-structure Hamiltonian be given as in \eq{plane_wave_dual}. The following statements hold for a general second-quantized operator $W$ with $q$ layers:
	\begin{enumerate}
		\item $\widetilde{W}=\big[T,W\big]$ is an operator with $q$ layers and $\norm{\vec{\widetilde{w}}}_1\leq2qn\norm{\vec{t}}_\infty\norm{\vec{w}}_1$;
		\item $\widetilde{W}=\big[U,W\big]$ is an operator with $q$ layers and $\norm{\vec{\widetilde{w}}}_1\leq2q\norm{\vec{u}}_\infty\norm{\vec{w}}_1$; and
		\item $\widetilde{W}=\big[V,W\big]$ is an operator with $q+1$ layers and $\norm{\vec{\widetilde{w}}}_1\leq 4qn\norm{\vec{v}}_\infty\norm{\vec{w}}_1$.
	\end{enumerate}
\end{proposition}
\begin{proof}
	For Statement $1$, we have
	\begin{equation}
	\begin{aligned}
	\widetilde{W}=\big[T,W\big]
	&=\bigg[\sum_{\alpha,\beta}t_{\alpha,\beta}A_\alpha^\dagger A_\beta,
	\sum_{\vec{j},\vec{k},\vec{l}}w_{\vec{j},\vec{k},\vec{l}}
	\cdots\big(A_{j_x}^\dagger A_{k_x}\big)\cdots (N_{l_y})\cdots\bigg]\\
	&=\sum_{\alpha,\beta}\sum_{\vec{j},\vec{k},\vec{l}}
	t_{\alpha,\beta}w_{\vec{j},\vec{k},\vec{l}}
	\bigg[A_\alpha^\dagger A_\beta,\ \cdots\big(A_{j_x}^\dagger A_{k_x}\big)\cdots (N_{l_y})\cdots\bigg].
	\end{aligned}
	\end{equation}
	Performing the commutation sequentially, it suffices to consider
	\begin{equation}
	\begin{aligned}
	&\cdots\bigg[A_\alpha^\dagger A_\beta,A_{j_x}^\dagger A_{k_x}\bigg]\cdots (N_{l_y})\cdots\\
	&\cdots\big(A_{j_x}^\dagger A_{k_x}\big)\cdots\bigg[A_\alpha^\dagger A_\beta,N_{l_y}\bigg]\cdots\\
	\end{aligned}
	\end{equation}
	For fixed $\alpha$, $\beta$, $\vec{j}$, $\vec{k}$, $\vec{l}$, there are at most $q$ such commutators.
	
	For the first type of commutator, we have from \lem{commutation_rule} that
	\begin{equation}
	\big[A_\alpha^\dagger A_\beta,A_{j_x}^\dagger A_{k_x}\big]
	=\delta_{\beta,j_x}A_\alpha^\dagger A_{k_x}
	-\delta_{\alpha,k_x}A_{j_x}^\dagger A_\beta.
	\end{equation}
	Without loss of generality, consider the first term; its contribution to $\norm{\vec{\widetilde{w}}}_1$ is at most
	\begin{equation}
	\sum_{\alpha,\beta}\sum_{\vec{j},\vec{k},\vec{l}}
	\delta_{\beta,j_x}\abs{t_{\alpha,\beta}w_{\vec{j},\vec{k},\vec{l}}}
	=\sum_{\alpha,\vec{j},\vec{k},\vec{l}}
	\abs{t_{\alpha,j_x}w_{\vec{j},\vec{k},\vec{l}}}
	\leq n\norm{\vec{t}}_\infty\norm{\vec{w}}_1.
	\end{equation}
	Similarly, we use \lem{commutation_rule} to analyze the second type of commutator
	\begin{equation}
	\big[A_\alpha^\dagger A_\beta,N_{l_y}\big]
	=\delta_{l_y,\beta}A_\alpha^\dagger A_\beta
	-\delta_{l_y,\alpha}A_\alpha^\dagger A_\beta
	\end{equation}
	and find its contribution to $\norm{\vec{\widetilde{w}}}_1$ as
	\begin{equation}
	\sum_{\alpha,\beta}\sum_{\vec{j},\vec{k},\vec{l}}
	\delta_{l_y,\beta}\abs{t_{\alpha,\beta}w_{\vec{j},\vec{k},\vec{l}}}
	=\sum_{\alpha,\vec{j},\vec{k},\vec{l}}
	\abs{t_{\alpha,l_y}w_{\vec{j},\vec{k},\vec{l}}}
	\leq n\norm{\vec{t}}_\infty\norm{\vec{w}}_1.
	\end{equation}
	
	For Statement $2$, we have
	\begin{equation}
	\begin{aligned}
	\widetilde{W}=\big[U,W\big]
	&=\bigg[\sum_{\alpha}u_\alpha N_\alpha,
	\sum_{\vec{j},\vec{k},\vec{l}}w_{\vec{j},\vec{k},\vec{l}}
	\cdots\big(A_{j_x}^\dagger A_{j_x}\big)\cdots (N_{l_y})\cdots\bigg]\\
	&=\sum_{\alpha}\sum_{\vec{j},\vec{k},\vec{l}}
	u_\alpha w_{\vec{j},\vec{k},\vec{l}}
	\bigg[N_\alpha,\ \cdots\big(A_{j_x}^\dagger A_{k_x}\big)\cdots (N_{l_y})\cdots\bigg].
	\end{aligned}
	\end{equation}
	Performing the commutation sequentially, it suffices to consider
	\begin{equation}
	\begin{aligned}
	&\cdots\bigg[N_\alpha,A_{j_x}^\dagger A_{k_x}\bigg]\cdots (N_{l_y})\cdots\\
	\end{aligned}
	\end{equation}
	For fixed $\alpha$, $\vec{j}$, $\vec{k}$, $\vec{l}$, there are at most $q$ such commutators. We use \lem{commutation_rule} again to get
	\begin{equation}
	\big[N_\alpha,A_{j_x}^\dagger A_{k_x}\big]
	=\delta_{\alpha,j_x}A_{j_x}^\dagger A_{k_x}
	-\delta_{\alpha,k_x}A_{j_x}^\dagger A_{k_x}
	\end{equation}
	and find its contribution to $\norm{\vec{\widetilde{w}}}_1$ as
	\begin{equation}
	\sum_{\alpha}\sum_{\vec{j},\vec{k},\vec{l}}
	\delta_{\alpha,j_x}\abs{u_\alpha w_{\vec{j},\vec{k},\vec{l}}}
	=\sum_{\vec{j},\vec{k},\vec{l}}
	\abs{u_{j_x}w_{\vec{j},\vec{k},\vec{l}}}
	\leq \norm{\vec{u}}_\infty\norm{\vec{w}}_1.
	\end{equation}
	
	For Statement $3$, we have
	\begin{equation}
	\begin{aligned}
	\widetilde{W}=\big[V,W\big]
	&=\bigg[\sum_{\alpha,\beta}v_{\alpha,\beta} N_\alpha N_\beta,
	\sum_{\vec{j},\vec{k},\vec{l}}w_{\vec{j},\vec{k},\vec{l}}
	\cdots\big(A_{j_x}^\dagger A_{k_x}\big)\cdots (N_{l_y})\cdots\bigg]\\
	&=\sum_{\alpha,\beta}\sum_{\vec{j},\vec{k},\vec{l}}
	v_{\alpha,\beta} w_{\vec{j},\vec{k},\vec{l}}
	\bigg[N_\alpha N_\beta,\ \cdots\big(A_{j_x}^\dagger A_{k_x}\big)\cdots (N_{l_y})\cdots\bigg].
	\end{aligned}
	\end{equation}
	Performing the commutation sequentially, it suffices to consider
	\begin{equation}
	\begin{aligned}
	&\cdots\bigg[N_\alpha N_\beta,A_{j_x}^\dagger A_{k_x}\bigg]\cdots (N_{l_y})\cdots\\
	\end{aligned}
	\end{equation}
	For fixed $\alpha$, $\beta$, $\vec{j}$, $\vec{k}$, $\vec{l}$, there are at most $q$ such commutators. Using \lem{commutation_rule}, we have
	\begin{equation}
	\big[N_{\alpha}N_{\beta},A_{j_x}^\dagger A_{k_x}\big]
	=\big(\delta_{\alpha,j_x}A_{j_x}^\dagger A_{k_x}-\delta_{\alpha,k_x}A_{j_x}^\dagger A_{k_x}\big)N_{\beta}
	+N_{\alpha}\big(\delta_{\beta,j_x}A_{j_x}^\dagger A_{k_x}-\delta_{\beta,k_x}A_{j_x}^\dagger A_{k_x}\big).
	\end{equation}
	Without loss of generality, consider the first term; its contribution to $\norm{\vec{\widetilde{w}}}_1$ is at most
	\begin{equation}
	\sum_{\alpha,\beta}\sum_{\vec{j},\vec{k},\vec{l}}
	\delta_{\alpha,j_x}\abs{v_{\alpha,\beta}w_{\vec{j},\vec{k},\vec{l}}}
	=\sum_{\beta,\vec{j},\vec{k},\vec{l}}
	\abs{v_{j_x,\beta}w_{\vec{j},\vec{k},\vec{l}}}
	\leq n\norm{\vec{v}}_\infty\norm{\vec{w}}_1.
	\end{equation}
\end{proof}

\begin{theorem}[Product-formula simulation of second-quantized plane-wave electronic structure]
	\label{thm:pf-electron}
	Let $H=T+U+V$ be a second-quantized plane-wave electronic-structure Hamiltonian with $n$ orbitals \eq{plane_wave_dual}. Let $\mathscr{S}(t)$ be a $p$th-order product formula as in \eq{plane_wave_dual_pf}. Then, the Trotter error has the scaling
	\begin{equation}
	\norm{\mathscr{S}(t)-e^{-itH}}
	=\cO{(nt)^{p+1}}.
	\end{equation}
	To simulate with accuracy $\epsilon$, it thus suffices to choose a Trotter number of
	\begin{equation}
	r=\cO{\frac{(nt)^{1+1/p}}{\epsilon^{1/p}}}.
	\end{equation}
	Choosing $p$ sufficiently large, letting $\epsilon$ be constant, and implementing each Trotter step as in \cite{Ferris14,LW18}, we have the gate complexity
	\begin{equation}
	n^{2+o(1)}t^{1+o(1)}.
	\end{equation}
\end{theorem}
\begin{proof}
	We compute the scaling of the spectral norm of
	\begin{equation}
	W=\big[H_{\gamma_{p+1}},\cdots\big[H_{\gamma_2},H_{\gamma_1}\big]\big],
	\end{equation}
	by induction, where $H_\gamma\in\{T,U,V\}$. In the base case where $p=1$, we have from \prop{tuv_base} and \prop{tuv_induction} that the coefficients of $W$ have $1$-norm in $\cO{n^2}$, which implies $\norm{W}=\cO{n^2}$. For the inductive step, suppose that $W=\big[H_{\gamma_{p+1}},\cdots\big[H_{\gamma_2},H_{\gamma_1}\big]\big]$ is a second-quantized operator whose coefficients have vector $1$-norm in $\cO{n^p}$. Then \prop{tuv_induction} implies that $\big[T,W\big]$, $\big[U,W\big]$, and $\big[V,W\big]$ are second-quantized operators and their coefficients have $1$-norm in $\cO{n^{p+1}}$. This proves that
	\begin{equation}
	\acommtilde=\sum_{\gamma_1,\gamma_2,\ldots,\gamma_{p+1}}\norm{\big[H_{\gamma_{p+1}},\cdots\big[H_{\gamma_2},H_{\gamma_1}\big]\big]}
	=\cO{n^{p+1}}.
	\end{equation}
	The theorem then follows from \thm{trotter_error_comm_scaling} and \cor{trotter_number_comm_scaling}.
\end{proof}

\section{Simulating \texorpdfstring{$k$}{k}-local Hamiltonians}
\label{append:local}

In this section, we consider simulating $k$-local Hamiltonians using product formulas.

Recall from \sec{app_dqs} that a $k$-local Hamiltonian on $n$ qubits can be expressed as
\begin{equation}
H=\sum_{{j}_1,\ldots,{j}_k}H_{{j}_1,\ldots,{j}_k},
\end{equation}
where each $H_{{j}_1,\ldots,{j}_k}$ acts nontrivially only on qubits labeled by ${j}_1,\ldots,{j}_k$. Our goal is to analyze the nested commutators
\begin{equation}
	\sum_{\gamma_1,\dots,\gamma_{p+1}=1}^{\Gamma}
	\norm{\commm{H_{\gamma_{p+1}},\dots,\commm{H_{\gamma_2},H_{\gamma_1}}}},
\end{equation}
where $\Gamma=n^k$ and $H_{\gamma_j}$ are local operators $H_{{j}_1,\ldots,{j}_k}$. We then bound the Trotter error and the complexity of the product-formula algorithm using \thm{trotter_error_comm_scaling} and \cor{trotter_number_comm_scaling}.

We claim that the operator
\begin{align}
W_{\gamma_1,\dots,\gamma_{p+1}}
\equiv\commm{H_{\gamma_{p+1}},\dots,\commm{H_{\gamma_2},H_{\gamma_1}}}
\end{align}
is supported on at most $k+p(k-1)$ qubits and $\sum_{\gamma_1,\dots,\gamma_{p+1}=1}^{\Gamma}\norm{W_{\gamma_1,\dots,\gamma_{p+1}}} = \cO{\vertiii{H}_1^p\norm{H}_1} $, where we have used the $1$-norm $\norm{H}_1=\sum_{{j}_1,\ldots,{j}_k}\norm{H_{{j}_1,\ldots,{j}_k}}$ and the induced $1$-norm $\vertiii{H}_1=\max_l\max_{{j}_l}\sum_{\substack{{j}_1,\ldots,{j}_{l-1},{j}_{l+1},\ldots,{j}_{k}}}\norm{H_{{j}_1,\ldots,{j}_k}}$. We prove this claim by induction on $p$. For $p=1$, the commutator $W_{\gamma_1,\gamma_2}=\commm{H_{\gamma_2},H_{\gamma_1}}$ takes the form $\big[H_{{j}_1,\ldots,{j}_k},H_{{i}_1,\ldots,{i}_k}\big]$, which is nonzero only when there exist $l,m=1,\ldots,k$ such that ${j}_l={i}_m$. It then follows that $W_{\gamma_1,\gamma_2}$ is supported on at most $2k-1$ qubits and that
\begin{equation}
	\sum_{\substack{{j}_1,\ldots,{j}_k,\\{i}_1,\ldots,{i}_k}}\norm{\commm{H_{{j}_1,\ldots,{j}_k},H_{{i}_1,\ldots,{i}_k}}}
	\leq 2k^2\max_l\max_{{j}_l}\sum_{\substack{{j}_1,\ldots,{j}_{l-1},\\{j}_{l+1},\ldots,{j}_{k}}}\norm{H_{{j}_1,\ldots,{j}_k}}
	\sum_{{i}_1,\ldots,{i}_k}\norm{H_{{i}_1,\ldots,{i}_k}}
	=\cO{\vertiii{H}_1\norm{H}_1},
\end{equation}
which proves the claim for $p=1$.

Suppose that the claim holds up to $p-1$. Following a similar argument, we have
\begin{equation}
\begin{aligned}
\sum_{\gamma_1,\dots,\gamma_{p+1}=1}^{\Gamma}\norm{W_{\gamma_1,\dots,\gamma_{p+1}}} 
&=\sum_{{j}_1,\ldots,{j}_k}\sum_{\gamma_1,\dots,\gamma_{p}=1}^{\Gamma}\norm{ \commm{H_{{j}_1,\ldots,{j}_k},W_{\gamma_1,\dots,\gamma_p}}}\\
&\leq 2k\big(k+(p-1)(k-1)\big)\max_l\max_{{j}_l}\sum_{\substack{{j}_1,\ldots,{j}_{l-1},\\{j}_{l+1},\ldots,{j}_{k}}}\norm{H_{{j}_1,\ldots,{j}_k}}
\sum_{\gamma_1,\dots,\gamma_{p}=1}^{\Gamma}\norm{W_{\gamma_1,\dots,\gamma_{p}}}\\
&=2k\big(k+(p-1)(k-1)\big)\vertiii{H}_1\cdot\cO{\vertiii{H}_1^{p-1}\norm{H}_1}\\
&=\cO{\vertiii{H}_1^p\norm{H}_1}.
\end{aligned}
\end{equation}
Since the support of $H_{{j}_1,\ldots,{j}_k}$ and $W_{\gamma_1,\dots,\gamma_p}$ overlaps, the operator $W_{\gamma_1,\dots,\gamma_{p+1}}$ acts nontrivially on at most $k+p(k-1)$ qubits. This completes the induction.

\begin{theorem}[Product-formula simulation of $k$-local Hamiltonians]
	\label{thm:pf-local}
	Let $H$ be a $k$-local Hamiltonian on $n$ qubits \eq{local_ham}. Let $\mathscr{S}(t)$ be a $p$th-order product formula. Then, the Trotter error has the scaling
	\begin{equation}
	\norm{\mathscr{S}(t)-e^{-itH}}
	=\cO{\vertiii{H}_1^p\norm{H}_1}.
	\end{equation}
	To simulate with accuracy $\epsilon$, it thus suffices to choose a Trotter number of
	\begin{equation}
	r=\cO{\frac{\vertiii{H}_1\norm{H}_1^{1/p}t^{1+1/p}}{\epsilon^{1/p}}}.
	\end{equation}
	Choosing $p$ sufficiently large, letting $\epsilon$ be constant, and implementing each Trotter step using $\cO{n^k}$ gates, we have the gate complexity
	\begin{equation}
	n^k\vertiii{H}_1\norm{H}_1^{o(1)}t^{1+o(1)}.
	\end{equation}
\end{theorem}

\section{Simulating power-law interactions}
\label{append:power-law}

In this section, we analyze the performance of product formulas for simulating power-law interactions (\sec{app_dqs}). Let $\Lambda\subseteq \mathbb R^d$ be an $n$-qubit $d$-dimensional square lattice. We say that $H$ is a power-law Hamiltonian on $\Lambda$ with an exponent $\alpha$ if it can be written as
\begin{align}
H = \sum_{\vec{i},\vec{j}\in \Lambda} H_{\vec{i},\vec j},
\end{align}
where $H_{\vec i,\vec j}$ is an operator that acts nontrivially only on two qubits $\vec i,\vec j\in \Lambda$ and
\begin{align}
\norm{H_{\vec i,\vec j}}\leq 
\begin{cases}
1, 					&\text{ if } \vec i = \vec j,\\
\frac{1}{\norm{\vec i-\vec j}_2^\alpha},\quad &\text{ if } \vec i\neq \vec j,
\end{cases}
\end{align}
where $\Vert{\vec i - \vec j}\Vert_2$ is the Euclidean distance between $\vec i$ and $\vec j$ on the lattice.

Our analysis uses the following lemma.
\begin{lemma}
	\label{lem:power-law-sum}
	Given an $n$-qubit $d$-dimensional square lattice $\Lambda\subseteq \mathbb R^d$, it holds that
	\begin{equation}
	\label{eq:lamdef2}
	\sum_{\vec j\in\Lambda\setminus \{\vec 0\}} \frac{1}{\big\Vert\vec j\big\Vert_2^{\alpha}} = 
	\begin{cases}
	\cO{n^{1-\alpha/d}},\quad & \text{for }0\leq\alpha<d,\\
	\cO{\log n}, & \text{for }\alpha=d,\\
	\cO 1, & \text{for }\alpha>d.
	\end{cases}
	\end{equation}
	Furthermore, for $\alpha>d$ and $x>0$, we have
	\begin{equation}
	\sum_{\substack{\vec j\in\Lambda,\norm{\vec{j}}_2\geq x}}\frac{1}{\big\Vert\vec j\big\Vert_2^{\alpha}}
	=\cO{\frac{1}{x^{\alpha-d}}}.\label{eq:lamdef3}
	\end{equation}
\end{lemma}
\begin{proof}
	Ref.~\cite{Tran18} provides a detailed proof of the lemma, which follows from rewriting the left-hand side of \cref{eq:lamdef2} as a Riemann sum of the $d$-dimensional integral 
		$\int_{\norm{\vec j}_2\geq 1} {\text{d}^d\vec j}/{\Vert{\vec j}\Vert_2} $. 
	Evaluating the integral gives the right-hand side of \cref{eq:lamdef2}.
	Similarly, \cref{eq:lamdef3} follows from evaluating the integral $\int_{\norm{\vec j}_2\geq x} {\text{d}^d\vec j}/{\Vert{\vec j}\Vert_2}\propto \frac{1}{x^{\alpha-d}}$.
\end{proof}

\begin{theorem}[Product-formula simulation of power-law interactions]
	\label{thm:pf-power-law}
	Let $\Lambda\subseteq \mathbb R^d$ be an $n$-qubit $d$-dimensional square lattice and $H$ be a power-law Hamiltonian \eq{power-law-def} with exponent $\alpha$. Let $\mathscr{S}(t)$ be a $p$th-order product formula. Then, the Trotter error has the scaling
	\begin{equation}
	\norm{\mathscr{S}(t)-e^{-itH}}=
	\begin{cases}
	\cO{n^{1+(p+1)(1-\alpha/d)} t^{p+1}},\quad & \text{ for }0\leq \alpha<d,\\
	\cO{n(\log n)^{p+1} t^{p+1}}, & \text{ for }\alpha=d,\\
	\cO{nt^{p+1}}, & \text{ for } \alpha>d.
	\end{cases}
	\end{equation}
	To simulate with accuracy $\epsilon$, it thus suffices to choose a Trotter number of
	\begin{equation}
	r=
	\begin{cases}
	\cO{n^{1-\frac{\alpha}{d}+\frac{1}{p}\left(2-\frac{\alpha}{d}\right)} t^{1+\frac{1}{p}}/\epsilon^{\frac{1}{p}}},\quad & \text{ for }0\leq \alpha<d,\\
	\cO{n^{\frac{1}{p}}(\log n)^{1+\frac{1}{p}} t^{1+\frac{1}{p}}/\epsilon^{\frac{1}{p}}}, & \text{ for }\alpha=d,\\
	\cO{n^{\frac{1}{p}} t^{1+\frac{1}{p}}/\epsilon^{\frac{1}{p}}}, & \text{ for } \alpha>d.
	\end{cases}
	\end{equation}
	Choosing $p$ sufficiently large, letting $\epsilon$ be constant, and implementing each Trotter step using $\cO{n^2}$ gates, we have the gate complexity
	\begin{align}
	g_\alpha 
	=\begin{cases}
	n^{3-\frac{\alpha}{d}+o(1)} t^{1+o(1)},\quad & \text{ for }0\leq \alpha<d,\\
	n^{2+o(1)} t^{1+o(1)}, & \text{ for }\alpha\geq d.
	\end{cases}
	\end{align}
\end{theorem}
\begin{proof}
	Given a power-law Hamiltonian $H$ with exponent $\alpha$, we use \lem{power-law-sum} to compute the scaling of its induced $1$-norm
	\begin{equation}
	\label{eq:power-law-vertiii}
	\vertiii{H}_1\leq
	1+\max_{\vec i}\sum_{\vec j\neq \vec{i}} \frac{1}{\big\Vert\vec{i}-\vec j\big\Vert_2^{\alpha}}=
	\begin{cases}
	\cO{n^{1-\alpha/d}},\quad & \text{for }0\leq\alpha<d,\\
	\cO{\log n}, & \text{for }\alpha=d,\\
	\cO 1, & \text{for }\alpha>d,
	\end{cases}
	\end{equation}
	and $1$-norm
	\begin{equation}
	\label{eq:power-law-norm}
	\norm{H}_1\leq
	\sum_{\vec i}\bigg(1+\sum_{\vec j\neq \vec{i}}\frac{1}{\big\Vert\vec{i}-\vec j\big\Vert_2^{\alpha}}\bigg)=
	\begin{cases}
	\cO{n^{2-\alpha/d}},\quad & \text{for }0\leq\alpha<d,\\
	\cO{n\log n}, & \text{for }\alpha=d,\\
	\cO n, & \text{for }\alpha>d.
	\end{cases}
	\end{equation}
	The claim then follows from \thm{pf-local}.
\end{proof}

As mentioned in \sec{app_dqs}, the performance of product formulas can be further improved for rapidly decaying power-law interactions (\thm{pf-weak-power-law}) and quasilocal interactions (\thm{pf-quasi-local}).
\begin{theorem}[Product-formula simulation of rapidly decaying power-law interactions]
	\label{thm:pf-weak-power-law}
	Let $\Lambda\subseteq \mathbb R^d$ be an $n$-qubit $d$-dimensional square lattice and $H$ be a power-law Hamiltonian \eq{power-law-def} with exponent $\alpha>d$. Let $\mathscr{S}(t)$ be a $p$th-order product formula for $\widetilde{H}$, the truncated Hamiltonian where summands acting on sites with distance larger than $\ell$ are removed. Then, the Trotter error has the scaling
	\begin{equation}
	\norm{\mathscr{S}(t)-e^{-it\widetilde H}}=\cO{nt^{p+1}}.
	\end{equation}
	To simulate with accuracy $\epsilon$, it thus suffices to choose the cutoff $\ell =  \Th{\left({nt}/{\epsilon}\right)^{1/(\alpha-d)}}$ and a Trotter number of
	\begin{equation}
	r=\cO{n^{\frac{1}{p}} t^{1+\frac{1}{p}}/\epsilon^{\frac{1}{p}}}.
	\end{equation}
	Choosing $p$ sufficiently large, letting $\epsilon$ be constant, and implementing each Trotter step using $\cO{n\ell^d}$ gates, we have the gate complexity
	\begin{align}
	g_\alpha 
	=(nt)^{1+\frac{d}{\alpha-d}+o(1)}.
	\end{align}
\end{theorem}
\begin{proof}
	We use \lem{power-law-sum} to bound the distance between the original and the truncated Hamiltonian
	\begin{equation}
	\norm{H-\widetilde{H}}\leq
	\sum_{\vec i}\sum_{\norm{\vec j-\vec{i}}_2> \ell} \frac{1}{\big\Vert\vec{i}-\vec j\big\Vert_2^{\alpha}}=\cO{\frac{n}{\ell^{\alpha-d}}}.
	\end{equation}
	We choose a cutoff value $\ell =  \Th{\left({nt}/{\epsilon}\right)^{1/(\alpha-d)}}$ and \cor{time_ordered_distance_bound} implies that the truncation error is at most
	\begin{align}
	\norm{e^{-itH} - e^{-it\widetilde{H}}}  \leq \norm{H-\widetilde{H}}t=\cO{\frac{nt}{\ell^{\alpha-d}}}=\cO{\epsilon}.
	\end{align}
	The theorem is then proved in a similar way as \thm{pf-power-law}.
\end{proof}
\begin{theorem}[Product-formula simulation of quasilocal interactions]
	\label{thm:pf-quasi-local}
	Let $\Lambda\subseteq \mathbb R^d$ be an $n$-qubit $d$-dimensional square lattice and $H$ be a quasilocal Hamiltonian \eq{quasi-local-def} with constant $\beta>0$. Let $\mathscr{S}(t)$ be a $p$th-order product formula for $\widetilde{H}$, the truncated Hamiltonian where summands acting on sites with distance larger than $\ell$ are removed. Then, the Trotter error has the scaling
	\begin{equation}
	\norm{\mathscr{S}(t)-e^{-it\widetilde H}}=\cO{nt^{p+1}}.
	\end{equation}
	To simulate with accuracy $\epsilon$, it thus suffices to choose the cutoff $\ell = \Theta(\log(nt/\epsilon))$ and a Trotter number of
	\begin{equation}
	r=\cO{n^{\frac{1}{p}} t^{1+\frac{1}{p}}/\epsilon^{\frac{1}{p}}}.
	\end{equation}
	Choosing $p$ sufficiently large, letting $\epsilon$ be constant, and implementing each Trotter step using $\cO{n\ell^d}$ gates, we have the gate complexity
	\begin{align}
	g_\beta 
	=(nt)^{1+o(1)}.
	\end{align}
\end{theorem}
\begin{proof}
	We choose $\ell = \Theta(\log(nt/\epsilon))$ so that the truncation error is at most
	\begin{equation}
	\norm{e^{-itH} - e^{-it\widetilde{H}}}
	\leq \norm{H-\widetilde{H}}t
	\leq\sum_{\vec i}\sum_{\norm{\vec j-\vec{i}}_2> \ell}e^{-\beta\norm{\vec{i}-\vec{j}}_2}\cdot t
	=\cO{\epsilon}.
	\end{equation}
	The remaining analysis proceeds in a similar way as in \thm{pf-weak-power-law}.
\end{proof}

\section{Simulating clustered Hamiltonians}
\label{append:cluster}

We continue the analysis in \sec{app_dqs} of the hybrid algorithm for simulating clustered Hamiltonians \cite{PHOW19}. An essential step of this algorithm is to decompose the Hamiltonian into parties using product formulas. We show that our Trotter error bound implies a more efficient decomposition and thereby gives a faster simulation of clustered Hamiltonians.

Let $H$ be an $n$-qubit Hamiltonian. Assume that each term in $H$ acts on at most two qubits with spectral norm at most one, and each qubit is interacted with at most a constant number $d'$ of qubits. We further group the qubits into multiple parties and write
\begin{equation}
H=A+B=\sum_l H_l^{(1)}+\sum_l H_l^{(2)},\quad \forall l:\norm{H_l^{(1)}},\norm{H_l^{(2)}}\leq 1,
\end{equation}
where terms in $A$ act on qubits within a single party and terms in $B$ act between two different parties.

The hybrid algorithm of \cite{PHOW19} applies the first-order Lie-Trotter formula to decompose the Hamiltonian $H=A+\sum_l H_l^{(2)}$. Their analysis shows that a Trotter number of
\begin{equation}
r=\cO{\frac{h_B^2 t^2}{\epsilon}}
\end{equation}
suffices to achieve error at most $\epsilon$, where $h_B=\sum_l\norm{H_l^{(2)}}$ is the interaction strength. Here, we show that it suffices to take
\begin{equation}
r=\cO{\frac{d'^{\frac{1+p}{2}} h_B^{\frac{1}{p}} t^{1+\frac{1}{p}}}{\epsilon^{\frac{1}{p}}}}=\cO{\frac{h_B^{1/p} t^{1+1/p}}{\epsilon^{1/p}}}
\end{equation}
using a $p$th-order product formula. This improves the analysis of \cite{PHOW19} for $p=1$ and extends their result to higher-order cases.

In light of \thm{trotter_error_comm_scaling}, we need to compute
\begin{equation}
\sum_{\gamma_1,\gamma_2,\ldots,\gamma_{p+1}}\norm{\big[H_{\gamma_{p+1}},\cdots\big[H_{\gamma_2},H_{\gamma_1}\big]\cdots\big]},
\end{equation}
where each $H_\gamma$ is either $H_l^{(2)}$ or $A$. Since $\big[A,A\big]=0$ and $\big[H_{\gamma},A\big]=-\big[A,H_{\gamma}\big]$, we may without loss of generality assume that $H_{\gamma_1}=H_{l_1}^{(2)}$, i.e.,
\begin{equation}
\sum_{\gamma_1,\gamma_2,\ldots,\gamma_{p+1}}\norm{\big[H_{\gamma_{p+1}},\cdots\big[H_{\gamma_2},H_{\gamma_1}\big]\cdots\big]}
=\sum_{l_1,\gamma_2,\ldots,\gamma_{p+1}}\norm{\big[H_{\gamma_{p+1}},\cdots\big[H_{\gamma_2},H_{l_1}^{(2)}\big]\cdots\big]}.
\end{equation}
We now replace each $A$ by $\sum_l H_l^{(1)}$ and apply the triangle inequality to get
\begin{equation}
\sum_{l_1,\gamma_2,\ldots,\gamma_{p+1}}\norm{\big[H_{\gamma_{p+1}},\cdots\big[H_{\gamma_2},H_{l_1}^{(2)}\big]\cdots\big]}
\leq\sum_{l_1,l_2,\ldots,l_{p+1}}\norm{\big[K_{l_{p+1}},\cdots\big[K_{l_2},H_{l_1}^{(2)}\big]\cdots\big]},
\end{equation}
where each $K_l$ is either $H_l^{(1)}$ or $H_l^{(2)}$. Since each qubit supports at most $d'$ terms and each term acts on at most two qubits,
\begin{equation}
\sum_{l_1,l_2,\ldots,l_{p+1}}\norm{\big[K_{l_{p+1}},\cdots\big[K_{l_2},H_{l_1}^{(2)}\big]\cdots\big]}
=\cO{d'^{p}\cdots d'^2d'\sum_{l_1}\norm{H_{l_1}^{(2)}}}=\cO{d'^{\frac{(1+p)p}{2}}h_B}.
\end{equation}
We have thus established:
\begin{theorem}[Product-formula decomposition of evolutions of clustered Hamiltonians]
	\label{thm:pf-cluster}
	Let $H=A+\sum_l H_l^{(2)}$ be a clustered Hamiltonian as in \eq{clustered_ham}, where each qubit is interacted with at most a constant number $d'$ of qubits and the interaction strength is $h_B=\sum_l\norm{H_l^{(2)}}$. Let $\mathscr{S}(t)$ be a $p$th-order product formula as in \eq{clustered_pf}. Then, the Trotter error has the scaling
	\begin{equation}
	\norm{\mathscr{S}(t)-e^{-itH}}
	=\cO{d'^{\frac{(1+p)p}{2}}h_Bt^{p+1}}.
	\end{equation}
	To decompose with accuracy $\epsilon$, it thus suffices to choose a Trotter number of
	\begin{equation}
	r=\cO{\frac{h_B^{1/p} t^{1+1/p}}{\epsilon^{1/p}}}.
	\end{equation}
	Choosing $p$ sufficiently large, we have
	\begin{equation}
	r=\cO{\frac{h_B^{o(1)} t^{1+o(1)}}{\epsilon^{o(1)}}}.
	\end{equation}
\end{theorem}

The hybrid simulator of \cite{PHOW19} has runtime $2^{O(r\cdot\cc(g))}$, where $\cc(g)$ is the contraction complexity of the interaction graph $g$ between the parties. \thm{pf-cluster} thus gives a hybrid simulator with complexity $2^{\cO{h_B^{o(1)} t^{1+o(1)} \cc(g)/\epsilon^{o(1)}}}$, dramatically improving the previous result of $2^{\cO{h_B^2t^2 \cc(g)/\epsilon}}$ \cite{PHOW19}.

\section{Simulating local observables}
\label{append:local-obs}

In this section, we analyze the performance of product formulas for simulating local observables. Following \sec{app_local_obs}, we consider a power-law Hamiltonian $H = \sum_{\vec i,\vec j\in\Lambda} H_{\vec i\vec j}$ on an $n$-qubit $d$-dimensional lattice $\Lambda\subseteq \mathbb R^d$ with exponent $\alpha>2d$. Our goal is to simulate the time evolution $\scA(t) = e^{itH} A e^{-itH}$ of a local observable $A$ with support $\supp(A)$ enclosed in a $d$-dimensional ball of constant radius $x_0$.

As mentioned in \sec{app_local_obs}, our approach is to construct a Hamiltonian $\Hlc$ whose support has radius independent of the system size. To this end, we consider a general observable $B$ and assume that $\mathcal S(B)$---the support of $B$---is a $d$-dimensional ball of radius $y_0$ centered on the origin. We define
\begin{align}
&H_1 = \sum_{\vec i,\vec j\in \mathcal B_{\ell}} H_{\vec i,\vec j},\\
&H_\gamma = \sum_{\vec i,\vec j\in\Delta \mathcal B_{\gamma\ell}} H_{\vec i,\vec j}+ \sum_{\substack{
		\vec i\in\Delta \mathcal B_{(\gamma-1)\ell}\\
		\vec j\in\Delta \mathcal B_{\gamma\ell}}
} H_{\vec i,\vec j} \quad\text{ for }\gamma=2,\dots,\Gamma-1,\\
&H_{\Gamma} = \sum_{\vec i,\vec j\notin\mathcal B_{(\Gamma-2)\ell}} H_{\vec i,\vec j},
\end{align}
where $\mathcal B_{y} = \{\vec i\in \Lambda:\inf_{\vec{j}\in\mathcal S(B)}\norm{\vec{i}-\vec{j}}_2\leq y\}$ is a ball of radius $y+y_0$ centered on $\mathcal S(B)$, $\Delta\mathcal B_{\gamma\ell} = \mathcal B_{\gamma\ell}\setminus\mathcal B_{(\gamma-1)\ell}$ is the shell containing sites between distance $(\gamma-1)\ell$ and $\gamma\ell$ from $\mathcal S(B)$, and $\ell$, $\Gamma$ are positive integers to be chosen later. We then define the truncated Hamiltonian $\Htrunc = \sum_{\gamma=1}^{\Gamma} H_\gamma$. We analyze the truncation error in the lemma below.

\begin{lemma}
	\label{lem:norm_delta_H}
	Let $\Lambda\subseteq \mathbb R^d$ be a $d$-dimensional square lattice of $n$ qubits. Let $H$ be a power-law Hamiltonian with exponent $\alpha>d$ and $B$ be an observable with support enclosed in a $d$-dimensional ball of radius $y_0$. Let $\Htrunc = \sum_{\gamma=1}^{\Gamma} H_\gamma$ be the truncated Hamiltonian as defined above. Assuming $\Gamma = \cO{1}$, we have
	\begin{align}
	\norm {H-\Htrunc} = \cO{\frac{(y_0+\Gamma\ell)^{d-1}}{\ell^{\alpha-d-1}}}.
	\end{align}
\end{lemma}
\begin{proof}
	We expand $H-\Htrunc$ as
	\begin{align}
	H-\Htrunc
	= \sum_{\gamma = 0}^{\Gamma-2} \sum_{\vec{i}\in \Delta  \mathcal B_{\gamma \ell}}\sum_{\nu \geq \gamma+2} \sum_{\vec{j}\in \Delta \mathcal B_{\nu \ell}} H_{\vec{i},\vec{j}}.
	\end{align}
	Applying the triangle inequality, we have
	\begin{align}
	\norm{H-\Htrunc}
	&\leq \sum_{\gamma = 0}^{\Gamma-2} \sum_{\vec i\in \Delta  \mathcal B_{\gamma \ell}}\sum_{\nu \geq \gamma+2} \sum_{\vec j\in \Delta \mathcal B_{\nu \ell}} \norm{H_{\vec i,\vec j}}\\
	&\leq \sum_{\gamma = 0}^{\Gamma-2} \sum_{\vec i\in \mathcal B_{\gamma\ell}}\sum_{\vec j\notin  \mathcal B_{(\gamma+1) \ell}} \norm{H_{\vec i,\vec j}}\\
	&\leq \sum_{\gamma = 0}^{\Gamma-2} \frac{(y_0+\Gamma\ell)^{d-1}}{\ell^{\alpha-d-1}}
	= \cO{\frac{(y_0+\Gamma\ell)^{d-1}}{\ell^{\alpha-d-1}}},
	\end{align}
	where the third inequality follows from \cite[Lemma 9]{Tran18} (see also \cite[the derivation of Eq.(A1)]{Tran18}, with $A$ and $C$ being $\mathcal B_{\gamma\ell}$ and the complement of $\mathcal B_{(\gamma+1)\ell}$ respectively).
	The factor $(y_0+\Gamma\ell)^{d-1}$ estimates the boundary area of $\mathcal B_{\gamma\ell}$.
	This establishes the claimed scaling of the truncation error.
\end{proof}

Next, we simulate the evolution $e^{-it\Htrunc}$ using the $p$th-order product formula
\begin{align}
\Strunc(t) = \prod_{\upsilon=1}^{\Upsilon} \prod_{\gamma=1}^{\Gamma} e^{-i t a_{(\upsilon,\gamma) }H_{\pi_{\upsilon}(\gamma)}},
\end{align}
where we put additional constraints on the permutation function $\pi_\nu$:
\begin{align}
\pi_{\upsilon}(1,2,3,4,5,6,\ldots) = \begin{cases}
(2,4,6,\dots,1,3,5,\dots),\quad& \text{ if } \upsilon \text{ is odd},\\
(1,3,5,\dots,2,4,6,\dots),& \text{ if } \upsilon \text{ is even}.
\end{cases}
\end{align}
By \thm{trotter_error_comm_scaling}, the Trotter error of approximating $e^{-it\Htrunc}$ by $\Strunc(t)$ depends on
\begin{align}
\sum_{\gamma_1,\dots,\gamma_{p+1}=1}^{\Gamma} \norm{\commm{H_{\gamma_{p+1}},\dots,\commm{H_{\gamma_2},H_{\gamma_1}}}},
\end{align}
which we analyze in the following lemma.

\begin{lemma}
	\label{lem:norm_nested_comm}
	Let $\Lambda\subseteq \mathbb R^d$ be a $d$-dimensional square lattice of $n$ qubits. Let $H$ be a power-law Hamiltonian with exponent $\alpha>d$ and $B$ be an observable with support enclosed in a $d$-dimensional ball of radius $y_0$. Let $\Htrunc = \sum_{\gamma=1}^{\Gamma} H_\gamma$ be the truncated Hamiltonian as defined above. Assuming $\Gamma = \cO{1}$, we have
	\begin{align}
	\sum_{\gamma_1,\dots,\gamma_{p+1}=1}^{\Gamma} \norm{\commm{H_{\gamma_{p+1}},\dots,\commm{H_{\gamma_2},H_{\gamma_1}}}}
	=\cO{(y_0+\Gamma\ell)^{d-1}\ell}.
	\end{align}
\end{lemma}
\begin{proof}
	For convenience, we define
	\begin{align}
	\mathcal{S}_{1}&=\left\{(\vec i,\vec j): \vec i,\vec j\in \mathcal B_{\ell}\right\},\\
	\mathcal{S}_{\gamma}&=\left\{(\vec i,\vec j):\vec i\in \Delta \mathcal B_{(\gamma-1)\ell}\cup\Delta \mathcal B_{\gamma\ell},\ \vec j\in\Delta \mathcal B_{\gamma\ell}\right\}
	\quad \text{ for }\gamma=2,\dots,\Gamma-1,\\
	\mathcal{S}_{\Gamma}&=\left\{(\vec i,\vec j):\vec i,\vec j\notin\mathcal B_{(\Gamma-2)\ell}\right\},
	\end{align}
	so that $H_\gamma=\sum_{(\vec i,\vec j)\in\mathcal{S}_\gamma} H_{\vec{i},\vec{j}}$ for $\gamma=1,\dots,\Gamma$. Our goal is to analyze
	\begin{equation}
		\sum_{\gamma_1,\dots,\gamma_{p+1}=1}^{\Gamma} \bigg\Vert\bigg[\sum_{(\vec{i}_{p+1},\vec{j}_{p+1})\in\mathcal{S}_{\gamma_{p+1}}} H_{\vec{i}_{p+1},\vec{j}_{p+1}}
				,\dots,\bigg[\sum_{(\vec{i}_2,\vec{j}_2)\in\mathcal{S}_{\gamma_2}} H_{\vec{i}_2,\vec{j}_2}
					,\sum_{(\vec{i}_1,\vec{j}_1)\in\mathcal{S}_{\gamma_1}} H_{\vec{i}_1,\vec{j}_1}\bigg]\bigg]\bigg\Vert.
	\end{equation}
	Note that at least one of $\gamma_1,\gamma_2$ must be different from $1,\Gamma$; otherwise, $\commm{H_{\gamma_1},H_{\gamma_2}}= 0$. Therefore, we may assume $1<\gamma_1< \Gamma$ without loss of generality and bound the norm of commutators as
	\begin{equation}
		\sum_{\gamma_1=2}^{\Gamma-1}\sum_{\gamma_2,\dots,\gamma_{p+1}=1}^{\Gamma}
		\sum_{\substack{(\vec{i}_1,\vec{j}_1)\in\mathcal{S}_{\gamma_1},(\vec{i}_2,\vec{j}_2)\in\mathcal{S}_{\gamma_2},\\\ldots,(\vec{i}_{p+1},\vec{j}_{p+1})\in\mathcal{S}_{\gamma_{p+1}}}}
		\big\Vert\big[H_{\vec{i}_{p+1},\vec{j}_{p+1}}
		,\dots,\big[ H_{\vec{i}_2,\vec{j}_2}
		,H_{\vec{i}_1,\vec{j}_1}\big]\big]\big\Vert.
	\end{equation}
	By a similar argument as in \thm{pf-local} and \thm{pf-power-law}, we find the upper bound
	\begin{equation}
		\cO{\sum_{\gamma_1=2}^{\Gamma-1}\sum_{(\vec{i}_1,\vec{j}_1)\in\mathcal{S}_{\gamma_1}}\vertiii{H}_1^{p}\norm{H_{\vec{i}_1,\vec{j}_1}}}
		=\cO{\sum_{\gamma_1=2}^{\Gamma-1}\sum_{\vec{j}_1\in\Delta \mathcal B_{\gamma_1\ell}}\vertiii{H}_1^{p+1}}
		=\cO{(y_0+\Gamma\ell)^{d-1}\ell},
	\end{equation}
	where we use the fact that $\vertiii{H}_1=\cO{1}$ \eq{power-law-vertiii} and upper bound the volume of $\Delta\mathcal B_{\gamma_1\ell}$ by $\cO{(y_0+\Gamma\ell)^{d-1}\ell}$---the product of its boundary area $\cO{(y_0+\Gamma\ell)^{d-1}}$ with its thickness $\ell$.
\end{proof}

Using the fact that product formulas can preserve the locality of the simulated system, we commute the matrix exponentials in $\Strunc(t)$ through $B$ to cancel with their counterpart in $\Strunc^\dag(t)$. By choosing $\Gamma = \Upsilon+1$, we have
\begin{align}
\Strunc^\dag(t) B \Strunc(t) = \Sreduce^\dag (t) B \Sreduce(t),
\end{align}
where
\begin{align}
\Sreduce(t) = \prod_{\upsilon=1}^{\Upsilon} \prod_{\gamma=1}^{\upsilon} e^{-i t a_{(\upsilon,\pi_{\upsilon}^{-1}(\gamma)) }H_{\gamma}}
\end{align}
is the reduced product formula of $\scA(t)$.
This gives a decomposition of the evolution of local observable $B$ with error
\begin{equation}
\begin{aligned}
\norm{e^{itH} Be^{-itH}-\Sreduce^\dag(t) B \Sreduce(t)}
&= \cO{\norm{B}t(y_0+\Gamma\ell)^{d-1}\left(\frac{1}{\ell^{\alpha-d-1}} + \ell t^{p}\right)}.
\end{aligned}
\end{equation}
The remaining analysis proceeds in the same way as in \sec{app_local_obs}.
We have then proved:
\proplocalobs*

Assuming $\distance_0 = \cO{1}$ and $\Gamma=\Upsilon = \cO{1}$, we have
\begin{equation}
\begin{aligned}
\norm{e^{itH} Ae^{-itH}-e^{it\Hlc} A e^{-it\Hlc}}
&=\cO{t(r\ell)^{d-1}\left(\frac{1}{\ell^{\alpha-d-1}}+\frac{\ell t^p}{r^{p}}\right)}
\end{aligned}
\end{equation}
where $\Hlc$ is supported on a ball of radius $\distance = \distance_0 + r\Gamma \ell = \cO{r\ell}$. To minimize the error, we choose $\ell = \Th{\left(\frac{r}{t}\right)^{\frac{p}{\alpha-d}}}\geq 1$, which is possible if $r\geq t$ and $\alpha>2d$.
With this choice of $\ell$, the error becomes
\begin{align}
\norm{e^{itH} Ae^{-itH}-e^{it\Hlc} A e^{-it\Hlc}}
= \cO{\frac{t^{\frac{p(\alpha-2d)+\alpha-d}{\alpha-d}}}{r^{\frac{p(\alpha-2d)-(\alpha-d)(d-1)}{\alpha-d}}}}
.\label{eq:tot_error_4}
\end{align}
To ensure that the error is at most $\epsilon$, we choose
\begin{align}
r
= \Th{\frac
	{t^{\frac{p(\alpha-2d)+\alpha-d}{{p(\alpha-2d)-(\alpha-d)(d-1)}}}}
	{\epsilon^{\frac{\alpha-d}{p(\alpha-2d)-(\alpha-d)(d-1)}}}}.\label{eq:loc_ob_r_choice}
\end{align}
Note that $r$ can be made to be greater than $1$ for large times if the exponent of $t$ in the above equation is positive, i.e. we require (assuming $\alpha>2d$)
\begin{align}
\alpha > \frac{2d - \frac{d(d-1)}{p}}{1 - \frac{d-1}{p}}
\Leftrightarrow p > \frac{(\alpha-d)(d-1)}{\alpha-2d}.
\end{align}
In addition, the choice of $r$ above is also consistent with the condition $r \geq t$ because
\begin{align}
\frac{p(\alpha-2d)+\alpha-d}{{p(\alpha-2d)-(\alpha-d)(d-1)}}>1.
\end{align}

Recall that $e^{-it\Hlc}$ is an evolution supported on a ball of radius $\distance = \cO{r\ell}$. Invoking \thm{pf-weak-power-law}, we obtain the gate count
\begin{align}
g_\alpha = \cO{{(\distance^dt)^{1+\frac{1}{p}+\frac{d}{\alpha-d}}}}
&=\cO{
	{t^{\frac{(\alpha  (p+1)-d) (\alpha  (d p+p+1)-(d+2) d p-d)}{p (\alpha-d) \left(\alpha +d^2-d (\alpha +2 p+1)+\alpha  p\right)}}}
}
\end{align}
for simulating local observable $A$ with constant accuracy, which simplifies to
\begin{align}
g_\alpha
&=
{t^{\frac
		{\alpha (\alpha  (d+1)-(d+2) d )}
		{(\alpha-d) (\alpha-2d)}+o(1)}}
={
	t^{\left(1+d\frac{\alpha -d}{\alpha -2 d}\right) \left(1+\frac{d}{\alpha -d}\right)+o(1)}
}
\end{align}
in the large $p$ limit. The remaining analysis proceeds as in \sec{app_local_obs}.

\section{Quantum Monte Carlo simulation}
\label{append:qmc}

In this section, we apply our Trotter error bound to improve the performance of quantum Monte Carlo simulation. This analysis is sketched in \sec{app_qmc} and detailed here.

We first consider simulating an $n$-qubit transverse field Ising Hamiltonian \cite{Bravyi15}
\begin{equation}
H=-A-B,\quad
A=\sum_{1\leq u<v\leq n}j_{u,v}Z_uZ_v,\quad
B=\sum_{1\leq u\leq n}h_uX_u,
\end{equation}
where $X_u$ and $Z_u$ are Pauli operators acting on the $u$th qubit, and $j_{u,v}\geq 0$ and $h_u\geq 0$ are nonnegative coefficients. Our goal is to approximate the partition function $\mathcal{Z}=\mathrm{Tr}\big(e^{-H}\big)$ up to multiplicative error $0<\epsilon<1$.

A key step in the algorithm of \cite{Bravyi15} is to decompose the evolution operator using the second-order Suzuki formula. Note that all the summands in $A$ (or $B$) commute with each other, so no error is introduced when the evolution under $A$ (or $B$) is further decomposed into elementary exponentials. It thus suffices to analyze the Trotter error of approximating $e^{t(A+B)}$ by $e^{\frac{t}{2}A}e^{tB}e^{\frac{t}{2}A}$ for time $t>0$. To this end, we define
\begin{equation}
\begin{aligned}
U&:=e^{t(A+B)},\\
V&:=e^{\frac{t}{2}A}e^{tB}e^{\frac{t}{2}A},\\
W&:=\exp_{\mathcal{T}}\bigg(\int_{0}^{t}\mathrm{d}\tau\ e^{-\tau (A+B)}\bigg[e^{\frac{\tau}{2}A}Be^{-\frac{\tau}{2}A}-B
+e^{\frac{\tau}{2}A}e^{\tau B}\frac{A}{2}e^{-\tau B}e^{-\frac{\tau}{2}A}-\frac{A}{2}\bigg]e^{\tau (A+B)}\bigg)
\end{aligned}
\end{equation}
so that \thm{error_type} implies $V=UW$. To analyze the operator $W$, we further compute
\begin{align}
	e^{\frac{\tau}{2}A}Be^{-\frac{\tau}{2}A}-B
	&=\bigg[\frac{A}{2},B\bigg]\tau
	+\int_{0}^{\tau}\mathrm{d}\tau_2\int_{0}^{\tau_2}\mathrm{d}\tau_3\ e^{\frac{\tau_3}{2}A}\bigg[\frac{A}{2},\bigg[\frac{A}{2},B\bigg]\bigg]e^{-\frac{\tau_3}{2}A},\\
	e^{\frac{\tau}{2}A}e^{\tau B}\frac{A}{2}e^{-\tau B}e^{-\frac{\tau}{2}A}-\frac{A}{2}
	&=e^{\frac{\tau}{2}A}\bigg(\bigg[B,\frac{A}{2}\bigg]\tau
	+\int_{0}^{\tau}\mathrm{d}\tau_2\int_{0}^{\tau_2}\mathrm{d}\tau_3\ e^{\tau_3 B}\bigg[B,\bigg[B,\frac{A}{2}\bigg]\bigg]e^{-\tau_3 B}\bigg)e^{-\frac{\tau}{2}A} \nonumber\\
	&=\bigg[B,\frac{A}{2}\bigg]\tau
	+\tau\int_{0}^{\tau}\mathrm{d}\tau_2\ e^{\frac{\tau_2}{2}A}\bigg[\frac{A}{2},\bigg[B,\frac{A}{2}\bigg]\bigg]e^{-\frac{\tau_2}{2}A} \nonumber\\
	&\quad +\int_{0}^{\tau}\mathrm{d}\tau_2\int_{0}^{\tau_2}\mathrm{d}\tau_3\ e^{\frac{\tau}{2}A}e^{\tau_3 B}\bigg[B,\bigg[B,\frac{A}{2}\bigg]\bigg]e^{-\tau_3 B}e^{-\frac{\tau}{2}A}.
\end{align}
By \lem{time_ordered_norm_bound}, we have
\begin{equation}
\begin{aligned}
	\norm{W}
	\leq\exp\bigg(
	e^{2t\norm{H}+t\norm{A}}\frac{t^3}{24}\norm{[A,[A,B]]}
	+e^{2t\norm{H}+t\norm{A}}\frac{t^3}{12}\norm{[A,[A,B]]}
	+e^{2t\norm{H}+t\norm{A}+2t\norm{B}}\frac{t^3}{12}\norm{[B,[B,A]]}
	\bigg).
\end{aligned}
\end{equation}
This bound is tighter than the previous result of \cite[Lemma 3]{Bravyi15} in that it exploits the commutativity of operator summands. For the transverse field Ising model, this leads to an asymptotic improvement on the performance of Monte Carlo simulation. The remaining analysis proceeds as in \sec{app_qmc}.

We also consider simulating the ferromagnetic quantum spin systems
\begin{equation}
H=\sum_{1\leq u<v\leq n}p_{uv}\big({-}X_uX_v-Y_uY_v\big)+\sum_{1\leq u<v\leq n}q_{uv}\big({-}X_uX_v+Y_uY_v\big)+\sum_{u=1}^{n}d_u\big(I+Z_u\big),
\end{equation}
where $p_{uv},q_{uv}\in[0,1]$. Our goal is to approximate the partition function $\mathcal{Z}(\beta,H)=\mathrm{Tr}\big[e^{-\beta H}\big]$ for $\beta>0$. Following \cite{BG17}, we restrict ourselves to the $n$-qubit (nonunitary) gate set
\begin{equation}
\bigg\{f_{u}\big(e^{\pm t}\big),g_{uv}(t),h_{uv}(t)\ \bigg|\ u,v=1,\ldots,n,\ u\neq v,\ 0<t<\frac{1}{2}\bigg\},
\end{equation}
where
\begin{equation}
f\big(e^{\pm t}\big)=
\begin{bmatrix}
e^{\pm t} & 0\\
0 & 1
\end{bmatrix}
,\qquad g(t)=
\begin{bmatrix}
1+t^2 & 0 & 0 & t\\
0 & 1 & 0 & 0\\
0 & 0 & 1 & 0\\
t & 0 & 0 & 1
\end{bmatrix}
,\qquad h(t)=
\begin{bmatrix}
1 & 0 & 0 & 0\\
0 & 1+t^2 & t & 0\\
0 & t & 1 & 0\\
0 & 0 & 0 & 1
\end{bmatrix}
\end{equation}
and the subscripts $u,v$ indicate the qubits on which the gates act nontrivially. These gates approximate the exponentials of terms of the original Hamiltonian. Specifically, we represent the gates as
\begin{equation}
f_u\big(e^{\pm t}\big)=e^{\pm\frac{t}{2}F_u},\quad g_{uv}(t)=e^{-\frac{t}{2}\widetilde{\mathscr{G}}_{uv}(t)},\quad h_{uv}(t)=e^{-\frac{t}{2}\widetilde{\mathscr{H}}_{uv}(t)},
\end{equation}
where $0<t<1/2$ and
\begin{equation}
\label{eq:perturb_term}
F_{u}=(I+Z_u),\quad
\widetilde{\mathscr{G}}_{uv}(t)=(-X_uX_v+Y_uY_v)-\frac{2}{t}\mathscr{G}_{uv}(t),\quad
\widetilde{\mathscr{H}}_{uv}(t)=(-X_uX_v-Y_uY_v)-\frac{2}{t}\mathscr{H}_{uv}(t).
\end{equation}
By \cite[Proposition 1]{BG17}, we have $\norm{\mathscr{G}_{uv}(t)}\leq t^2$ and $\norm{\mathscr{H}_{uv}(t)}\leq t^2$.

We divide the evolution into $r$ steps. We choose $r> 2\beta$ so that we can implement the product formula using gates from \eq{gate_set} with parameters
\begin{equation}
-\frac{1}{2}<-\frac{\beta}{r}d_u<\frac{1}{2},\quad
0<\frac{\beta}{r}q_{uv}<\frac{1}{2},\quad
0<\frac{\beta}{r}p_{uv}<\frac{1}{2}.
\end{equation}
Consider the gate sequence
\begin{equation}
\begin{aligned}
&\prod_{1\leq u\leq n}f_u\big(e^{-\frac{\beta}{r}d_u}\big)\prod_{1\leq u<v\leq n}g_{uv}\bigg(\frac{\beta}{r}q_{uv}\bigg)\prod_{1\leq u<v\leq n}h_{uv}\bigg(\frac{\beta}{r}p_{uv}\bigg)\\
&\quad\cdot\prod_{1\leq u<v\leq n}h_{uv}\bigg(\frac{\beta}{r}p_{uv}\bigg)\prod_{1\leq u<v\leq n}g_{uv}\bigg(\frac{\beta}{r}q_{uv}\bigg)\prod_{1\leq u\leq n}f_u\big(e^{-\frac{\beta}{r}d_u}\big)\\
=&\prod_{1\leq u\leq n}e^{-\frac{\beta}{2r}d_uF_u}\prod_{1\leq u<v\leq n}e^{-\frac{\beta}{2r}q_{uv}\widetilde{\mathscr{G}}_{uv}(\frac{\beta}{r}q_{uv})}\prod_{1\leq u<v\leq n}e^{-\frac{\beta}{2r}p_{uv}\widetilde{\mathscr{H}}_{uv}(\frac{\beta}{r}p_{uv})}\\
&\quad\cdot\prod_{1\leq u<v\leq n}e^{-\frac{\beta}{2r}p_{uv}\widetilde{\mathscr{H}}_{uv}(\frac{\beta}{r}p_{uv})}\prod_{1\leq u<v\leq n}e^{-\frac{\beta}{2r}q_{uv}\widetilde{\mathscr{G}}_{uv}(\frac{\beta}{r}q_{uv})}\prod_{1\leq u\leq n}e^{-\frac{\beta}{2r}d_uF_u}\\
=&\ \exp\bigg(-\frac{\beta}{r}\bigg(\sum_{1\leq u<v\leq n}p_{uv}\widetilde{\mathscr{H}}_{uv}\bigg(\frac{\beta}{r}p_{uv}\bigg)+\sum_{1\leq u<v\leq n}q_{uv}\widetilde{\mathscr{G}}_{uv}\bigg(\frac{\beta}{r}q_{uv}\bigg)+\sum_{u=1}^{n}d_uF_u\bigg)\bigg)\cdot W
\end{aligned}
\end{equation}
that implements the second-order Suzuki formula, where we have applied \thm{error_type} in the last line. Since
\begin{equation}
\norm{F_{u}}\leq 2,\quad
\norm{\widetilde{\mathscr{G}}_{uv}\bigg(\frac{\beta}{r}q_{uv}\bigg)}\leq 2+2\frac{\beta}{r}q_{uv}\leq 3,\quad
\norm{\widetilde{\mathscr{H}}_{uv}\bigg(\frac{\beta}{r}p_{uv}\bigg)}\leq 2+2\frac{\beta}{r}p_{uv}\leq 3,
\end{equation}
the perturbed Hamiltonian satisfies
\begin{equation}
\sum_{1\leq u<v\leq n}p_{uv}\norm{\widetilde{\mathscr{H}}_{uv}\bigg(\frac{\beta}{r}p_{uv}\bigg)}+\sum_{1\leq u<v\leq n}q_{uv}\norm{\widetilde{\mathscr{G}}_{uv}\bigg(\frac{\beta}{r}q_{uv}\bigg)}+\sum_{u=1}^{n}|d_u|\norm{F_u}
\leq\binom{n}{2}3+\binom{n}{2}3+2n
\leq 3n^2.
\end{equation}
We also need to bound nested commutators of Hamiltonian terms with two layers of nesting. This analysis is similar to that for the transverse field Ising model; the resulting scaling is $\cO{n^4}$. By \lem{time_ordered_norm_bound}, \thm{error_type} and an analysis of the exponentiated-type error $\mathscr{E}(\tau)$, there exists a constant $c>0$ such that
\begin{equation}
\norm{W}\leq\exp\bigg(\frac{cn^4\beta^3}{r^3}e^{\frac{12n^2\beta}{r}}\bigg).
\end{equation}

To proceed, we apply \lem{interaction_picture} to switch to the interaction picture, giving
\begin{equation}
\exp\bigg(-\frac{\beta}{r}\bigg(\sum_{1\leq u<v\leq n}p_{uv}\widetilde{\mathscr{H}}_{uv}\bigg(\frac{\beta}{r}p_{uv}\bigg)+\sum_{1\leq u<v\leq n}q_{uv}\widetilde{\mathscr{G}}_{uv}\bigg(\frac{\beta}{r}q_{uv}\bigg)+\sum_{u=1}^{n}d_uF_u\bigg)\bigg)
=e^{-\frac{\beta}{r}H} V,
\end{equation}
where
\begin{equation}
\begin{aligned}
V=\exp_{\mathcal{T}}\bigg(-\int_{0}^{\frac{\beta}{r}}\mathrm{d}\tau\ e^{\tau H}\bigg(\sum_{1\leq u<v\leq n}p_{uv}\widetilde{\mathscr{H}}_{uv}\bigg(\frac{\beta}{r}p_{uv}\bigg)+\sum_{1\leq u<v\leq n}q_{uv}\widetilde{\mathscr{G}}_{uv}\bigg(\frac{\beta}{r}q_{uv}\bigg)+\sum_{u=1}^{n}d_uF_u-H\bigg)e^{-\tau H}\bigg).
\end{aligned}
\end{equation}
From \eq{perturb_term},
\begin{equation}
\begin{aligned}
&\ \norm{\sum_{1\leq u<v\leq n}p_{uv}\widetilde{\mathscr{H}}_{uv}\bigg(\frac{\beta}{r}p_{uv}\bigg)+\sum_{1\leq u<v\leq n}q_{uv}\widetilde{\mathscr{G}}_{uv}\bigg(\frac{\beta}{r}q_{uv}\bigg)+\sum_{u=1}^{n}d_uF_u-H}\\
=&\ \norm{\sum_{1\leq u<v\leq n}p_{uv}\frac{2r}{\beta p_{uv}}\mathscr{H}_{uv}\bigg(\frac{\beta}{r}p_{uv}\bigg)+\sum_{1\leq u<v\leq n}q_{uv}\frac{2r}{\beta q_{uv}}\mathscr{G}_{uv}\bigg(\frac{\beta}{r}q_{uv}\bigg)}\\
\leq&\ \binom{n}{2}2\frac{\beta}{r}p_{uv}+\binom{n}{2}2\frac{\beta}{r}q_{uv}
=2n^2\frac{\beta}{r},
\end{aligned}
\end{equation}
whereas the original Hamiltonian has spectral norm
\begin{equation}
\begin{aligned}
\norm{H}&\leq\sum_{1\leq u<v\leq n}p_{uv}\norm{{-}X_uX_v-Y_uY_v}+\sum_{1\leq u<v\leq n}q_{uv}\norm{{-}X_uX_v+Y_uY_v}+\sum_{u=1}^{n}|d_u|\norm{I+Z_u}\\
&\leq\binom{n}{2}\cdot 2+\binom{n}{2}\cdot 2+n\cdot 2=2n^2.
\end{aligned}
\end{equation}
Thus, \lem{time_ordered_norm_bound} implies
\begin{equation}
\norm{V}\leq \exp\bigg(\frac{2n^2\beta^2}{r^2}e^{\frac{4n^2\beta}{r}}\bigg).
\end{equation}

Altogether, we obtain
\begin{equation}
\begin{aligned}
&\prod_{1\leq u\leq n}f_u\big(e^{-\frac{\beta}{r}d_u}\big)\prod_{1\leq u<v\leq n}g_{uv}\bigg(\frac{\beta}{r}q_{uv}\bigg)\prod_{1\leq u<v\leq n}h_{uv}\bigg(\frac{\beta}{r}p_{uv}\bigg)\\
&\quad\cdot\prod_{1\leq u<v\leq n}h_{uv}\bigg(\frac{\beta}{r}p_{uv}\bigg)\prod_{1\leq u<v\leq n}g_{uv}\bigg(\frac{\beta}{r}q_{uv}\bigg)\prod_{1\leq u\leq n}f_u\big(e^{-\frac{\beta}{r}d_u}\big)\\
=&\ e^{-\frac{\beta}{r}H} U,
\end{aligned}
\end{equation}
where the operator $U=VW$ has spectral norm bounded by
\begin{equation}
\norm{U}=\norm{VW}\leq
\exp\bigg(\frac{2n^2\beta^2}{r^2}e^{\frac{4n^2\beta}{r}}
+\frac{cn^4\beta^3}{r^3}e^{\frac{12n^2\beta}{r}}\bigg)
\end{equation}
for some constant $c>0$. The remaining analysis continues as in \sec{app_qmc}. Similar to the case of transverse field Ising model, our Trotter error bound gives improved quantum Monte Carlo simulation of the ferromagnetic quantum spin systems.

\section{Higher-order error bounds with small prefactors}
\label{append:pf2k}

We have showed in \sec{prefactor_pf12} that our analysis reproduces known tight error bounds for first- and second-order formulas. In this section, we give heuristic strategies to derive higher-order Trotter error bounds with small prefactors. We illustrate this for the fourth-order formula, which is advantageous for simulating small-size systems \cite{CMNRS18} but does not have a tight error analysis. We further benchmark our bounds in \sec{prefactor_pf2k} by numerically simulating systems with nearest-neighbor and power-law interactions. Throughout this section, we assume $H$ is Hermitian, $t\in\R$, and consider the real-time evolution $e^{-itH}$.

We first consider a Hamiltonian $H=A+B$ consisting of two summands. This models systems with nearest-neighbor interactions where summands are grouped in an even-odd pattern \eq{even-odd}. The ideal evolution under $H$ for time $t$ is $e^{-itH}$, which we decompose using the fourth-order product formula $\mathscr{S}_4(t)$. Recall from \eq{pf2k} that $\mathscr{S}_4(t)$ is defined by
\begin{equation}
\begin{aligned}
\mathscr{S}_2(t)&:=e^{-i\frac{t}{2}A}e^{-itB}e^{-i\frac{t}{2}A},\\
\mathscr{S}_4(t)&:=\big[\mathscr{S}_2(u_2t)\big]^2\mathscr{S}_2((1-4u_2)t)\big[\mathscr{S}_2(u_2t)\big]^2,
\end{aligned}
\end{equation}
with $u_2:=1/(4-4^{1/3})$. Expanding this definition, we obtain
\begin{equation}
\mathscr{S}_4(t)=e^{-ita_6A}e^{-itb_5B}e^{-ita_5A}e^{-itb_4B}e^{-ita_4A}e^{-itb_3B}e^{-ita_3A}e^{-itb_2B}e^{-ita_2A}e^{-itb_1B}e^{-ita_1A},
\end{equation}
where
\begin{equation}
a_1:=a_6:=\frac{u_2}{2},\quad
b_1:=a_2:=b_2:=b_4:=a_5:=b_5:=u_2,\quad
a_3:=a_4:=\frac{1-3u_2}{2},\quad
b_3:=1-4u_2.
\end{equation}

Without loss of generality, we analyze the additive Trotter error of $\mathscr{S}_4(t)$. We gave an analysis in \sec{theory_type} that works for a general product formula, and we improve that here to obtain an error bound for $\mathscr{S}_4(t)$ with a small prefactor. To this end, we compute
\begin{equation}
\begin{aligned}
&\frac{\mathrm{d}}{\mathrm{d}t}\mathscr{S}_4(t)-(-iH)\mathscr{S}_4(t)\\
=&\big[e^{-ita_6A},-ib_5B\big]e^{-itb_5B}e^{-ita_5A}e^{-itb_4B}e^{-ita_4A}e^{-itb_3B}e^{-ita_3A}e^{-itb_2B}e^{-ita_2A}e^{-itb_1B}e^{-ita_1A}\\
&+\big[e^{-ita_6A}e^{-itb_5B},-ia_5A\big]e^{-ita_5A}e^{-itb_4B}e^{-ita_4A}e^{-itb_3B}e^{-ita_3A}e^{-itb_2B}e^{-ita_2A}e^{-itb_1B}e^{-ita_1A}\\
&+\cdots\\
&+\big[e^{-ita_6A}e^{-itb_5B}e^{-ita_5A}e^{-itb_4B}e^{-ita_4A}e^{-itb_3B}e^{-ita_3A}e^{-itb_2B}e^{-ita_2A},-ib_1B\big]e^{-itb_1B}e^{-ita_1A}\\
&+\big[e^{-ita_6A}e^{-itb_5B}e^{-ita_5A}e^{-itb_4B}e^{-ita_4A}e^{-itb_3B}e^{-ita_3A}e^{-itb_2B}e^{-ita_2A}e^{-itb_1B},-ia_1A\big]e^{-ita_1A}.
\end{aligned}
\end{equation}
Performing the commutation sequentially, we have
\begin{equation}
\begin{aligned}
&\frac{\mathrm{d}}{\mathrm{d}t}\mathscr{S}_4(t)-(-iH)\mathscr{S}_4(t)\\
=&\big[e^{-ita_6A},-ib_5B\big]e^{-itb_5B}e^{-ita_5A}e^{-itb_4B}e^{-ita_4A}e^{-itb_3B}e^{-ita_3A}e^{-itb_2B}e^{-ita_2A}e^{-itb_1B}e^{-ita_1A}\\
&+e^{-ita_6A}\big[e^{-itb_5B},-ia_5A\big]e^{-ita_5A}e^{-itb_4B}e^{-ita_4A}e^{-itb_3B}e^{-ita_3A}e^{-itb_2B}e^{-ita_2A}e^{-itb_1B}e^{-ita_1A}\\
&+\cdots\\
&+e^{-ita_6A}e^{-itb_5B}e^{-ita_5A}e^{-itb_4B}e^{-ita_4A}e^{-itb_3B}e^{-ita_3A}e^{-itb_2B}\big[e^{-ita_2A},-ib_1B\big]e^{-itb_1B}e^{-ita_1A}\\
&+e^{-ita_6A}e^{-itb_5B}e^{-ita_5A}e^{-itb_4B}e^{-ita_4A}e^{-itb_3B}\big[e^{-ita_3A},-ib_1B\big]e^{-itb_2B}e^{-ita_2A}e^{-itb_1B}e^{-ita_1A}\\
&+e^{-ita_6A}e^{-itb_5B}e^{-ita_5A}e^{-itb_4B}\big[e^{-ita_4A},-ib_1B\big]e^{-itb_3B}e^{-ita_3A}e^{-itb_2B}e^{-ita_2A}e^{-itb_1B}e^{-ita_1A}\\
&+e^{-ita_6A}e^{-itb_5B}\big[e^{-ita_5A},-ib_1B\big]e^{-itb_4B}e^{-ita_4A}e^{-itb_3B}e^{-ita_3A}e^{-itb_2B}e^{-ita_2A}e^{-itb_1B}e^{-ita_1A}\\
&+\big[e^{-ita_6A},-ib_1B\big]e^{-itb_5B}e^{-ita_5A}e^{-itb_4B}e^{-ita_4A}e^{-itb_3B}e^{-ita_3A}e^{-itb_2B}e^{-ita_2A}e^{-itb_1B}e^{-ita_1A}\\
&+e^{-ita_6A}e^{-itb_5B}e^{-ita_5A}e^{-itb_4B}e^{-ita_4A}e^{-itb_3B}e^{-ita_3A}e^{-itb_2B}e^{-ita_2A}\big[e^{-itb_1B},-ia_1A\big]e^{-ita_1A}\\
&+e^{-ita_6A}e^{-itb_5B}e^{-ita_5A}e^{-itb_4B}e^{-ita_4A}e^{-itb_3B}e^{-ita_3A}\big[e^{-itb_2B},-ia_1A\big]e^{-ita_2A}e^{-itb_1B}e^{-ita_1A}\\
&+e^{-ita_6A}e^{-itb_5B}e^{-ita_5A}e^{-itb_4B}e^{-ita_4A}\big[e^{-itb_3B},-ia_1A\big]e^{-ita_3A}e^{-itb_2B}e^{-ita_2A}e^{-itb_1B}e^{-ita_1A}\\
&+e^{-ita_6A}e^{-itb_5B}e^{-ita_5A}\big[e^{-itb_4B},-ia_1A\big]e^{-ita_4A}e^{-itb_3B}e^{-ita_3A}e^{-itb_2B}e^{-ita_2A}e^{-itb_1B}e^{-ita_1A}\\
&+e^{-ita_6A}\big[e^{-itb_5B},-ia_1A\big]e^{-ita_5A}e^{-itb_4B}e^{-ita_4A}e^{-itb_3B}e^{-ita_3A}e^{-itb_2B}e^{-ita_2A}e^{-itb_1B}e^{-ita_1A}.
\end{aligned}
\end{equation}
We further define
\begin{equation}
\begin{aligned}
c_1&:=a_1,\qquad\qquad\qquad\qquad\qquad &d_1&:=b_1,\\
c_2&:=a_1+a_2, &d_2&:=b_1+b_2,\\
c_3&:=a_1+a_2+a_3, &d_3&:=b_1+b_2+b_3,\\
c_4&:=a_1+a_2+a_3+a_4, &d_4&:=b_1+b_2+b_3+b_4,\\
c_5&:=a_1+a_2+a_3+a_4+a_5, &d_5&:=b_1+b_2+b_3+b_4+b_5,
\end{aligned}
\end{equation}
so that
\begin{equation}
\begin{aligned}
&\frac{\mathrm{d}}{\mathrm{d}t}\mathscr{S}_4(t)-(-iH)\mathscr{S}_4(t)\\
=&\big[e^{-ita_6A},-id_5B\big]e^{-itb_5B}e^{-ita_5A}e^{-itb_4B}e^{-ita_4A}e^{-itb_3B}e^{-ita_3A}e^{-itb_2B}e^{-ita_2A}e^{-itb_1B}e^{-ita_1A}\\
&+e^{-ita_6A}\big[e^{-itb_5B},-ic_5A\big]e^{-ita_5A}e^{-itb_4B}e^{-ita_4A}e^{-itb_3B}e^{-ita_3A}e^{-itb_2B}e^{-ita_2A}e^{-itb_1B}e^{-ita_1A}\\
&+\cdots\\
&+e^{-ita_6A}e^{-itb_5B}e^{-ita_5A}e^{-itb_4B}e^{-ita_4A}e^{-itb_3B}e^{-ita_3A}e^{-itb_2B}\big[e^{-ita_2A},-id_1B\big]e^{-itb_1B}e^{-ita_1A}\\
&+e^{-ita_6A}e^{-itb_5B}e^{-ita_5A}e^{-itb_4B}e^{-ita_4A}e^{-itb_3B}e^{-ita_3A}e^{-itb_2B}e^{-ita_2A}\big[e^{-itb_1B},-ic_1A\big]e^{-ita_1A}.
\end{aligned}
\end{equation}

In \sec{theory_type} and \append{type}, we factor out the operator-valued function $\mathscr{S}_4(t)$ from the left-hand side of the above equation as
\begin{equation}
\frac{\mathrm{d}}{\mathrm{d}t}\mathscr{S}_4(t)-(-iH)\mathscr{S}_4(t)=\mathscr{S}_4(t)\mathscr{T}(t).
\end{equation}
This approach suffices to establish the asymptotic bound in \thm{trotter_error_comm_scaling} and \cor{trotter_number_comm_scaling}. However, the resulting function $\mathscr{T}(t)$ contains unitary conjugations with a large number of conjugating layers, which defeats the goal of establishing tight error bounds. We improve this by simultaneously factoring out $\mathscr{S}_{4,\text{left}}(t)$ from the left-hand side of the equation and $\mathscr{S}_{4,\text{right}}(t)$ from the right-hand side, obtaining
\begin{equation}
\frac{\mathrm{d}}{\mathrm{d}t}\mathscr{S}_4(t)-(-iH)\mathscr{S}_4(t)=\mathscr{S}_{4,\text{left}}(t)\mathscr{T}_4(t)\mathscr{S}_{4,\text{right}}(t),
\end{equation}
where
\begin{equation}
\begin{aligned}
\mathscr{S}_{\text{left}}(t)&:=e^{-ita_6A}e^{-itb_5B}e^{-ita_5A}e^{-itb_4B}e^{-ita_4A},\\
\mathscr{S}_{\text{right}}(t)&:=e^{-itb_3B}e^{-ita_3A}e^{-itb_2B}e^{-ita_2A}e^{-itb_1B}e^{-ita_1A}.
\end{aligned}
\end{equation}
It then remains to analyze $\mathscr{T}_4(t)$.

To this end, we use the fact that
\begin{equation}
\begin{aligned}
\big[e^{tX},Y\big]&=e^{tX}\int_{0}^{t}\mathrm{d}\tau\ e^{-\tau X}\big[X,Y\big]e^{\tau X}\\
&=\int_{0}^{t}\mathrm{d}\tau\ e^{\tau X}\big[X,Y\big]e^{-\tau X}e^{tX},
\end{aligned}
\end{equation}
for any $t\in \R$ and operators $X$, $Y$. We then have from \lem{td_Duhamel} that
\begin{equation}
\mathscr{S}_4(t)=e^{-itH}+\int_{0}^{t}\mathrm{d}\tau_1\ e^{-i(t-\tau_1)H}\mathscr{S}_{4,\text{left}}(\tau_1)\mathscr{T}_4(\tau_1)\mathscr{S}_{4,\text{right}}(\tau_1),
\end{equation}
where
\begin{equation}
\begin{aligned}
&\mathscr{T}_4(\tau_1)\\
=&\int_{0}^{\tau_1}\mathrm{d}\tau_2\ e^{i\tau_1a_4A}e^{i\tau_1b_4B}e^{i\tau_1a_5A}e^{i\tau_1b_5B}e^{i\tau_2 a_6A}
\big[-ia_6A,-id_5B\big]e^{-i\tau_2 a_6A}e^{-i\tau_1b_5B}e^{-i\tau_1a_5A}e^{-i\tau_1b_4B}e^{-i\tau_1a_4A}\\
&+\int_{0}^{\tau_1}\mathrm{d}\tau_2\ e^{i\tau_1a_4A}e^{i\tau_1b_4B}e^{i\tau_1a_5A}e^{i\tau_2 b_5B}
\big[-ib_5B,-ic_5A\big]e^{-i\tau_2 b_5B}e^{-i\tau_1a_5A}e^{-i\tau_1b_4B}e^{-i\tau_1a_4A}\\
&+\int_{0}^{\tau_1}\mathrm{d}\tau_2\ e^{i\tau_1a_4A}e^{i\tau_1b_4B}e^{i\tau_2 a_5A}
\big[-i a_5A,-id_4B\big]e^{-i\tau_2 a_5A}e^{-i\tau_1b_4B}e^{-i\tau_1a_4A}\\
&+\int_{0}^{\tau_1}\mathrm{d}\tau_2\ e^{i\tau_1a_4A}e^{i\tau_2 b_4B}
\big[-ib_4B,-ic_4A\big]e^{-i\tau_2 b_4B}e^{-i\tau_1a_4A}\\
&+\int_{0}^{\tau_1}\mathrm{d}\tau_2\ e^{i\tau_2 a_4A}
\big[-ia_4A,-id_3B\big]e^{-i\tau_2 a_4A}\\
&+\int_{0}^{\tau_1}\mathrm{d}\tau_2\ e^{-i\tau_2 b_3B}
\big[-ib_3B,-ic_3A\big]e^{i\tau_2 b_3B}\\
&+\int_{0}^{\tau_1}\mathrm{d}\tau_2\ e^{-i\tau_1b_3B}e^{-i\tau_2 a_3A}
\big[-ia_3A,-id_2B\big]e^{i\tau_2 a_3A}e^{i\tau_1b_3B}\\
&+\int_{0}^{\tau_1}\mathrm{d}\tau_2\ e^{-i\tau_1b_3B}e^{-i\tau_1a_3A}e^{-i\tau_2 b_2B}
\big[-ib_2B,-ic_2A\big]e^{i\tau_2 b_2B}e^{i\tau_1a_3A}e^{i\tau_1b_3B}\\
&+\int_{0}^{\tau_1}\mathrm{d}\tau_2\ e^{-i\tau_1b_3B}e^{-i\tau_1a_3A}e^{-i\tau_1b_2B}e^{-i\tau_2 a_2A}
\big[-ia_2A,-id_1B\big]e^{i\tau_2 a_2A}e^{i\tau_1b_2B}e^{i\tau_1a_3A}e^{i\tau_1b_3B}\\
&+\int_{0}^{\tau_1}\mathrm{d}\tau_2\ e^{-i\tau_1b_3B}e^{-i\tau_1a_3A}e^{-i\tau_1b_2B}e^{-i\tau_1a_2A}e^{-i\tau_2 b_1B}
\big[-ib_1B,-ic_1A\big]e^{i\tau_2 b_1B}e^{i\tau_1a_2A}e^{i\tau_1b_2B}e^{i\tau_1a_3A}e^{i\tau_1b_3B}.
\end{aligned}
\end{equation}

The operator-valued function $\mathscr{T}_4(\tau_1)$ has the order condition $\mathscr{T}_4(\tau_1)=O(\tau_1^4)$, which follows from \prop{order_cond_rule} and the fact that $\mathscr{S}_4(t)=e^{-itH}+O(t^5)$. For terms in $\mathscr{T}_4(\tau_1)$, we compute the Taylor expansion of each layer of unitary conjugation as in \sec{theory_rep}. In light of \lem{monomial}, we expand the time variables $\tau_1$ and $\tau_2$ to third order, as there already exists the double integral $\int_{0}^{t}\mathrm{d}\tau\int_{0}^{\tau_1}\mathrm{d}\tau_2$. We then apply the triangle inequality to bound the spectral norm of a linear combination of nested commutators of $A$ and $B$ with four nesting layers. Since $[A,A]=[B,B]=0$ and $[A,B]=[B,A]$, the bound only contains $2^5/4=8$ nonzero terms. Altogether, we obtain
\begin{equation}
\begin{aligned}
\norm{\mathscr{S}_4(t)-e^{-itH}}
&\leq t^5\Big(0.0047\norm{\big[A,\big[A,\big[A,\big[B,A\big]\big]\big]\big]}+0.0057\norm{\big[A,\big[A,\big[B,\big[B,A\big]\big]\big]\big]}\\
&\qquad\ +0.0046\norm{\big[A,\big[B,\big[A,\big[B,A\big]\big]\big]\big]}+0.0074\norm{\big[A,\big[B,\big[B,\big[B,A\big]\big]\big]\big]}\\
&\qquad\ +0.0097\norm{\big[B,\big[A,\big[A,\big[B,A\big]\big]\big]\big]}+0.0097\norm{\big[B,\big[A,\big[B,\big[B,A\big]\big]\big]\big]}\\
&\qquad\ +0.0173\norm{\big[B,\big[B,\big[A,\big[B,A\big]\big]\big]\big]}+0.0284\norm{\big[B,\big[B,\big[B,\big[B,A\big]\big]\big]\big]}\Big),
\end{aligned}
\end{equation}
assuming $t\geq 0$.

\begin{proposition}[Trotter error bound for the fourth-order Suzuki formula with two summands]
	\label{prop:pf4_bound_2term}
	Let $H=A+B$ be a Hamiltonian consisting of two summands and $t\geq 0$. Let $\mathscr{S}_4(t)$ be the fourth-order Suzuki formula \eq{pf2k}. Then,
	\begin{equation}
	\begin{aligned}
	\norm{\mathscr{S}_4(t)-e^{-itH}}
	&\leq t^5\Big(0.0047\norm{\big[A,\big[A,\big[A,\big[B,A\big]\big]\big]\big]}+0.0057\norm{\big[A,\big[A,\big[B,\big[B,A\big]\big]\big]\big]}\\
	&\qquad\ +0.0046\norm{\big[A,\big[B,\big[A,\big[B,A\big]\big]\big]\big]}+0.0074\norm{\big[A,\big[B,\big[B,\big[B,A\big]\big]\big]\big]}\\
	&\qquad\ +0.0097\norm{\big[B,\big[A,\big[A,\big[B,A\big]\big]\big]\big]}+0.0097\norm{\big[B,\big[A,\big[B,\big[B,A\big]\big]\big]\big]}\\
	&\qquad\ +0.0173\norm{\big[B,\big[B,\big[A,\big[B,A\big]\big]\big]\big]}+0.0284\norm{\big[B,\big[B,\big[B,\big[B,A\big]\big]\big]\big]}\Big).
	\end{aligned}
	\end{equation}
\end{proposition}

A generalization of this approach analyzes Hamiltonians with three summands, which is relevant for certain nearest-neighbor and power-law systems where terms are ordered in an \textit{X-Y-Z} pattern \eq{x-y-z} and \eq{x-y-z-power-law}.
\begin{proposition}[Trotter error bound for the fourth-order Suzuki formula with three summands]
	\label{prop:pf4_bound_3term}
	Let $H=H_1+H_2+H_3$ be a Hamiltonian consisting of three summands and $t\geq 0$. Let $\mathscr{S}_4(t)$ be the fourth-order Suzuki formula \eq{pf2k}. Then,
	\begin{equation}
	\label{eq:pf4_bound_3term_coeff}
	\begin{aligned}
	\norm{\mathscr{S}_4(t)-e^{-itH}}
	&\leq t^5\sum_{i,j,k,l,m=1}^{3}c_{i,j,k,l,m}\norm{\big[H_i,\big[H_j,\big[H_k,\big[H_l,H_m\big]\big]\big]\big]},
	\end{aligned}
	\end{equation}
	where the coefficients $c_{i,j,k,l,m}$ are given by \tab{pf4-3}.
\end{proposition}

Unlike the first- and second-order cases, we do not have a rigorous proof of the tightness of these bounds. However, our numerical result suggests that these bounds are close to tight for one-dimensional Heisenberg models with nearest-neighbor \eq{heisenberg} and power-law \eq{heisenberg-power-law} interactions. We hope future work will shed light on the tightness of our analysis through either theoretical justification or numerical calculation.

\begin{table}[!htbp]
	\begin{center}
		\renewcommand{\arraystretch}{1.15}
		\small
		\scalebox{0.85}{\parbox{1.5\linewidth}{%
		\begin{tabular}{l|c|l|c|l|c}
			Commutator & Coefficient & Commutator & Coefficient & Commutator & Coefficient \\ \hline
			$\norm{[H_1,[H_1,[H_1,[H_2,H_1]]]]}$ & $0.0047$ & $\norm{[H_1,[H_1,[H_1,[H_3,H_1]]]]}$ & $0.0047$ & $\norm{[H_1,[H_1,[H_1,[H_3,H_2]]]]}$ & $0.0043$ \\
			$\norm{[H_1,[H_1,[H_2,[H_2,H_1]]]]}$ & $0.0057$ & $\norm{[H_1,[H_1,[H_2,[H_3,H_1]]]]}$ & $0.0057$ & $\norm{[H_1,[H_1,[H_2,[H_3,H_2]]]]}$ & $0.0057$ \\
			$\norm{[H_1,[H_1,[H_3,[H_2,H_1]]]]}$ & $0.0057$ & $\norm{[H_1,[H_1,[H_3,[H_3,H_1]]]]}$ & $0.0057$ & $\norm{[H_1,[H_1,[H_3,[H_3,H_2]]]]}$ & $0.0057$ \\
			$\norm{[H_1,[H_2,[H_1,[H_2,H_1]]]]}$ & $0.0046$ & $\norm{[H_1,[H_2,[H_1,[H_3,H_1]]]]}$ & $0.0046$ & $\norm{[H_1,[H_2,[H_1,[H_3,H_2]]]]}$ & $0.0035$ \\
			$\norm{[H_1,[H_2,[H_2,[H_2,H_1]]]]}$ & $0.0074$ & $\norm{[H_1,[H_2,[H_2,[H_3,H_1]]]]}$ & $0.0070$ & $\norm{[H_1,[H_2,[H_2,[H_3,H_2]]]]}$ & $0.0062$ \\
			$\norm{[H_1,[H_2,[H_3,[H_2,H_1]]]]}$ & $0.0082$ & $\norm{[H_1,[H_2,[H_3,[H_3,H_1]]]]}$ & $0.0082$ & $\norm{[H_1,[H_2,[H_3,[H_3,H_2]]]]}$ & $0.0082$ \\
			$\norm{[H_1,[H_3,[H_1,[H_2,H_1]]]]}$ & $0.0046$ & $\norm{[H_1,[H_3,[H_1,[H_3,H_1]]]]}$ & $0.0046$ & $\norm{[H_1,[H_3,[H_1,[H_3,H_2]]]]}$ & $0.0035$ \\
			$\norm{[H_1,[H_3,[H_2,[H_2,H_1]]]]}$ & $0.0070$ & $\norm{[H_1,[H_3,[H_2,[H_3,H_1]]]]}$ & $0.0058$ & $\norm{[H_1,[H_3,[H_2,[H_3,H_2]]]]}$ & $0.0046$ \\
			$\norm{[H_1,[H_3,[H_3,[H_2,H_1]]]]}$ & $0.0082$ & $\norm{[H_1,[H_3,[H_3,[H_3,H_1]]]]}$ & $0.0074$ & $\norm{[H_1,[H_3,[H_3,[H_3,H_2]]]]}$ & $0.0074$ \\
			$\norm{[H_2,[H_1,[H_1,[H_2,H_1]]]]}$ & $0.0150$ & $\norm{[H_2,[H_1,[H_1,[H_3,H_1]]]]}$ & $0.0150$ & $\norm{[H_2,[H_1,[H_1,[H_3,H_2]]]]}$ & $0.0141$ \\
			$\norm{[H_2,[H_1,[H_2,[H_2,H_1]]]]}$ & $0.0161$ & $\norm{[H_2,[H_1,[H_2,[H_3,H_1]]]]}$ & $0.0161$ & $\norm{[H_2,[H_1,[H_2,[H_3,H_2]]]]}$ & $0.0161$ \\
			$\norm{[H_2,[H_1,[H_3,[H_2,H_1]]]]}$ & $0.0161$ & $\norm{[H_2,[H_1,[H_3,[H_3,H_1]]]]}$ & $0.0161$ & $\norm{[H_2,[H_1,[H_3,[H_3,H_2]]]]}$ & $0.0161$ \\
			$\norm{[H_2,[H_2,[H_1,[H_2,H_1]]]]}$ & $0.0239$ & $\norm{[H_2,[H_2,[H_1,[H_3,H_1]]]]}$ & $0.0239$ & $\norm{[H_2,[H_2,[H_1,[H_3,H_2]]]]}$ & $0.0212$ \\
			$\norm{[H_2,[H_2,[H_2,[H_2,H_1]]]]}$ & $0.0315$ & $\norm{[H_2,[H_2,[H_2,[H_3,H_1]]]]}$ & $0.0306$ & $\norm{[H_2,[H_2,[H_2,[H_3,H_2]]]]}$ & $0.0290$ \\
			$\norm{[H_2,[H_2,[H_3,[H_2,H_1]]]]}$ & $0.0303$ & $\norm{[H_2,[H_2,[H_3,[H_3,H_1]]]]}$ & $0.0303$ & $\norm{[H_2,[H_2,[H_3,[H_3,H_2]]]]}$ & $0.0303$ \\
			$\norm{[H_2,[H_3,[H_1,[H_2,H_1]]]]}$ & $0.0179$ & $\norm{[H_2,[H_3,[H_1,[H_3,H_1]]]]}$ & $0.0179$ & $\norm{[H_2,[H_3,[H_1,[H_3,H_2]]]]}$ & $0.0153$ \\
			$\norm{[H_2,[H_3,[H_2,[H_2,H_1]]]]}$ & $0.0232$ & $\norm{[H_2,[H_3,[H_2,[H_3,H_1]]]]}$ & $0.0206$ & $\norm{[H_2,[H_3,[H_2,[H_3,H_2]]]]}$ & $0.0179$ \\
			$\norm{[H_2,[H_3,[H_3,[H_2,H_1]]]]}$ & $0.0259$ & $\norm{[H_2,[H_3,[H_3,[H_3,H_1]]]]}$ & $0.0241$ & $\norm{[H_2,[H_3,[H_3,[H_3,H_2]]]]}$ & $0.0241$ \\
			$\norm{[H_3,[H_1,[H_1,[H_2,H_1]]]]}$ & $0.0204$ & $\norm{[H_3,[H_1,[H_1,[H_3,H_1]]]]}$ & $0.0204$ & $\norm{[H_3,[H_1,[H_1,[H_3,H_2]]]]}$ & $0.0186$ \\
			$\norm{[H_3,[H_1,[H_2,[H_2,H_1]]]]}$ & $0.0225$ & $\norm{[H_3,[H_1,[H_2,[H_3,H_1]]]]}$ & $0.0225$ & $\norm{[H_3,[H_1,[H_2,[H_3,H_2]]]]}$ & $0.0217$ \\
			$\norm{[H_3,[H_1,[H_3,[H_2,H_1]]]]}$ & $0.0225$ & $\norm{[H_3,[H_1,[H_3,[H_3,H_1]]]]}$ & $0.0225$ & $\norm{[H_3,[H_1,[H_3,[H_3,H_2]]]]}$ & $0.0225$ \\
			$\norm{[H_3,[H_2,[H_1,[H_2,H_1]]]]}$ & $0.0423$ & $\norm{[H_3,[H_2,[H_1,[H_3,H_1]]]]}$ & $0.0423$ & $\norm{[H_3,[H_2,[H_1,[H_3,H_2]]]]}$ & $0.0377$ \\
			$\norm{[H_3,[H_2,[H_2,[H_2,H_1]]]]}$ & $0.0585$ & $\norm{[H_3,[H_2,[H_2,[H_3,H_1]]]]}$ & $0.0571$ & $\norm{[H_3,[H_2,[H_2,[H_3,H_2]]]]}$ & $0.0537$ \\
			$\norm{[H_3,[H_2,[H_3,[H_2,H_1]]]]}$ & $0.0502$ & $\norm{[H_3,[H_2,[H_3,[H_3,H_1]]]]}$ & $0.0502$ & $\norm{[H_3,[H_2,[H_3,[H_3,H_2]]]]}$ & $0.0502$ \\
			$\norm{[H_3,[H_3,[H_1,[H_2,H_1]]]]}$ & $0.0423$ & $\norm{[H_3,[H_3,[H_1,[H_3,H_1]]]]}$ & $0.0423$ & $\norm{[H_3,[H_3,[H_1,[H_3,H_2]]]]}$ & $0.0377$ \\
			$\norm{[H_3,[H_3,[H_2,[H_2,H_1]]]]}$ & $0.0681$ & $\norm{[H_3,[H_3,[H_2,[H_3,H_1]]]]}$ & $0.0641$ & $\norm{[H_3,[H_3,[H_2,[H_3,H_2]]]]}$ & $0.0601$ \\
			$\norm{[H_3,[H_3,[H_3,[H_2,H_1]]]]}$ & $0.0648$ & $\norm{[H_3,[H_3,[H_3,[H_3,H_1]]]]}$ & $0.0621$ & $\norm{[H_3,[H_3,[H_3,[H_3,H_2]]]]}$ & $0.0628$ \\
		\end{tabular}
		}}
		\renewcommand{\arraystretch}{1}
	\end{center}
	\caption{Coefficients of the fourth-order Trotter error bound \eq{pf4_bound_3term_coeff} for Hamiltonians with three summands.
		\label{tab:pf4-3}}
\end{table}

\clearpage
\bibliographystyle{myhamsplain2}
\bibliography{TrotterErrorTheory}

\end{document}